\PassOptionsToPackage{pagebackref}{hyperref}
\PassOptionsToPackage{dvipsnames}{xcolor}
\documentclass[format=acmsmall, review=false, screen, nonacm]{acmart}
\usepackage{booktabs} \usepackage{thm-restate}
\usepackage{hyperref}

\usepackage[nameinlink]{cleveref}
\usepackage{crossreftools}
\crefname{ineq}{Inequality}{Inequality}
\creflabelformat{ineq}{#2{\upshape(#1)}#3}
\crefname{ineqs}{Inequalities}{Inequalities}
\creflabelformat{ineqs}{#2{\upshape(#1)}#3} 

\usepackage{subcaption}
\usepackage[ruled]{algorithm2e}

\SetAlFnt{\small}
\SetAlCapFnt{\small}
\SetAlCapNameFnt{\small}
\SetAlCapHSkip{0pt}
\IncMargin{-\parindent}

\usepackage{wrapfig}
\usepackage{etoolbox}
\makeatletter
\patchcmd\WF@putfigmaybe{\lower\intextsep}{}{}{\fail}\AddToHook{env/wrapfigure/begin}{\setlength{\intextsep}{0pt}}
\makeatother

\usepackage{tikz}
\usetikzlibrary{backgrounds,matrix}
\usetikzlibrary{shapes.geometric}
\usepackage{pgfplots}

\usepackage{nicefrac}

\newcommand{\ssucc}{{\vartriangleright}}
\newcommand{\mchoosetwo}{{\textstyle \binom{m}{2}}}
\renewcommand{\epsilon}{\varepsilon}
\renewcommand{\le}{\leqslant}
\renewcommand{\leq}{\leqslant}
\renewcommand{\ge}{\geqslant}
\renewcommand{\geq}{\geqslant}
\newcommand{\K}{\mathrm{K}}
\newcommand{\sqK}{\mathrm{SqK}}
\DeclareMathOperator{\supp}{supp}
\DeclareMathOperator{\round}{round}
\DeclareMathOperator{\swap}{swap}
\DeclareMathOperator{\weight}{w}

\newcommand{\moveuprank}[1]{\!\!\!\!\!\!($\bigtriangleup${\tiny$#1$})}
\newcommand{\movedownrank}[1]{\!\!\!\!\!\!($\bigtriangledown${\tiny$#1$})}
\newcommand{\staysamerank}{\!\!\!\!\!\!({\scriptsize$-$})}

\setcitestyle{authoryear}

\title{The Squared Kemeny Rule for Averaging Rankings}

\author{Patrick Lederer}
\affiliation{
	\institution{Technical University of Munich}
	\country{Germany}
}
\email{ledererp@in.tum.de}

\author{Dominik Peters}
\affiliation{
	\institution{CNRS, LAMSADE, Universit\'e Paris Dauphine - PSL}
	\country{France}
}
\email{dominik.peters@lamsade.dauphine.fr}

\author{Tomasz W\k{a}s}
\affiliation{
	\institution{CNRS, LAMSADE, Universit\'e Paris Dauphine - PSL}
	\country{France}
}
\email{tomasz.was@dauphine.psl.eu}

\begin{abstract}
{\large Manuscript: April 2024}
	
\bigskip
\noindent
For the problem of aggregating several rankings into one ranking, \citet{Kem-1959-Kemeny} proposed two methods: the \emph{median rule} which selects the ranking with the smallest total swap distance to the input rankings, and the \emph{mean rule} which minimizes the \emph{squared} swap distances to the input rankings. The median rule has been extensively studied since and is now known simply as \emph{Kemeny's rule}. It exhibits majoritarian properties, so for example if more than half of the input rankings are the same, then the output of the rule is the same ranking. 

We observe that this behavior is undesirable in many rank aggregation settings. For example, when we rank objects by different criteria (quality, price, etc.) and want to aggregate them with specified weights for the criteria, then a criterion with weight 51\% should have 51\% influence on the output instead of 100\%. We show that the Squared Kemeny rule (i.e., the mean rule) behaves this way, by establishing a bound on the distance of the output ranking to any input rankings, as a function of their weights. Furthermore, we give an axiomatic characterization of the Squared Kemeny rule, which mirrors the existing characterization of the Kemeny rule but replaces the majoritarian Condorcet axiom by a proportionality axiom. Finally, we discuss the computation of the rule and show its behavior in a simulation study.
\end{abstract}

\begin{document}

\begin{titlepage}
	\maketitle
	\vspace{1cm}
	\hrule
	\vspace{10pt}
	\tableofcontents
	\vspace{-15pt}
	\hrule
\end{titlepage}

\section{Introduction}
\begin{wrapfigure}{r}{0.23\textwidth}
	\includegraphics[width=\linewidth]{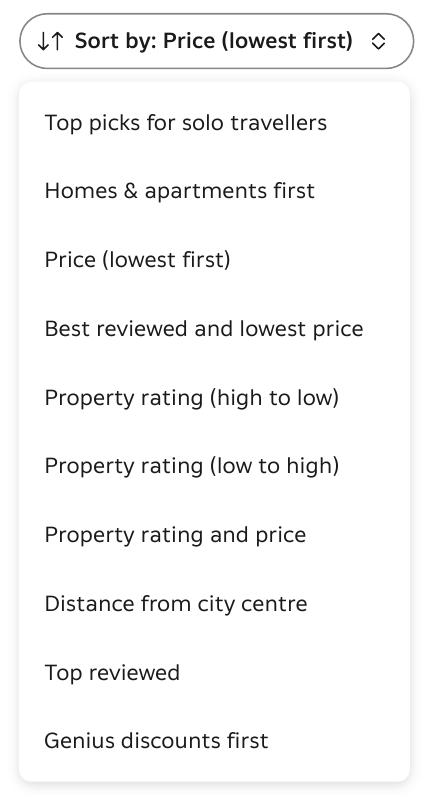}
	\vspace{-21pt}
	\caption{Sort options on booking.com.}
	\label{fig:booking-sor}
\end{wrapfigure}
Many search engines allow users to sort results by several criteria. For example, websites such as booking.com and expedia.com allow users to sort hotels by their price, their review score, or their distance to the city center. They also offer sorting by combinations of these criteria (for example ``best reviewed and lowest price'', see \Cref{fig:booking-sor}). More generally, they could allow users to specify weights over the different ways to sort alternatives (e.g. 60\% price, 30\% reviews, 10\% location) and then present an aggregated ranking of hotels.

The task of combining several rankings (possibly with weights) into one ranking is known as \emph{rank aggregation}. The best-known and most frequently discussed rule for this problem is \emph{Kemeny's rule} \citep{Kem-1959-Kemeny,KemSne-1960-Kemeny}, which minimizes the total distance to all input rankings. To be precise, let $A$ be a set of $m$ alternatives (e.g., hotels), and let $\mathcal R$ be the set of rankings (linear orders) on $A$. For two rankings ${\succ}, {\rhd} \in \mathcal R$, the \emph{swap distance} (or Kendall-tau distance) between them is the number of pairs of alternatives on which they disagree:
$\swap({\succ}, \rhd) = |\{  \{x,y\} \subseteq A : x \succ y \text{ and } y \rhd x  \}|$. 
A \emph{profile} $R$ is a function that assigns to each linear order $\succ$ a weight $R(\succ) \ge 0$, with weights summing to $1$.
To aggregate the rankings in a profile into a collective ranking, we use a \emph{social preference function} (\emph{SPF}) $f$, which selects for each profile $R$ a set $f(R) \subseteq \mathcal{R}$ of output rankings (ideally just one ranking, but there may be ties). Finally, Kemeny's rule is the SPF that selects the rankings which minimize the average swap distance to the input:
\[
\textstyle
\text{Kemeny}(R) = \arg\min_{{\rhd} \in \mathcal{R}} \sum_{{\succ} \in \mathcal{R}} R(\succ) \cdot \swap(\succ, \rhd).
\]
Kemeny's rule is an attractive SPF for several reasons: it is the maximum likelihood estimator (MLE) of the Mallow's $\phi$ model \citep{young1988condorcet} which makes it a good fit for epistemic problems where we wish to discover a ground truth ranking given noisy estimates. Further, the rule is axiomatically characterized by a Condorcet-style axiom \citep{YouLev-1978-KemenyAxioms}; in particular, when the weight of some ranking exceeds 50\%, then the output of Kemeny's rule will be that ranking.

While this latter property is desirable in epistemic and electoral settings, for a hotel booking website it is disqualifying. In the above example of a user wishing to sort hotels by 60\% price, 30\% reviews, 10\% location, the Kemeny output would just be the ranking by price, with review scores and location having no influence. Instead, what the user wants is a ranking where a hotel can compensate a lower position in the price ranking by a high position in the reviews and location rankings. Thus, we need a rule that faithfully follows the desired weighting, and makes use of all the information instead of ignoring low-weight criteria.

As it turns out, exactly such a rule was proposed under the name ``mean rule'' by \citet{Kem-1959-Kemeny} and \citet{KemSne-1960-Kemeny} in the same articles that introduced Kemeny's rule. Our name for this SPF is the \emph{Squared Kemeny rule} as it is obtained by squaring the distances in the objective function:
\[
\textstyle
\text{SqK}(R) = \arg\min_{{\rhd} \in \mathcal{R}} \sum_{{\succ} \in \mathcal{R}} R(\succ) \cdot \swap(\succ, \rhd)^2.
\]
The Squared Kemeny rule appears to have been almost entirely ignored in the subsequent literature. \citet[p.~290]{YouLev-1978-KemenyAxioms} quickly dismiss it, writing that ``Kemeny left the problem of which solution to choose unresolved. But from the standpoint of collective decision-making there is ample reason to prefer the median, since it turns out that the median consensus leads to a Condorcet method, while the mean does not.'' We believe that this dismissal was too quick.

The contribution of this paper is showing that the Squared Kemeny rule is well-suited to perform rank aggregation when we wish each input ranking to be reflected in the output ranking, to an extent that is proportional to its weight.\footnote{Proportionality can be seen as a fairness notion with respect to voters (who have a guaranteed amount of influence on the output ranking). The rank aggregation literature has studied distinct notions of fairness for candidates that come labelled as belonging to protected groups \citep{2020fairconsensus,chakraborty2022fair,wei2022rank}}
Such proportionality notions have recently attracted significant attention in voting, in particular in the settings of multi-winner voting and participatory budgeting \citep[see, e.g.,][]{aziz2017jr,lackner2023multi,peters2021proportional}.
Following the approach in that literature, we could formalize proportionality, for example, by saying that a ranking that makes up $\alpha\%$ of the weight should agree with the output ranking on at least roughly $\alpha\%$ pairwise comparisons. While Squared Kemeny does not satisfy this in general, we will see that it does in important special cases, and that it satisfies an approximate version in general. More generally, we will show that Squared Kemeny behaves more like an average and thus is more responsive to changes in its input than the Kemeny rule, which behaves more like a median.

\subsection*{Rank Aggregation With Two Criteria}

\newcommand{\hotelbubble}[2]{\tikz[transform shape, scale=0.88]{\node [fill=#1, rounded corners, baseline, text height=7pt, inner ysep=1pt, inner xsep=0pt, minimum height=7pt, minimum width=39pt, font=\footnotesize] {\vphantom{Q}#2};
}}

\definecolor{GraduateColor}{RGB}{230, 207, 154}
\definecolor{StudyColor}{RGB}{225, 191, 227}
\definecolor{OmniColor}{RGB}{159, 180, 219} 

\newcommand{\LaQuinta}{\hotelbubble{SpringGreen}{La Quinta}}
\newcommand{\NHH}{\hotelbubble{Melon}{NHH}}
\newcommand{\Graduate}{\hotelbubble{GraduateColor}{Graduate}}
\newcommand{\Omni}{\hotelbubble{OmniColor}{Omni}}
\newcommand{\HotelMarcel}{\hotelbubble{Apricot}{H.\,Marcel}}
\newcommand{\TheStudy}{\hotelbubble{StudyColor}{The Study}}

\begin{figure}
	\centering
	\makebox[\textwidth][c]{
	\setlength{\tabcolsep}{1.2pt}
	\renewcommand{\arraystretch}{0.9}
	\footnotesize
	\begin{tabular}{ccccccccccc}
		\toprule
		Price & 90\% & 80\% & 70\% & 60\% & 50\% & 40\% & 30\% & 20\% & 10\% & Score \\
		\midrule
		\LaQuinta & \Graduate & \Graduate & \Graduate & \Graduate & \Graduate & \Graduate & \Graduate & \Graduate & \Graduate & \Graduate \\
		\Graduate & \LaQuinta & \LaQuinta & \LaQuinta & \NHH & \NHH & \NHH & \HotelMarcel & \HotelMarcel & \HotelMarcel & \TheStudy \\
		\NHH & \NHH & \NHH & \NHH & \LaQuinta & \HotelMarcel & \HotelMarcel & \NHH & \TheStudy & \TheStudy & \HotelMarcel \\
		\Omni & \Omni & \HotelMarcel & \HotelMarcel & \HotelMarcel & \LaQuinta & \TheStudy & \TheStudy & \NHH & \NHH & \NHH \\
		\HotelMarcel & \HotelMarcel & \Omni & \TheStudy & \TheStudy & \TheStudy & \LaQuinta & \LaQuinta & \LaQuinta & \Omni & \Omni \\
		\TheStudy & \TheStudy & \TheStudy & \Omni & \Omni & \Omni & \Omni & \Omni & \Omni & \LaQuinta & \LaQuinta \\
		\bottomrule
	\end{tabular}
	}
	\caption{Squared Kemeny when mixing two criteria.}
	\label{fig:booking-single-crossing}
\end{figure}

To understand how the Squared Kemeny rule behaves, it is instructive to consider the problem of aggregating just two different rankings with different weights. Let us again consider a hotel booking example. In \Cref{fig:booking-single-crossing}, we show a ranking of 6 hotels offered on booking.com in New Haven, Connecticut, for the nights 8--11 July 2024 (accessed 5 February 2024). At the very left, the hotels are ranked by price, and at the very right, they are ranked by average user score.
The figure shows the output of Squared Kemeny (with ties broken consistently) when the two rankings are given different weights; the top row shows the weight given to the price ranking.

We see that Squared Kemeny smoothly interpolates between the two rankings. Indeed, the price and score rankings differ on exactly 10 pairwise comparisons, and going through the rankings from left to right, we see that in each step one pairwise swap is performed. Thus, for example, when the price ranking has weight 70\% (and the score ranking has weight 30\%), the Squared Kemeny ranking agrees with the price ranking on 7 of the 10 disagreement comparisons, and it agrees with the score ranking on 3 of 10.

This is true in general. We formalize this by saying that Squared Kemeny satisfies \emph{2-Rankings-Proportionality} (2RP). This axiom says that for every profile $R$ containing just two rankings $\succ_1$, $\succ_2$ with positive weight that disagree on $d = \swap(\succ_1, \succ_2)$ pairwise comparisons, we have
\[
{\rhd} \in \sqK(R) \iff d - \swap(\succ_i, \rhd) \in \round(R(\succ_i) \cdot d) \text{ for both $i = 1$ and $i = 2$},
\]
where $\round(z)$ is the set of one or two integers closest to $z \in \mathbb R$. Thus, for profiles with two input rankings, the Squared Kemeny rule chooses all ``mean rankings'' $\rhd$ where the number of pairwise agreements between $\rhd$ and the input rankings is proportional to their weights. 

Note that the Kemeny rule behaves very differently on profiles with two rankings -- it just outputs the input ranking with higher weight.

Our main result is an axiomatic characterization of the Squared Kemeny rule that uses the same axioms as the famous characterization of the Kemeny rule by \citet{YouLev-1978-KemenyAxioms}, but replaces their Condorcet axiom by the 2RP axiom. We also impose a mild continuity axiom.

{
\renewcommand{\thetheorem}{\ref{thm:characterization}}
\begin{theorem}
An SPF satisfies neutrality, reinforcement, continuity, and 2RP if and only if it is the Squared Kemeny rule.
\end{theorem}
}

Neutrality is a standard symmetry condition.
The reinforcement axiom is a consistency or convexity axiom, which says that if a ranking $\rhd$ is selected at two different profiles $R_1$ and $R_2$, i.e., ${\rhd} \in f(R_1) \cap f(R_2)$, then $\rhd$ is also selected for all convex combinations $\lambda R_1 + (1-\lambda) R_2$ of the two profiles ($\lambda \in (0,1)$), and that in addition $f(\lambda R_1 + (1-\lambda) R_2) = f(R_1) \cap f(R_2)$. This is a classic axiom that has been used in many characterizations in social choice \citep[e.g.,][]{You-1975-ScoringFunctions,Fish78dAVcharacterization,Myer95bSingleWinnerScoring,LaSk21aABCScoring}.

Our proof operates within the space $\mathbb{Q}^{m!}$ of (generalized) profiles and uses reinforcement in a standard way to obtain separating hyperplanes between the regions of profiles where some particular output ranking is selected. However, unlike \citet{YouLev-1978-KemenyAxioms}, we cannot pass to a lower-dimensional space of majority margins. Instead, we characterize the hyperplanes by using 2RP to construct profiles where the rule chooses rankings that form a single-crossing path, which allows us to deduce that the hyperplanes encode the Squared Kemeny rule.

\subsection*{General Proportionality Guarantee}

The 2RP axiom applies only to profiles that contain two different rankings. Can we say anything similar about Squared Kemeny in the general case? 
Inspired by the literature on proportionality in multi-winner voting \citep{lackner2023multi}, we will consider groups of rankings and will bound the maximum distance of the output ranking to the group as a function of the group's size.

\begin{figure}[th]
	\centering
	\begin{subfigure}{0.47\linewidth}
		\begin{tikzpicture}
			\begin{axis}[
				width=7cm,
				height=4.5cm,
legend style={
					cells={anchor=west},
					legend pos=south west,
					draw=none,
					font=\small
				},
				]
				\addplot [red, domain=0:1, samples=100] {x<0.5};
\legend{Actual}
			\end{axis}
			\node at (0.05,-0.25) {$\alpha = $};
			\node[rotate=90, transform shape, scale=0.83] at (-0.85, 1.4) {normalized swap distance};
		\end{tikzpicture}
		\caption{Kemeny}
	\end{subfigure}
	\quad
	\begin{subfigure}{0.47\linewidth}
		\begin{tikzpicture}
			\begin{axis}[
				width=7cm,
				height=4.5cm,
legend style={
					cells={anchor=west},
					legend pos=south west,
					draw=none,
					font=\small
				},
				]
\addplot [blue, domain=0:1, samples=200] {min(1,
					sqrt((1-x)/x))
					)};
\addplot [red,
				] coordinates {
					(0.0,1.0) (0.02,1.0) (0.04,1.0) (0.06,1.0) (0.08,1.0) (0.1,1.0) (0.12,1.0) (0.14,1.0) (0.16,1.0) (0.18,0.933333) (0.2,0.866667) (0.22,0.866667) (0.24,0.8) (0.26,0.8) (0.28,0.733333) (0.3,0.733333) (0.32,0.733333) (0.34,0.666667) (0.36,0.666667) (0.38,0.6) (0.4,0.6) (0.42,0.6) (0.44,0.533333) (0.46,0.533333) (0.48,0.533333) (0.5,0.533333) (0.52,0.466667) (0.54,0.466667) (0.56,0.466667) (0.58,0.4) (0.6,0.4) (0.62,0.4) (0.64,0.333333) (0.66,0.333333) (0.68,0.333333) (0.7,0.333333) (0.72,0.266667) (0.74,0.266667) (0.76,0.266667) (0.78,0.2) (0.8,0.2) (0.82,0.2) (0.84,0.133333) (0.86,0.133333) (0.88,0.133333) (0.9,0.133333) (0.92,0.0666667) (0.94,0.0666667) (0.96,0.0666667) (0.98,0.0) (1.0,0.0)
				};
\legend{Upper bound,Actual}
			\end{axis}
			\node at (0.05,-0.25) {$\alpha = $};
		\end{tikzpicture}
		\caption{Squared Kemeny}
	\end{subfigure}
	\vspace{-7pt}
	\caption{The maximum swap distance (normalized to $[0, 1]$) between the output of the Kemeny or Squared Kemeny rule to an input ranking, as a function of the weight $\alpha$ of the input ranking, for $m = 6$ alternatives.}
	\label{fig:intro-alpha-curve}
\end{figure}
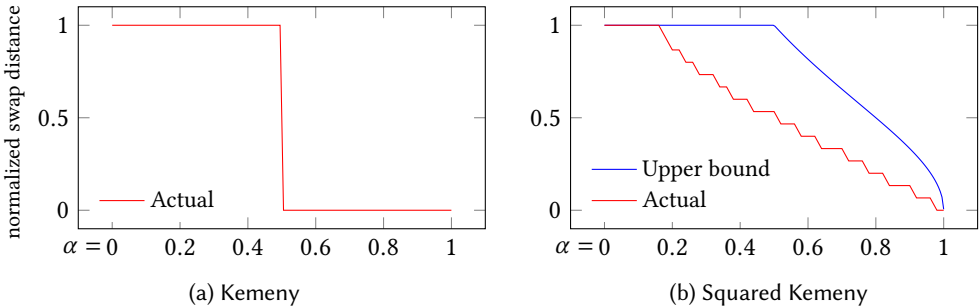

Our first question asks how large the swap distance can be between an input ranking $\succ$ with weight $\alpha \in [0,1]$ and the output ranking $\rhd$ of some rule $f$, as a function of $\alpha$. For Kemeny, this is easy to specify: for $\alpha < \frac12$, Kemeny might output the reverse ranking of $\succ$ (if this ranking has weight more than $\frac12$), and so the distance can be as large as $\binom{m}{2} = \frac{m(m-1)}{2}$ (the highest possible swap distance). When $\alpha > \frac12$, the distance is guaranteed to be 0 by the Condorcet property of Kemeny.

For Squared Kemeny, we expect the bound to be smoother, giving some guaranteed influence to rankings with weights below $\frac12$. This is indeed the case. We can compute the exact worst-case bound provided by Squared Kemeny for fixed small numbers $m$ of alternatives (see \Cref{fig:intro-alpha-curve}), and we can see that it is approximately linear in $\alpha$, except for small $\alpha$. We can also prove a theoretical upper bound that holds for all $m$: the maximum distance of the Squared Kemeny output to an input ranking with weight $\alpha$ is at most $\sqrt{(1-\alpha)/\alpha} \cdot \binom{m}{2}$ (\Cref{thm:prop-guarantee-one-ranking}). This bound implies that Squared Kemeny will never output the reverse ranking of an input ranking that has weight more than $\frac12$. For large $m$, we prove another bound that shows this even for rankings with weight more than $\frac14$.

Our second question is about giving a similar type of guarantee not to a single ranking, but to a group of rankings. For example, there could be many  similar rankings in the profile that each have a small weight, but which collectively have a significant weight. We want to show that the Squared Kemeny outcome cannot be too far away from those rankings, on average. \Cref{thm:prop-guarantee} establishes a bound that applies to all groups of rankings with total weight $\alpha$ (whether these rankings are similar or not), guaranteeing that the Squared Kemeny output has an average distance of at most $\smash{\sqrt{1/(4\alpha)} \cdot \binom{m}{2} + o(m^{1.5})}$ to the rankings in the group, where the lower-order term vanishes quickly.

\subsection*{Empirical Analysis}

We complement our theoretical analysis with results from simulations to better understand how the Squared Kemeny rule compares to the Kemeny rule. Since we need to compute the outcomes of the rules, we discuss their computational complexity in \Cref{sec:computation}. In \Cref{sec:experiments:city}, we then perform a detailed analysis of an example of using the rank aggregation rules to rank cities according to a mixture of three criteria (GDP per capita, air quality, and sunniness).

\begin{wrapfigure}[13]{r}{0.23\textwidth}
	\begin{tikzpicture}
		[euclid/.style={inner sep=0pt, draw=black!40}]
		\node[euclid] at (0,0) {\includegraphics[width=\linewidth]{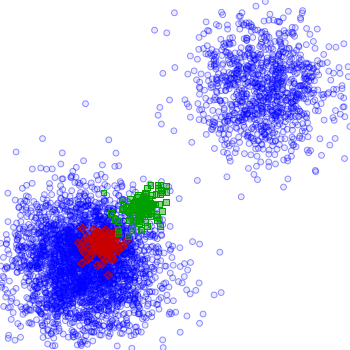}};
	\end{tikzpicture}
	\caption{Euclidean profiles with Kemeny (red) and Squared Kemeny (green) ranking locations shown}
	\label{fig:intro:euclidean}
\end{wrapfigure}
Next, we analyze in \Cref{sec:embeddings} the rankings chosen by Kemeny and Squared Kemeny on random data. For example, we sample Euclidean profiles, where criteria and alternatives correspond to points in 2D space, and rankings are induced by sorting the alternatives by their distance to each criterion. 
In more detail, for the example in \Cref{fig:intro:euclidean}, we sampled 100 profiles with 40 rankings each, where 75\% of the input rankings come from a Gaussian in the lower left corner and the other 25\% from a Gaussian in the upper right corner. We then computed the output rankings of the Kemeny rule (red diamonds) and of the Squared Kemeny rule (green squares), and embed these rankings in the same Euclidean space. We observe in \Cref{fig:intro:euclidean} that the Kemeny rule is located within the larger of the two voter clusters, while the Squared Kemeny rule interpolates between the two clusters.

Finally, in \Cref{sec:experiments:group-distance} we revisit our quantitative worst-case proportionality bounds from an average-case perspective, and experimentally investigate the distance between an output ranking and a group of input rankings, as a function of the total group weight.
The results confirm our theoretical predictions: the Kemeny rule ensures only that large groups are satisfied, while the Squared Kemeny rule caters to all group sizes.

\subsection*{Applications of Proportional Rank Aggregation}

Hotel booking websites are just one example where it makes sense to give users fine-grained control over how to sort items and where proportional rank aggregation methods such as Squared Kemeny are desirable. Further examples are lists of products in e-commerce (ranking by cost, rating, delivery time, etc.), newsfeeds of social networks, and database display applications in general.

There are also less technical applications, such as producing university rankings. These rankings are usually a result of aggregating rankings for several criteria (such as student satisfaction, \% of students employed after graduating, research output). These rankings could be weighted and then be used to produce an aggregate ranking via Squared Kemeny as all criteria should be taken into account. Similarly, one could produce rankings of cities by livability or suitability for remote work.

In all previous examples, the input rankings are \emph{criteria} (which tend to be objective) and the whole setup is essentially single-agent. However, there are also compelling multi-agent applications, where we can think of the input rankings as \emph{votes}.
An example might be a university hiring committee needing to rank applicants. In such a scenario, each committee member can provide their personal ranking. The output should be a ranking instead of a single winning candidate, because we do not know which candidates will accept the job offer. Proportionality may be desirable in this context to ensure that the output ranking reflects the diverse interests of the university department. Other multi-agent examples are groups of friends wanting to produce rankings of favorite music, restaurants, or travel destinations, a context in which a majoritarian method seems out of place.

In our example applications based on criteria, one may object that many criteria are numerical in nature (e.g., hotel price and average user rating). Using only the induced ranking throws away information, and taking a weighted average of the underlying numbers may produce a better result. The advantage of the rank aggregation approach is that it does not require the aggregator to decide on how to normalize the numerical values of different criteria. Normalization can be a difficult task without principled solutions -- consider that hotel prices are in the hundreds while user ratings are between 0 and 10, and that for prices, lower is better, while for ratings, higher is better. 
Rank aggregation sidesteps these problems, and it is also more robust to outlier values.

\section{The Model}

Let $A$ be a finite set of $m\geq 2$ alternatives.
A \emph{ranking} $\succ$ is
a complete, transitive, and anti-symmetric binary relation on $A$.
We denote by  $\mathcal{R}$ the set of all rankings on $A$.
In this paper, we study the problem of aggregating rankings into a collective ranking.
To this end, we define a \emph{(ranking) profile} $R$
to be a function from $\mathcal{R}$ to weights in $[0,1]\cap \mathbb{Q}$
such that $\sum_{{\succ}\in\mathcal{R}} R({\succ})=1$.
More intuitively, a ranking profile specifies for each ranking a weight;
these weights may, e.g., arise from a user assigning importance to different criteria
(as in the hotel example in the introduction)
or represent the fraction of voters who report a ranking in an election.
The restriction that the weights are rational
is to ensure compatibility with electorate settings (where discrete voters report preference relations) and does not affect our results.
We denote the set of all rankings profiles $R$ by $\mathcal{R}^*$. 
For a profile $R$, we write $\supp(R) = \{{\succ} \in \mathcal{R} : R(\succ) > 0\}$ for the set of rankings with positive weight.

Given a profile $R$, we want to derive an aggregate ranking.
A \emph{social preference function (SPF)}
does this, being a function $f : \mathcal{R}^* \rightarrow 2^\mathcal{R}\setminus \{\emptyset\}$
that for every profile $R\in\mathcal{R}^*$ returns
a non-empty set of chosen rankings $f(R) \subseteq \mathcal{R}$.
To help distinguish between input and output rankings,
we will follow a convention of denoting the input rankings by $\succ$
and the output rankings by $\ssucc$. 
For two rankings $\succ, \rhd \in \mathcal{R}$, their \emph{swap distance}
is the number of pairs of alternatives
which they order differently, i.e., $\swap(\succ, \rhd) = |\{a, b \in A : a \succ b \text{ and } b \rhd a\}|$.

We focus on two SPFs in this paper.
The \emph{Kemeny} rule
is defined as the set of rankings minimizing the (weighted) average swap distance between the input rankings and the output ranking, i.e.,
$\text{Kemeny}(R) = \arg \min_{\rhd \in \mathcal{R}} \sum_{{\succ}\in\mathcal{R}} R({\succ}) \cdot \swap({\succ}, \ssucc)$.
The \emph{Squared Kemeny}
rule is defined analogously, but with the swap distance squared, i.e.,
$\sqK(R) = \arg \min_{\rhd \in \mathcal{R}} \sum_{{\succ}\in\mathcal{R}} R({\succ}) \cdot \swap({\succ}, \ssucc)^2$.
For a ranking $\rhd$, we let $C_\sqK(R, \rhd) = \sum_{{\succ}\in\mathcal{R}} R({\succ}) \cdot \swap({\succ}, \ssucc)^2$ be its \emph{Squared Kemeny cost} in $R$.

\section{Axiomatic Analysis}
\label{sec:axioms}

We begin by analyzing the Squared Kemeny rule from an axiomatic perspective. First, we demonstrate in \Cref{subsec:proportional} that this rule indeed behaves like an average for special profiles (profiles where only two rankings have positive weight as well as single-crossing profiles). This means that the Squared Kemeny rule is proportional on these profiles, and we use this insight to characterize this SPF in \Cref{subsec:characterization}. Finally, we also consider standard properties such as efficiency, participation, and strategyproofness, and check which of them are satisfied by the Squared Kemeny rule in \Cref{subsec:moreaxioms}.

\subsection{2-Rankings-Proportionality and Single-Crossing Profiles}
\label{subsec:proportional}

We start by showing that the Squared Kemeny rule is proportional in some special cases by defining two properties that formalize what it means to ``proportionally'' aggregate rankings on important classes of profiles. Our first property concerns cases where the rule must ``average'' two rankings with specified weights, just like in the introduction's hotel example. It requires that the output ranking must agree with the input rankings on a proportional number of pairwise comparisons.

\paragraph{2-Rankings-Proportionality.} An SPF $f$ satisfies  \emph{2-Rankings-Proportionality} (2RP) if, for all profiles $R$ with $\supp(R)=\{{\succ_1},{\succ_2}\}$ for two rankings $\succ_1$ and $\succ_2$ with $d=\swap({\succ_1}, {\succ_2})$, it holds that 
\begin{align*}
	f(R) &= 
	\{\ssucc \in \mathcal{R}\colon d - \swap({\succ_i}, \ssucc) \in \round(R(\succ_i) \cdot d) \text{ for $i \in \{1,2\}$}
	\}, \\
	\intertext{or equivalently,}
	f(R) &= 
	\{\ssucc \in \mathcal{R}\colon \swap({\succ_i}, \ssucc) \in \round((1-R(\succ_i)) \cdot d) \text{ for $i \in \{1,2\}$}
	\},
\end{align*}
where $\round(z)$ denotes the set of closest integers to $z$.\footnote{For $k\in \mathbb{Z}$, $\round(k+x)=\{k\}$ if $x\in [0, 0.5)$, $\round(k+0.5)=\{k, k+1\}$, and $\round(k+x)=\{k+1\}$ if $x\in (0.5,1]$.}
Less formally, this means that the higher the weight of $\succ_1$ (resp. $\succ_2$), the closer the output rankings are to $\succ_1$ (resp. $\succ_2$).  
\medskip

For example, suppose that $\swap(\succ_1, \succ_2) = 10$, so the two rankings in $R$ disagree on 10 pairwise comparisons and that $R(\succ_1) = 30\%$. Then the output ranking $\rhd$ should agree with $\succ_1$ on $30\% \cdot 10 = 3$ of those disagreement pairs, and disagree with $\succ_1$ on $(1 - 30\%) \cdot 10 = 7$ of the disagreement pairs.

The Squared Kemeny rule satisfies 2RP, and in fact it satisfies a stronger property about \emph{single-crossing} profiles.
A sequence of rankings $\succ_0, \dots, \succ_n \in \mathcal{R}$ is called \emph{single-crossing}
if for every pair of alternative $a, b \in A$ with $a \succ_0 b$,
\[
\text{there exists $i \in \{0, \dots, n\}$ such that }
a \succ_0 b, \dots, a \succ_i b \text{ and } b \succ_{i+1} a, \dots, b \succ_n a.
\]
Thus, scanning the rankings from left to right, the relative positions of every pair of alternatives cross at most once. We say that a sequence is \emph{maximal single-crossing} if every pair of alternatives crosses exactly once (which implies that $n = \binom{m}{2}$ and that $\succ_0$ and $\succ_n$ are reverse rankings). 
As an example, the rankings shown in \Cref{fig:single-crossing} form a maximal single-crossing sequence.

\begin{wrapfigure}{r}{0.48\textwidth}
	\centering
	\includegraphics[width=\linewidth]{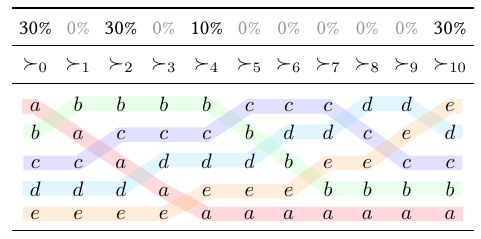}
	\vspace{-18pt}
	\caption{A single-crossing profile, on which Squared Kemeny outputs $\succ_4$ (the average of the voter locations) and Kemeny outputs $\succ_2$ (the median).}
	\label{fig:single-crossing}
\end{wrapfigure}

We say that a \emph{profile} $R$ is \emph{single-crossing} if the rankings in $\supp(R)$ can be arranged in a single-crossing sequence.
On single-crossing profiles, there is a natural definition of what it means to be an average, because each input ranking is associated with a location $0, \dots, n$ in a one-dimensional space; thus the output ranking should be at the weighted average of these locations. For example, the profile in \Cref{fig:single-crossing} has 30\% of the voters each in locations 0, 2, and 10, with the remaining 10\% in location 4. This gives an average location of $0.3 \cdot (0 + 2 + 10) + 0.1 \cdot 4 = 4$. Thus, the ``average'' ranking for this profile is $\succ_4$, and indeed this is the output ranking of Squared Kemeny. In contrast, the Kemeny rule takes the median location, which is location 2, and so $\succ_2$ is the Kemeny output.\footnote{This is known as the \emph{representative voter theorem} \citep{rothstein1991representative} which states that the majority relation of a single-crossing profile (and hence the Kemeny ranking) coincides with the input ranking of the median voter.}

To formalize this behavior, let us say that $R$ is \emph{compatible} with a maximal single-crossing sequence $\succ_0, \dots, \succ_n$ if $\supp(R) \subseteq \{\succ_0, \dots, \succ_n\}$. We can now state an axiom specifying what it means to proportionally aggregate rankings on a single-crossing profile.

\paragraph{Single-Crossing-Proportionality.} An SPF $f$ satisfies \emph{Single-Crossing-Proportionality} if for every single-crossing profile $R$ we have that a ranking $\rhd$ is in $f(R)$ if and only if there exists a maximal single-crossing sequence $\succ_0, \dots, \succ_n$ compatible with $R$ and ${\rhd} = {\succ_i}$ for $i \in \round(\sum_{j=1}^n R(\succ_j) \cdot j)$.

\medskip
It is straightforward to check that Single-Crossing-Proportionality implies 2RP, because any profile in which only 2 rankings occur is single-crossing.

\begin{theorem}
	\label{thm:sqk-single-crossing}
	The Squared Kemeny rule satisfies Single-Crossing-Proportionality and 2RP.
\end{theorem}
\begin{proof}
Since Single-Crossing-Proportionality implies 2RP,
	it is sufficient to prove that Squared Kemeny satisfies the former.

	For this, we fix a single-crossing profile $R$ and an arbitrary maximal single-crossing sequence $\succ_0, \dots, \succ_n$ compatible with $R$.
	Also, let $\rhd \in \mathcal{R}$ be an arbitrary ranking and define $d = \swap(\succ_0, \rhd)$.
Because of the triangle inequality, we can now compute that 
	\begin{align}
	\notag
	\swap(\succ_i,\rhd) &\ge
	\swap(\succ_0, \rhd) - \swap(\succ_0, \succ_i) =
	d - i,
	\quad \mbox{for } i \in [d], \mbox{ and}\\
	\label[ineqs]{ineqs:single-crossing}
	\swap(\succ_i, \rhd) &\ge
	\swap(\succ_0, \succ_i) - \swap(\succ_0, \rhd) =
	i - d,
	\quad \mbox{for } i \in [n] \setminus [d].
	\end{align}
	We will first show that if $\rhd$ does not belong to any maximally single-crossing sequence compatible with $R$, then at least one of \Cref{ineqs:single-crossing} for $i \in [n]$ such that ${\succ_i} \in \supp(R)$ is strict.
	Assume otherwise, i.e.,
	$\swap(\succ_i,\rhd) = |d - i|$ for every $i \in [n]$ such that ${\succ_i} \in \supp(R)$.
	If $R(\succ_d) > 0$,
	this implies that $\rhd = {\succ_d}$,
	so $\rhd$ belongs to the maximal single-crossing sequence $\succ_0, \dots, \succ_n$.
	Thus, let us assume that $R(\succ_d) = 0$.
	Then, let $i \in \{0, \dots, d-1\}$ be maximal and $j \in \{d+1, \dots, n\}$ minimal such that ${\succ_i}, {\succ_j} \in \supp(R)$.
	By adding the respective equalities sidewise, we obtain that
	\[
		\swap(\succ_i,\rhd) + \swap(\succ_j,\rhd) =
		(d - i) + (j - d) =
		j - i = \swap(\succ_i, \succ_j).
	\]
	By \cite[Proposition 4.6]{ElkLacPet-2022-RestrictedDomains},
	this implies that $\succ_i$, $\rhd$, and $\succ_j$ form a single-crossing sequence
	and $\rhd$ is between $\succ_i$ and $\succ_j$.
	This means that for every pair of alternatives $x, y \in A$ such that $x \mathrel{\rhd} y$,
	it holds that $y \succ_i x$ implies $x \succ_j y$
	and $y \succ_j x$ implies $x \succ_i y$.
	Consequently, the sequence
	$\succ_0, \dots, \succ_i, \rhd, \succ_j, \dots, \succ_n$ is single-crossing,
	which contradicts the assumption
	that $\rhd$ does not belong to any maximally single-crossing sequence compatible with $R$.
	
	Therefore, if $\rhd$ does not belong to any maximally single-crossing sequence compatible with $R$,
	we can show that there will be a ranking in such set, i.e., $\succ_d$, for which
	$C_{\sqK}(R, \rhd) > C_{\sqK}(R, \succ_d)$.
	Indeed, from \Cref{ineqs:single-crossing} and the fact that one of them is strict, we get that
	\[
		C_{\sqK}(R, \rhd) =
		\textstyle \sum_{i= 0}^n
			R(\succ_i) \cdot \swap(\succ_i, \rhd)^2 >
		\textstyle \sum_{i= 0}^n
			R(\succ_i) \cdot (i - d)^2 =
		C_{\sqK}(R, \succ_d).
	\]
	Thus, $\rhd \not \in \sqK(R)$
	and we know that the only rankings selected by the Squared Kemeny rule belong to some maximally single-crossing sequence compatible with $R$.

	Then, take an arbitrary $\rhd \in \sqK(R)$ and again denote $d = \swap(\succ_0, \rhd)$.
	This means that
	\(
		C_{\sqK}(R, \rhd) = 
		\sum_{i = 0}^n
		R(\succ_{i}) \cdot (i - d)^2.
	\)
	Taking the derivative with respect to $d$, we get that $C_{\sqK}(R, \rhd)$ is minimized when
	$\sum_{i = 0}^n 2 R(\succ_{i}) \cdot (d - i) = 0$,
	which is equivalent to
	$\sum_{i = 0}^n R(\succ_{i}) \cdot i = d$.
	Since $d$ has to be an integer and a quadratic function grows symmetrically from its minimum
	(note that a convex combination of quadratic functions is still a quadratic function),
	the rankings
	$\succ_d$ for $d \in \round(\sum_{i = 0}^n R(\succ_{i}) \cdot i)$ have the lowest Squared Kemeny cost among the rankings in $\succ_0, \dots, \succ_n$.
	
	It remains to show that for $\succ_d$ selected in this way, the cost $C_{\sqK}(R, \succ_d)$
	is the same no matter which maximally single-crossing sequence we have chosen at the beginning.
	To this end,
	take two arbitrary maximally single-crossing sequences $\succ_0, \dots, \succ_n$ and
	$\succ'_0, \dots, \succ'_n$, both compatible with $R$.
Observe that there is a linear function
	$\ell(x) = ax + b$,
	with $a \in \{1, -1\}$ 
	such that
	${\succ_i} ={ \succ'_{\ell(i)}}$
	for every ${\succ_i} \in \supp(R)$.
	Then, by the linearity of the mean,
	if $\succ_d$ minimizes $C_{\sqK}(R,\succ_j)$
	among $\succ_0, \dots, \succ_n$,
	then $\succ'_{\ell(d)}$ minimizes $C_{\sqK}(R,\succ'_j)$
	among $\succ'_0, \dots, \succ'_n$.
	Furthermore, we have that
	\(
		C_{\sqK}(R,\succ_d) =
		\textstyle\sum_{i=0}^n R(\succ_i) \cdot (i - d)^2 =
		\textstyle\sum_{i=0}^n R(\succ_i) \cdot (\ell(i) - \ell(d))^2 =
		C_{\sqK}(R,\succ_{\ell(d)}),
	\)
	which concludes the proof.
\end{proof}

\subsection{Characterization of the Squared Kemeny Rule}
\label{subsec:characterization}

We will next present our characterization of the Squared Kemeny rule, which combines 2RP with three standard properties, namely neutrality, reinforcement, and continuity.

\paragraph{Neutrality.} Neutrality is a mild symmetry condition that precludes a rule from depending on the names of candidates. An SPF $f$ is \emph{neutral} if $f(\tau(R))=\{\tau({\ssucc})\colon {\ssucc}\in f(R)\}$ for all profiles $R\in\mathcal{R}^*$ and permutations $\tau:A\rightarrow A$. Here, we denote by ${\ssucc'}=\tau({\ssucc})$ the ranking defined by $\tau(x)\mathrel{\ssucc}'\tau(y)$ if and only if $x\mathrel{\ssucc} y$ for all $x,y\in A$ and $R'=\tau(R)$ is the profile defined by $R'(\tau({\succ}))=R({\succ})$ for all ${\succ}\in\mathcal{R}$.

\paragraph{Reinforcement.} Reinforcement is a classic axiom in social choice theory which describes that if some outcomes are chosen for two profiles, then precisely these common outcomes should be chosen in a convex combination of these profiles. An SPF $f$ satisfies \emph{reinforcement} if for all profiles $R,R'\in\mathcal{R}^*$ with $f(R)\cap f(R')\neq\emptyset$, we have $f(\lambda R+(1-\lambda) R')=f(R)\cap f(R')$ for all $\lambda\in (0,1)\cap \mathbb{Q}$.

\paragraph{Continuity.} Continuity requires that a group of rankings with sufficient weight can overrule any other set of rankings and thus determine the outcome. Formally, an SPF $f$ is \emph{continuous} if for all profiles $R,R'\in\mathcal{R}^*$, there is a scalar $\lambda\in(0,1)\cap \mathbb{Q}$ such that $f(\lambda R+(1-\lambda) R')\subseteq f(R)$.\medskip

The above three axioms that have been frequently used in social choice theory to characterize scoring rules in various contexts \citep[e.g.,][]{You-1975-ScoringFunctions,Myer95bSingleWinnerScoring,SFS19aCommitteeScoring,LaSk21aABCScoring,Lede23aSPFscoring}. Moreover, Kemeny's rule has been characterized as the unique SPF satisfying neutrality, reinforcement, and a Condorcet axiom \citep{YouLev-1978-KemenyAxioms}.

We can now state our characterization result.

\begin{restatable}{theorem}{characterization}\label{thm:characterization}    
An SPF satisfies neutrality, reinforcement, continuity, and 2RP if and only if it is the Squared Kemeny rule.
\end{restatable}

Like other reinforcement-based characterizations of SPFs, our proof is quite involved, and thus we defer it to \Cref{app:proof-characterization}. Using well-known techniques, the proof uses reinforcement to divide the space of profiles into convex regions where a particular ranking is chosen by the SPF, and then uses 2RP in combination with the other axioms to characterize the boundaries (hyperplanes) of these regions. As for the independence of the axioms, we do not know if neutrality or continuity can be dropped from the characterization. Without 2RP, the Kemeny rule satisfies the remaining axioms. 
Without reinforcement, the rule that agrees with the Squared Kemeny rule on all profiles in which two rankings jointly have more than 90\% of the weight, and returns the set of all rankings $\mathcal{R}$ for all other profiles, satisfies the remaining axioms. 

To give more insights into the proof of \Cref{thm:characterization}, we will show a weaker statement that is still of interest. To this end, we introduce the family of \emph{ranking scoring functions}. 
These are analogues of scoring rules in voting and are defined based on a \emph{cost function} $c \colon \mathcal{R} \times \mathcal{R} \to \mathbb{R}$ that assigns to each pair of rankings ${\succ}, \rhd \in \mathcal{R}$ a cost $c({\succ}, \rhd)$. Intuitively, we interpret the term $c({\succ},\rhd)$ as the disutility that the outcome ranking $\rhd$ would give to a voter with a preference order $\succ$.
The ranking scoring function $f_c$ based on $c$ returns the rankings with minimal total cost: $f_c(R) = \arg \textstyle\min_{\rhd \in \mathcal{R}} \sum_{\succ \in \mathcal{R}} R(\succ) \cdot c(\succ, \rhd)$ for every profile $R$. For example, $f_{\swap}$ is the Kemeny rule and $f_{\swap^2}$ is the Squared Kemeny rule.

It is straightforward to check that every ranking scoring function satisfies reinforcement and continuity. The class of (neutral) ranking scoring functions was introduced by \citet{ConRogXia-2009-ScoringRankings}, who conjectured that this class is in fact characterized by neutrality, reinforcement, and continuity.

We now give a proof that the only ranking scoring function that satisfies the 2RP axiom is the Squared Kemeny rule. The proof uses a very similar strategy to the proof of our full axiomatic characterization (\Cref{thm:characterization}) but avoids some of its overhead. Note that this version of the characterization does not require neutrality.

\begin{theorem}
	\label{thm:welfarist}
	A ranking scoring function satisfies 2RP if and only if it is the Squared Kemeny rule.
\end{theorem}
\begin{proof}	
We already showed that the Squared Kemeny rule satisfies 2RP (\Cref{thm:sqk-single-crossing}), so we need to show that it is the only ranking scoring function with that property. Hence, let $f$ denote a ranking scoring function that satisfies 2RP and let $c$ be its cost function.
Without loss of generality,
we can assume for every ranking ${\succ} \in \mathcal{R}$ that $c(\succ, \succ) = 0$
because adding a constant to a cost function,
even a different constant for each first argument,
does not affect the outcomes of a ranking scoring function.
For a profile $R$ and a ranking $\rhd$, we define
$C(R, \rhd) = \sum_{\succ \in \mathcal{R}} R(\succ) \cdot c(\succ, \rhd)$.
We will show that $c$ is proportional to
the Squared Kemeny cost function $c_\sqK(\succ, \rhd) = \swap(\succ, \rhd)^2$, which  implies that $f$ is the Squared Kemeny rule.
We prove this in three steps:
first, we will show that for every pair of rankings $\rhd_1, \rhd_2 \in \mathcal{R}$
with swap distance $1$, the difference in the costs of $\rhd_1$ and $\rhd_2$ with respect to any other ranking 
is proportional to the difference in their Squared Kemeny costs, i.e.,
there exists a constant $\alpha_{\rhd_1,\rhd_2} > 0$ such that
\(
c(\succ, \rhd_1)-c(\succ, \rhd_2) =
\alpha_{\rhd_1,\rhd_2}
(c_\sqK({\succ,} \rhd_1)-c_\sqK({\succ}, \rhd_2))
\)
for every ${\succ} \in \mathcal{R}$
(Step 1).
Next, we will prove that these constants are equal
if two such pairs of rankings intersect,
i.e., 
$\alpha_{\rhd_1,\rhd_2}=\alpha_{\rhd_2,\rhd_3}$
for all $\rhd_1,\rhd_2,\rhd_3 \in \mathcal{R}$
such that $\swap(\rhd_1, \rhd_2)=\swap(\rhd_2, \rhd_3)=1$
(Step~2).
Finally, we infer that all these constants are equal and derive that there is $\alpha>0$ such that
\(
	c({\succ}, \rhd_1) - c({\succ}, \rhd_2) =
	\alpha (c_\sqK({\succ}, \rhd_1) - c_\sqK({\succ}, \rhd_2))
\)
for all $\succ,\rhd_1, \rhd_2 \in \mathcal{R}$. Since $c(\ssucc,\ssucc)=c_{\sqK}(\ssucc,\ssucc)$ for all $\ssucc\in\mathcal{R}$, this means that $c$ is indeed 
proportional to $c_\sqK$ (Step~3).
\medskip
	
\textbf{Step 1:}
Fix two rankings $\rhd_1,\rhd_2 \in \mathcal{R}$ with $\swap(\rhd_1,\rhd_2)=1$
and define $\alpha_{\rhd_1,\rhd_2} = c(\rhd_1, \rhd_2)$.
We will show that $\alpha_{\rhd_1, \rhd_2}>0$ and
$c(\succ, \rhd_1) - c(\succ, \rhd_2) = 
\alpha_{\rhd_1,\rhd_2}(c_\sqK({\succ}, \rhd_1) - c_\sqK({\succ}, \rhd_2))$ for all $\succ \ \in \mathcal{R}$.
Let us first prove that $\alpha_{\rhd_1,\rhd_2}$
is equal to $\alpha_{\rhd_2,\rhd_1}=c(\rhd_2, \rhd_1)$.
For this, we consider the profile $R$ with $R(\rhd_1)=R(\rhd_2)=\nicefrac{1}{2}$.
2RP requires that $f(R) = \{ \rhd_1, \rhd_2\}$.
Thus, from the definition of ranking scoring functions
and our assumption that $c(\succ,\succ) = 0$ for every ${\succ} \in \mathcal{R}$
we derive that $c(\rhd_1, \rhd_2) = c(\rhd_2, \rhd_1)$ and therefore 
\begin{equation}
	\label{eq:welfarist:alpha:symmetric}
	\alpha_{\rhd_1,\rhd_2} =
	\alpha_{\rhd_2,\rhd_1}.
\end{equation}
To show that $\alpha_{\rhd_1,\rhd_2} > 0$,
consider another profile $R'$
with $R'(\rhd_1)=\nicefrac{2}{3}$ and $R'(\rhd_2)=\nicefrac{1}{3}$.
By 2RP, we get that $f(R') = \{ \rhd_1\}$ and thus that
$\nicefrac{2}{3} \cdot c(\ssucc_1,\ssucc_1)+\nicefrac{1}{3} \cdot c(\ssucc_2, \ssucc_1) < \nicefrac{2}{3} \cdot c(\ssucc_1,\ssucc_2)+\nicefrac{1}{3} \cdot c(\ssucc_2, \ssucc_2)$.
Because $c(\ssucc_1, \ssucc_2)=c(\ssucc_2, \ssucc_1)$ and $c(\rhd_1,\rhd_1)=c(\rhd_2,\rhd_2)=0$, we derive that
$\alpha_{\rhd_1,\rhd_2} > 0$.

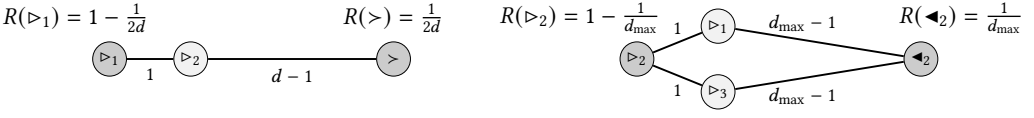
\begin{figure*}[t]
\centering
\scalebox{0.9}{
	\begin{tikzpicture}
		\def\s{1.2cm} 

		\tikzset{
			node/.style={circle,draw,minimum size=0.5cm,inner sep=0, fill = black!05, font=\footnotesize},
			nodeb/.style={circle,draw,minimum size=0.5cm,inner sep=0, fill = black!20, font=\footnotesize},
			edge/.style={draw,thick,below}
		}
		
\node[nodeb] (1) at (0*\s, 0) {$\rhd_1$};
		\node (t1) at (-0.5, 0.6) {$R(\rhd_1) = 1 - \frac{1}{2d}$};
		\node[node] (2) at (1*\s, 0) {$\rhd_2$};
\node[nodeb] (r) at (3.5*\s, 0) {$\succ$};
		\node (t2) at (3.5*\s, 0.6) {$R(\succ) = \frac{1}{2d}$};
		
		\path[edge]
		(1) edge node[below] {\footnotesize $1$} (2)
		(2) edge node[below] {\footnotesize $d-1$} (r)
		;
		
\node[nodeb] (v2) at (6.5*\s, 0) {$\rhd_2$};
		\node (t1) at (5.8*\s, 0.6) {$R(\rhd_2) = 1- \frac{1}{d_{\max}}$};
		\node[node, label={270:}] (v1) at (7.5*\s, 0.4*\s) {$\rhd_1$};
		\node[node, label={270:}] (v3) at (7.5*\s, -0.4*\s) {$\rhd_3$};
\node[nodeb] (vr) at (10*\s, 0) {$\blacktriangleleft_2$};
		\node (t1) at (10.5*\s, 0.6) {$R(\blacktriangleleft_2) = \frac{1}{d_{\max}}$};
		
		\path[edge]
		(v2) edge node[above] {\footnotesize $1$} (v1)
		(v2) edge node[below] {\footnotesize $1$} (v3)
		(v1) edge node[above,pos=0.4] {\footnotesize $d_{\max}-1$} (vr)
		(v3) edge node[below,pos=0.4] {\footnotesize $d_{\max}-1$} (vr);
	\end{tikzpicture}
}
\caption{An illustration of Steps 1 (left) and 2 (right) of the proof of \Cref{thm:characterization}.}
\label{fig:thm:welfarist}
\end{figure*}
	
Now, fix an arbitrary ranking $\succ$.
Since $\swap(\rhd_1,\rhd_2)=1$,
there is exactly one pair of alternatives
on which $\rhd_1$ and $\rhd_2$ disagree.
Let us denote them by $a$ and $b$, i.e.,
$a\mathrel{\ssucc_1} b$ and $b\mathrel{\ssucc_2} a$.
We subsequently assume that $b \succ a$ and discuss the case that $a\succ b$ later. Let $d=\swap({\succ}, \ssucc_1)$ and observe that
$\swap({\succ}, \ssucc_2)=d-1$
(see \Cref{fig:thm:welfarist} for an illustration). Moreover, we define $R$ as the profile with $R(\succ)=\nicefrac{1}{2d}$ and $R(\ssucc_1)=1 - \nicefrac{1}{2d}$.
Since $\nicefrac{1}{2d} \cdot \swap(\succ, \rhd_1) = \nicefrac{1}{2}$
and $(1 - \nicefrac{1}{2d}) \cdot \swap(\succ, \rhd_1) = d -\nicefrac{1}{2}$,
2RP implies that $\rhd_1, \rhd_2 \in f(R)$.
Thus, by the definition of ranking scoring functions,
$C(R,\ssucc_1) = C(R,\ssucc_2)$, which further implies that
$0=2d\cdot C(R, \ssucc_1)-2d\cdot C(R, \ssucc_2)$.
We can now compute that
\begin{align*}
	0&=c({\succ}, \ssucc_1)+(2d-1)c(\ssucc_1, \ssucc_1) - c({\succ}, \ssucc_2) - (2d-1)c(\ssucc_1, \ssucc_2)\\
	&=c({\succ}, \ssucc_1)- c({\succ}, \ssucc_2) - (2d-1)\alpha_{\ssucc_1, \ssucc_2}.
\end{align*}
Since 
\(
	c_\sqK(\succ, \rhd_1) - c_\sqK(\succ, \rhd_2) =
	\swap({\succ}, \rhd_1)^2 - \swap({\succ}, \rhd_2)^2 = 
	d^2 - (d-1)^2 = 
	2d -1
\), it follows that 
$c({\succ}, \ssucc_1)- c({\succ}, \ssucc_2) = \alpha_{\ssucc_1, \ssucc_2} (c_\sqK(\succ, \rhd_1) - c_\sqK(\succ, \rhd_2))$.

Lastly, let us consider the case of $a \succ b$.
By analogous reasoning, we obtain that
$c({\succ}, \ssucc_2)- c({\succ}, \ssucc_1) = 
\alpha_{\ssucc_2, \ssucc_1} (c_\sqK(\succ, \rhd_2) - c_\sqK(\succ, \rhd_1))$.
By \Cref{eq:welfarist:alpha:symmetric} and
sidewise multiplication by $-1$, we hence infer the thesis of this Step also in this case.\medskip

\textbf{Step 2:} 
Consider three rankings $\ssucc_1, \ssucc_2, \ssucc_3 \in \mathcal{R}$
with $\swap(\ssucc_1, \ssucc_2)=\swap(\ssucc_2, \ssucc_3)=1$ and let $\alpha_{\ssucc_1, \ssucc_2}$ and $\alpha_{\ssucc_2, \ssucc_3}$
denote the constants derived in the previous step.
We will show in this step that
$\alpha_{\ssucc_1, \ssucc_2}=\alpha_{\ssucc_2, \ssucc_3}$.
If $m=2$, then necessarily $\rhd_1 = \rhd_3$
and our claim directly follows from \Cref{eq:welfarist:alpha:symmetric}.

Thus, assume $m \ge 3$, let $\blacktriangleleft_2$ denote the ranking that is completely reverse of $\ssucc_2$, and let $d_{\max}=\binom{m}{2}=\swap({\blacktriangleleft_2}, \rhd_2)$.
Furthermore, we define $R$ as the profile 
with $R(\blacktriangleleft_2)=\nicefrac{1}{d_{\max}}$ 
and $R(\ssucc_2)=1-\nicefrac{1}{d_{\max}}$
(see \Cref{fig:thm:welfarist} for an illustration).
Observe that $\nicefrac{1}{d_{\max}} \cdot \swap({\blacktriangleleft_2}, \rhd_2) = 1$
and $(1-\nicefrac{1}{d_{\max}}) \cdot \swap({\blacktriangleleft_2}, \rhd_2) = d_{\max}-1$.
Hence, 2RP implies that all rankings in swap distance $1$ from $\rhd_2$
are selected by $f$.
In particular, $\ssucc_1,  \ssucc_3\in f(R)$.
By the definition of the ranking scoring function, this means that
$C(R, \ssucc_1)=C(R, \ssucc_3)$, and therefore also 
$C(R, \ssucc_1)-C(R, \ssucc_2)=C(R, \ssucc_3)-C(R, \ssucc_2)$. Next, Step 1 implies that $C(R, \ssucc_1)-C(R, \ssucc_2)=\alpha_{\ssucc_1, \ssucc_2} (C_{\sqK}(R, \ssucc_1)-C_{\sqK}(R, \ssucc_2))$ and $C(R, \ssucc_3)-C(R, \ssucc_2)=\alpha_{\ssucc_3, \ssucc_2} (C_{\sqK}(R, \ssucc_3)-C_{\sqK}(R, \ssucc_2))$. Since we have
\(
	C_{\sqK}(R, \ssucc_1)
	=C_{\sqK}(R, \ssucc_1)
	=\nicefrac{1}{d_{\max}} \cdot (d_{\max} - 1)^2 + (1-\nicefrac{1}{d_{\max}}) \cdot 1^2
	<\nicefrac{1}{d_{\max}} \cdot d_{\max}^2 =
	C_{\sqK}(R, \ssucc_2)
\), 
we now derive that
\[
	\alpha_{\rhd_1,\rhd_2} =
	\frac{C(R, \ssucc_1)-C(R, \ssucc_2)}{C_{\sqK}(R, \ssucc_1)-C_{\sqK}(R, \ssucc_2)} =
	\frac{C(R, \ssucc_3)-C(R, \ssucc_2)}{C_{\sqK}(R, \ssucc_3)-C_{\sqK}(R, \ssucc_2)} =
	\alpha_{\rhd_3,\rhd_2}.
\]
Hence, by \Cref{eq:welfarist:alpha:symmetric}, it follows that $\alpha_{\rhd_1,\rhd_2} =\alpha_{\rhd_2,\rhd_3}$.
\medskip
	
	\textbf{Step 3:}
	Finally, we will show that
	$c$ is proportional to $c_{\sqK}$,
	so that $f$ is the Squared Kemeny rule.
	To this end, we first show that $\alpha_{\rhd, \rhd'} = \alpha_{\rhd^+,\rhd^*}$,
	for all $\rhd,\rhd',\rhd^+,\rhd^*\in \mathcal{R}$
	with $\swap(\rhd,\rhd') = \swap(\rhd^+,\rhd^*)=1$.
	To see this,
	observe that there exists a sequence of rankings $\rhd_1,\rhd_2,\dots,\rhd_k$
	such that $\rhd_1=\rhd$, $\rhd_2=\rhd'$, $\rhd_{k-1}=\rhd^+$, $\rhd_k = \rhd^*$, 
	and $\swap(\rhd_i,\rhd_{i+1})=1$
	for every $i \in \{1,\dots, k-1\}$.
	Then, Step 2 implies that
	$\alpha_{\rhd,\rhd'} = 
	\alpha_{\rhd_1,\rhd_2} = \dots =
	\alpha_{\rhd_{k-1},\rhd_k} = 
	\alpha_{\rhd^+,\rhd^*}$.
	We hence drop the index of these constants and simply refer to them by $\alpha$.

Next, take arbitrary rankings $\succ,\rhd, \rhd' \in \mathcal{R}$
and let $\ssucc_1,\dots, \ssucc_k$ be a sequence of rankings such that
$\ssucc_1=\ssucc$, $\ssucc_k=\ssucc'$,
and $\swap(\ssucc_i, \ssucc_{i+1})=1$
for every $i \in \{1,\dots, k-1\}$.
Using the ``telescoping sum'' technique twice, we get that
\(
	c(\succ, \ssucc)-c(\succ,\ssucc') =
	\textstyle\sum_{i=1}^{k-1}
		\left( c(\succ, \ssucc_i)-c(\succ, \ssucc_{i+1}) \right) =
	\alpha \textstyle\sum_{i=1}^{k-1} 
		\left( c_{\sqK}(\succ, \ssucc_i) - c_{\sqK}(\succ, \ssucc_{i+1}) \right) =
	\alpha(c_{\sqK}(\succ, \ssucc) -c_{\sqK}(\succ, \ssucc')).
\)

From this we get that $c$ is proportional to $c_{\sqK}$.
Indeed, for every two rankings $\succ, \rhd \in \mathcal{R}$,
we get that
\(
	c(\succ, \rhd) = 
	c(\succ, \rhd) - c(\succ, \succ) = 
	\alpha (c_{\sqK}(\succ, \rhd) - c_{\sqK}(\succ, \succ)) = 
	\alpha (c_{\sqK}(\succ, \rhd).
\)
Since $\alpha>0$, this means that $f$ is the Squared Kemeny rule.
\end{proof}

\subsection{Efficiency, Participation, and Strategyproofness}
\label{subsec:moreaxioms}

We will next show that the Squared Kemeny rule satisfies desirable efficiency and participation properties but violates strategyproofness. To define these axioms for SPFs, we first need to specify how we compare two output rankings $\ssucc_1$, $\ssucc_2$ based on an input ranking $\succ$. Following the literature \citep[e.g.,][]{BosSto-1992-KemenySP, BosSpr-2014-SPrankingaggregation,Ath-2016-KemenyCriteria}, we use the swap distance between the input ranking and the output rankings: given an input ranking $\succ$, $\ssucc_1$ is \emph{weakly preferred} to $\ssucc_2$ (denoted by $\ssucc_1\succsim\ssucc_2$) if $\swap({\succ}, {\ssucc_1})\leq \swap({\succ}, {\ssucc_2})$, and $\ssucc_1$ is strictly preferred to $\ssucc_2$ (denoted by $\ssucc_1\succ\ssucc_2$) if $\swap({\succ}, {\ssucc_1})< \swap({\succ}, {\ssucc_2})$. It does not make a difference for these purposes whether we use the swap distance or the squared swap distance since they induce the same preferences. Next, we will define efficiency, participation, and strategyproofness. 

\paragraph{Efficiency.} An outcome is (Pareto) efficient if it is not possible to make one voter better off without making any other voter worse off. Formally, we say that a ranking $\ssucc_1$ dominates another ranking $\ssucc_2 $ in a profile $R$ if $\ssucc_1\succsim\ssucc_2$ for all ${\succ}\in\supp(R)$ and $\ssucc_1\succ\ssucc_2$ for some ${\succ}\in\supp(R)$. Moreover, a ranking $\ssucc$ is efficient for a profile $R$ if it is not dominated by any other ranking. Finally, an SPF $f$ is \emph{efficient} if, for every profile $R$, every ranking $\ssucc\in f(R)$ is efficient.\footnote{\citet{BosSpr-2014-SPrankingaggregation} show that if $a \succ b$ for all ${\succ} \in \supp(R)$, then we also have $a \rhd b$ whenever $\rhd$ is efficient in $R$.}

\paragraph{Participation.} The axiom of participation is typically formulated in electoral settings and intuitively requires that it is never better for a group of agents to abstain from an election than to participate. In our context, participation can be seen as a consistency notion: if $\ssucc_1$ is a winning ranking in the profile $R$ and we add additional criteria to $R$ according to which $\ssucc_1$ is better than $\ssucc_2$, then $\ssucc_2$ should not be winning in the extended profile. More formally, we say an SPF $f$ satisfies \emph{participation} if there are no profiles $R^1, R^2$, a constant $\lambda\in (0,1)\cap \mathbb{Q}$, and rankings $\ssucc_1\in f(R^1)$, $\ssucc_2\in f(\lambda R^1+(1-\lambda R^2))$ such that $\ssucc_1 \succsim \ssucc_2$ for all ${\succ}\in\supp(R^2)$ and $\ssucc_1 \succ \ssucc_2$ for some ${\succ}\in\supp(R^2)$.

\paragraph{Strategyproofness.} As the last axiom of this section, we will consider strategyproofness. This axiom is also typically studied in electoral settings and requires that agents should never be better off by lying than by voting truthfully. We hence say an SPF $f$ is \emph{strategyproof} if there are no profiles $R^1$, $R^2$, rankings ${\succ_1}\in \supp(R^1)$, ${\succ_2}\in\mathcal{R}$, and a constant $\epsilon\in (0, R^1({\succ_1})]$ such that \emph{(i)} $f(R^1)=\{\ssucc_1\}$ and $f(R^2)=\{\ssucc_2\}$ for some rankings $\ssucc_1, \ssucc_2$ with $\ssucc_2 \succ_1 \ssucc_1$, and \emph{(ii)} $R^2({\succ_1})=R^1({\succ_1})-\epsilon$, $R^2({\succ_2})=R^1({\succ_2})+\epsilon$, and $R^2({\succ})=R^1({\succ})$ for all ${\succ}\in\mathcal{R}\setminus \{{\succ_1},{\succ_2}\}$. 

\begin{proposition}\label{prop:EfficiencyParticipationStrategyproofness}
	The Squared Kemeny rule satisfies efficiency and participation but violates strategyproofness. 
\end{proposition}
\begin{proof}
	We first show that Squared Kemeny is efficient. For this, let $R$ be a profile and $\ssucc_1$, $\ssucc_2$ be two rankings such that $\ssucc_1$ dominates $\ssucc_2$ in $R$. It is easy to check that $C_{\sqK}(R, \ssucc_1)<C_{\sqK}(R, \ssucc_2)$, so $\ssucc_2\not\in \sqK(R)$, which implies that the Squared Kemeny rule is indeed efficient.
	
	For participation, consider two profiles $R^1$, $R^2$,  and two rankings $\ssucc_1$, $\ssucc_2$ such that $\ssucc_1\in \sqK(R^1)$, $\ssucc_1\succsim \ssucc_2$ for all ${\succ}\in \supp(R^2)$, and $\ssucc_1\succ \ssucc_2$ for some ${\succ}\in \supp(R^2)$. Since $\ssucc_1\in \sqK(R^1)$, we have $C_{\sqK}(R^1, \ssucc_1)\leq C_{\sqK}(R^1, \ssucc_2)$. Moreover, from the conditions on $R^2$ we get that $C_{\sqK}(R^2, \ssucc_1)< C_{\sqK}(R^1, \ssucc_2)$. This implies for every $\lambda\in (0,1)\cap\mathbb{Q}$ that $\ssucc_2\not\in \sqK(\lambda R^1+(1-\lambda) R^2)$ because 
	\begin{align*}
		C_{\sqK}(\lambda R^1+(1-\lambda) R^2, \ssucc_1)&=\lambda C_{\sqK}(R^1, \ssucc_1)+(1-\lambda) C_{\sqK}(R^2, \ssucc_1)\\
		&>\lambda C_{\sqK}(R^1, \ssucc_2)+(1-\lambda) C_{\sqK}(R^2, \ssucc_2)=C_{\sqK}(\lambda R^1+(1-\lambda) R^2, \ssucc_2).
	\end{align*}
	Hence, the Squared Kemeny rule indeed satisfies participation.
	
	Finally, we turn to strategyproofness, and consider the following two profiles $R^1$ and $R^2$.
	\begin{center}
		\begin{tabular}{cc}
			$R^1$: \begin{tabular}{ccc}
				$\nicefrac{1}{3}$ & $\nicefrac{5}{9}$ &$\nicefrac{1}{9}$\\\hline
				$a$ & $b$ & $c$\\
				$b$ & $a$ & $a$\\
				$c$ & $c$ & $b$
			\end{tabular}\hspace{2cm}
			& 
			$R^2$: \begin{tabular}{cccc}
				$\nicefrac{1}{3}$ & $\nicefrac{4}{9}$ &$\nicefrac{1}{9}$ & $\nicefrac{1}{9}$ \\\hline
				$a$ & $b$ & $c$ & $c$ \\
				$b$ & $a$ & $a$ & $b$\\
				$c$ & $c$ & $b$ & $a$
			\end{tabular}
		\end{tabular}
	\end{center}
	It can be verified that Squared Kemeny uniquely chooses the ranking $\ssucc_1=a\mathrel{\ssucc_1}  b\mathrel{\ssucc_1}  c$ for $R^1$ and the ranking $\ssucc_2=b\mathrel{\ssucc_2}  a\mathrel{\ssucc_2}  c$ for $R^2$. Since the profile $R^2$ arises from the profile $R^1$ by moving probability $\nicefrac{1}{9}$ from ${\succ}=b\succ a\succ c$ to $c\succ' b \succ' a$ and $\ssucc_2 \succ \ssucc_1$, this shows that the Squared Kemeny rule is manipulable. This example also shows that Squared Kemeny fails the weaker ``betweenness'' version of strategyproofness of \citet{BosSpr-2014-SPrankingaggregation}, which Kemeny does satisfy.
\end{proof}

\section{Proportionality Guarantees}
\label{sec:proportionality}

\begin{figure}[t]
	\centering
	\input{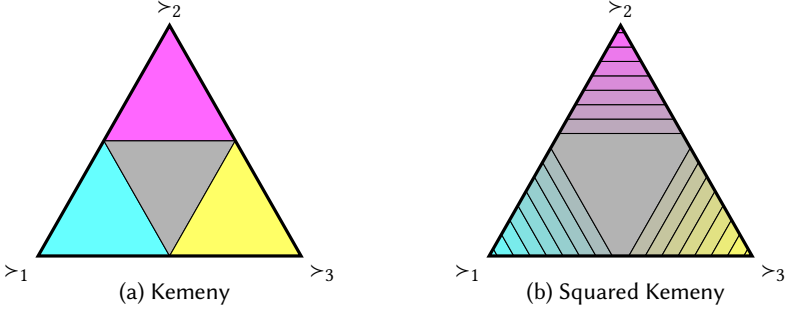}
	\vspace{-5pt}
	\caption{The simplex of profiles in which the rankings ${\succ_1} = abcde\textit{f}gh$, ${\succ_2} = \textit{f}edcbahg$, and ${\succ_3} = bahg\textit{f}dec$ occur. Each point of the simplex is colored according to the swap distance of the (a) Kemeny and (b) Squared Kemeny ranking to the input rankings, with a point's cyan (resp., magenta, yellow) component being more intense if the output ranking is closer to $\succ_1$ (resp., $\succ_2$, $\succ_3$). Both rules output $\rhd^* = ba\textit{f}edchg$ (which is equidistant to the three input rankings) when no input ranking has weight greater than $\frac12$. Otherwise, Kemeny outputs the majority ranking, while Squared Kemeny smoothly moves towards an input ranking as its weight increases.}
	\label{fig:simplex}
\end{figure}

In our axiomatic treatment, we have discussed the behavior of the Squared Kemeny rule on profiles with a lot of structure (single-crossing) and in particular those with only two rankings (2RP). Does Squared Kemeny retain its behavior as an average in general? This is what we will quantify in this section. As an initial matter, we can check what Squared Kemeny does on profiles in which \emph{three} rankings occur. For a fixed set of three rankings, we can use a simplex of weights to picture all profiles based on these rankings. For each point of the simplex (i.e., for each profile), we can compute the Kemeny and Squared Kemeny outcomes, and we can color the point to indicate how close the output rankings are to each of the input rankings. We show the result of this exercise in \Cref{fig:simplex}. This confirms our expectation that Squared Kemeny takes rankings with smaller weight into account, while Kemeny frequently ignores them.

For general profiles with any number of rankings, we can ask about the maximum (over all possible profiles) swap distance between the output ranking and an input ranking, as a function of its weight $\alpha \in [0,1]$. For an ideal proportional rule, there should be a roughly linear relationship between these. For Kemeny, the distance can be as large as $\binom{m}{2}$ (the maximum possible swap distance) when $\alpha < \frac12$, and it is $0$ when $\alpha > \frac12$. For Squared Kemeny, we can compute its behavior for fixed $m$ using a linear program (see \Cref{app:proportionality-linear-programs}) that searches for the worst-case profile, which yields the plot in \Cref{fig:intro-alpha-curve} shown in the introduction. That function is approximately linear, except for a ``hump'' for small $\alpha$, which indicates that Squared Kemeny can output the reverse of an input ranking which has weight as large as $\alpha \approx 17\%$. We do not have a satisfactory explanation for these humps (the profiles witnessing this worst-case behavior are very complicated, see \Cref{app:hump-profiles}), and we do not know how big the hump is as $m \to \infty$, but we do know it cannot exceed $\alpha = 25\%$ (by \Cref{thm:prop-guarantee} below). In addition to computing the exact worst-case behavior for fixed $m$, we can also prove a theoretical upper bound that works for any $m$, which bounds the distance between the Squared Kemeny ranking and a weight-$\alpha$ input ranking. This bound is also shown in \Cref{fig:intro-alpha-curve}. 

\begin{theorem}
	\label{thm:prop-guarantee-one-ranking}
	Let $R$ be a profile and let ${\succ^*} \in \mathcal{R}$ be a ranking with weight $R(\succ^*) = \alpha$. Then
	\[ 
		\swap(\succ^*, \rhd) \le \sqrt{\frac{1-\alpha}{\alpha}} \cdot \binom{m}{2}
		\quad \text{ for every }\: {\rhd} \in \sqK(R).
	\]
\end{theorem}
\begin{proof}
	Note that $C_{\sqK}(R, {\succ^*}) \le (1 - \alpha) \cdot \binom{m}{2}^2$ since $\binom{m}{2}$ is the maximum swap distance between two rankings, and an $\alpha$ fraction of the profile has swap distance $0$ to $\succ^*$. 
	Let ${\rhd} \in \sqK(R)$ be a ranking selected by Squared Kemeny and write $d = \swap(\succ^*, \rhd)$. 
	Then we have $C_{\sqK}(R, \rhd) \ge \alpha \cdot d^2$. Because $\rhd$ optimizes the Squared Kemeny cost, we have $C_{\sqK}(R, \rhd) \le C_{\sqK}(R, {\succ^*})$ and thus $\alpha \cdot d^2 \le (1 - \alpha) \binom{m}{2}^2$.
	Solving for $d$, we get $d \le \sqrt{(1-\alpha)/\alpha}\cdot\binom{m}{2}$, as required.
\end{proof}

The above bound makes sense for profiles with few rankings that are not very similar to each other. But in contexts with many rankings, some of which are similar to each other, it would be better to guarantee to \emph{groups} of rankings that the output ranking should not be too far away from them, on average. 
We will formalize this in a similar way to the work of \citet{skowron2022proportionalbinarydecisions}, by considering arbitrary groups of rankings, without making any cohesiveness assumptions (that would say that the rankings in a group must be similar to each other). Note that such a setup limits the guarantees we can give: for a group of size $\alpha = 1$ (i.e. all the rankings together), we cannot guarantee that the output agrees with the group on more than half the pairwise comparisons on average (consider for example a profile where one ranking and its reverse each have weight $\frac12$).

To state our result, given a profile $R$, we say that $S : \mathcal{R} \to [0, 1] \cap \mathbb{Q}$ is a \emph{subprofile} of $R$ if $S({\succ}) \le R({\succ})$ for all rankings ${\succ} \in \mathcal{R}$. The \emph{size} of $S$ is $\sum_{{\succ}\in\mathcal{R}} S({\succ})$. We can think of $S$ as a \emph{group} of voters, and its size as the fraction of the entire electorate that they make up. We now provide a bound on the average satisfaction of any group.

\begin{theorem}
	\label{thm:prop-guarantee}
	Let $R$ be a profile and let $S$ be a subprofile of $R$ with size $\alpha$. Then
	\[
		\frac{1}{\alpha} \sum_{{\succ} \in \mathcal{R}} S({\succ}) \cdot \swap(\succ, \rhd) \le \sqrt{\frac{1}{4\alpha}} \cdot \binom{m}{2} + o(m^{1.5})
		\quad \text{ for every }\: {\rhd} \in \sqK(R).
	\]
\end{theorem}
\begin{proof}
	Let $d_{\max} = \binom{m}{2}$ be the maximum swap distance between two rankings. Fix any ranking ${\succ} \in \mathcal{R}$.
	For each $i \in \{0,1,\dots,d_{\max}\}$, let $M_i$ be the number of rankings ${\rhd} \in \mathcal{R}$ with $\swap(\succ, \rhd) = i$.
	The values $M_0, M_1, \dots, M_{d_{\max}}$ are the \emph{Mahonian numbers} and it is known~\citep{Ben-2010-Mahonian} that 
	\[
		\sum_{i = 0}^{d_{\max}} M_i \cdot i^2 = m! \left( \frac{1}{4} d_{\max}^2 + \frac{2m^3 + 3m^2 - 5m}{72} \right). 
	\]
	
	For a fixed profile $R$, the average Squared Kemeny cost $C_\sqK(R, \rhd)$ over all rankings $\rhd$ is
	\begin{align*}
		\frac{1}{m!} \sum_{{\rhd} \in \mathcal{R}} \sum_{{\succ} \in \mathcal{R}} R(\succ) \cdot \swap(\succ, \rhd)^2
		= \frac{1}{m!} \sum_{{\succ} \in \mathcal{R}} \sum_{{\rhd} \in \mathcal{R}} R(\succ) \cdot \swap(\succ, \rhd)^2 
		=\frac{1}{m!} \cdot m! \left( \frac{1}{4} d_{\max}^2 + O(m^3) \right).
	\end{align*}
	Therefore, there exists a ranking ${\rhd} \in \mathcal{R}$ such that
	\(	
		C_\sqK(R, \rhd) \le d_{\max}^2 /4 + O(m^3)
	\).
	From minimality of Squared Kemeny, we know that this is true for any ${\rhd} \in \sqK(R)$. Since $S$ is a subprofile of $R$,
	\[
		\frac{1}{\alpha} \sum_{{\succ} \in \mathcal{R}} S(\succ) \cdot \swap(\succ, \rhd)^2
		\le \frac{1}{\alpha}  \sum_{{\succ} \in \mathcal{R}} R(\succ) \cdot \swap(\succ, \rhd)^2
		\le \frac{1}{\alpha} d_{\max}^2 /4 + O(m^3).
	\]
Finally, by Jensen's inequality we get that
	\(	
		\left(\sum_{{\succ} \in \mathcal{R}} S(\succ) \cdot \swap(\succ, \rhd) / \alpha \right)^2 \le d_{\max}^2 /(4\alpha) + O(m^3),
	\)
	since the square function is convex,
	which yields the thesis.\end{proof}

\begin{wrapfigure}{r}{0.5\textwidth}
	\scalebox{0.85}{%
		\input{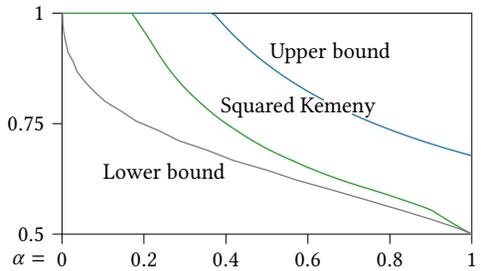}
	}
	\vspace{-5pt}
	\caption{Bound of \Cref{thm:prop-guarantee} for groups of rankings.}
	\label{fig:noncohesive-group-bound}
\end{wrapfigure}

In \Cref{fig:noncohesive-group-bound}, we show the upper bound obtained in \Cref{thm:prop-guarantee} in terms of normalized swap distance (so $\binom{m}{2}$ is mapped to 1). For $m = 6$, we also show the actual worst-case performance of Squared Kemeny, which can be computed for fixed $m$ using linear programs for finding worst-case profiles that maximize the distance between the output and some size-$\alpha$ group. The figure also shows a lower bound, which is obtained by linear programs that find profiles where all $m!$ output rankings are bad simultaneously in the sense that some size-$\alpha$ group incurs at least the lower bound's distance.

\section{Computation}
\label{sec:computation}

The computational complexity of the Kemeny rule has been extensively studied. The problem of deciding if there is a ranking with at most a given cost is NP-complete \citep{bartholdi1989voting}, even for a constant number of input rankings \citep{dwork2001rank,biedl2009complexity,bachmeier2019k}. Thus, it is reasonable to expect that the analogous problem for the Squared Kemeny rule is also NP-complete, and this is indeed the case, see \Cref{app:np-hard-proof}.

\begin{restatable}{theorem}{npcomplete}\label{thm:npcomplete}The problem of deciding, given a profile $R$ and a number $B$, whether there exists a ranking $\rhd$ with $C_{\sqK}(R, \rhd) \le B$, is NP-complete, even for profiles with 4 rankings with equal weight.
\end{restatable}

The proof is by reduction from the problem for the Kemeny rule, and uses the same technique used by \citet{biedl2009complexity} for showing that the egalitarian Kemeny rule (which selects the ranking where the maximum swap distance to any input ranking is minimized) is NP-complete to compute.

\begin{wrapfigure}{r}{0.4\textwidth}
	\centering
\begin{tikzpicture}[scale=0.65]
		
		\definecolor{cornflowerblue}{RGB}{100,149,237}
		\definecolor{darkblue}{RGB}{0,0,139}
		\definecolor{darkgray176}{RGB}{176,176,176}
		\definecolor{dimgray}{RGB}{105,105,105}
		\definecolor{gray}{RGB}{140,140,140}
		\definecolor{lightgray}{RGB}{160,160,160}
		\definecolor{lightgray204}{RGB}{180,180,180}
		
		\begin{axis}[
			legend cell align={left},
			legend style={
				fill opacity=0.8,
				draw opacity=1,
				text opacity=1,
				at={(0.03,0.97)},
				anchor=north west,
				draw=lightgray204
			},
			tick align=outside,
			tick pos=left,
			x grid style={darkgray176},
			xlabel={number of alternatives},
			xmin=6.5, xmax=83.5,
			xtick style={color=black},
			xtick distance=10,
			y grid style={darkgray176},
			ylabel={median running time (s.)},
			ymin=-30.8283993721008, ymax=647.880527448654,
			ytick style={color=black},
			ytick distance=120
			]
			\addplot [thick, cornflowerblue]
			table {10 0.0544687509536743
				20 0.657991886138916
				30 3.39515292644501
				40 12.1495265960693
				50 26.5022867918015
				60 115.724662661552
				70 186.585289597511
				80 617.030121684074
			};
			\addlegendentry{Sq. Kemeny, $n=5$}
			\addplot [thick, blue]
			table {10 0.0481027364730835
				20 0.54197096824646
				30 2.73315131664276
				40 8.70813143253326
				50 19.7152553796768
				60 68.6061573028564
				70 121.98901116848
				80 237.46488404274
			};
			\addlegendentry{Sq. Kemeny, $n=4$}
			\addplot [thick, darkblue]
			table {10 0.0385540723800659
				20 0.448656678199768
				30 2.13855218887329
				40 6.7280935049057
				50 15.7510339021683
				60 47.3511701822281
				70 91.8183234930038
				80 194.284425973892
			};
			\addlegendentry{Sq. Kemeny, $n=3$}
			\addplot [thick, lightgray]
			table {10 0.0236619710922241
				20 0.204893112182617
				30 0.722732186317444
				40 1.84698355197906
				50 3.98798882961273
				60 10.5086648464203
				70 16.7370330095291
				80 32.4897501468658
			};
			\addlegendentry{Kemeny, $n=5$}
			\addplot [thick, gray]
			table {10 0.0220063924789429
				20 0.202386379241943
				30 0.726245641708374
				40 1.85857021808624
				50 3.93119180202484
				60 8.65702593326569
				70 11.2142436504364
				80 17.4852221012115
			};
			\addlegendentry{Kemeny, $n=4$}
			\addplot [thick, dimgray]
			table {10 0.0265276432037354
				20 0.207897782325745
				30 0.731740713119507
				40 1.87427616119385
				50 4.09808743000031
				60 9.68661451339722
				70 12.4826674461365
				80 21.5191498994827
			};
			\addlegendentry{Kemeny, $n=3$}
		\end{axis}
		
	\end{tikzpicture}
	\vspace{-8pt}
	\caption{The median running time of computing the Squared Kemeny using Gurobi
		for a given number of alternatives and $n=\{3, 4, 5\}$ rankings
		occurring in the profile with equal weights, drawn uniformly at random.
		The values are based on 50 samples.}
	\label{fig:computation-time}
\end{wrapfigure}
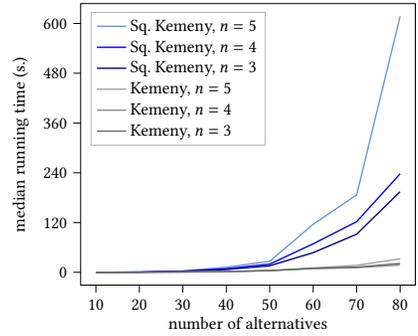

There exists an ILP formulation for computing the Kemeny rule, which is reasonably efficient in practice \citep{conitzer2006improved}. While this ILP formulation depends on the linear nature of the Kemeny objective, it is still possible to give an ILP formulation for the Squared Kemeny rule, using the same trick used by \citet{caragiannis2019unreasonable} for computing the maximum Nash welfare solution for fair allocation, and generalized by \citet{bredereck2020mixed} for various covering problems. The encoding is described in \Cref{app:ilp-formulation}. We found that it allows us to evaluate the Squared Kemeny rule reasonably efficiently up to $m = 80$ (see \Cref{fig:computation-time}).\footnote{We note that outliers (with up to $9$ times longer running time than the median) do occur.}

The Kemeny rule also admits efficient approximation algorithms \citep{ailon2008aggregating,van2009deterministic,coppersmith2010ordering} and even a PTAS \citep{KMSc07a}. The Squared Kemeny rule admits a simple 4-approximation algorithm (output the input ranking that has the best score, see \Cref{app:4-approximation}). In addition, we can show that the optimal Kemeny ranking provides a 2-approximation to the Squared Kemeny rule. Combining this with the known PTAS for Kemeny, we obtain the following result, proved in \Cref{app:approximation-proof}:

\begin{restatable}{theorem}{approximation}\label{thm:approximation}For every constant $\epsilon > 0$, there exists a polynomial-time $(2 + \epsilon)$-approximation to the Squared Kemeny rule.
\end{restatable}

We believe, however, that such approximation algorithms have limited interest for our applications, since a ranking may have a good approximation factor to the optimum Squared Kemeny score while not satisfying the desirable proportionality properties of the Squared Kemeny rule. Indeed, the rankings returned by Kemeny and Squared Kemeny may be very far apart from each other (in the extreme case, they may be opposite to each other except for 1 shared pairwise comparison, see \Cref{app:dist-kem-sqkem}), even though Kemeny provides a 2-approximation to Squared Kemeny. Still, approximation algorithms may have their use, for example, as subroutines in branch and bound algorithms. We leave the question of whether Squared Kemeny admits a PTAS to future work. 

\section{Empirical Analysis}
\label{sec:experiments}

In this section,
we compare the performance of the Squared Kemeny rule
and the Kemeny rule
based on several empirical experiments.

\definecolor{DublinColor}{RGB}{202, 255, 191}
\definecolor{ZurichColor}{RGB}{160, 196, 255}
\definecolor{DubaiColor}{RGB}{253, 255, 182}
\definecolor{TorontoColor}{RGB}{255, 173, 173}
\newcommand{\citybubble}[2]{\smash{\tikz[transform shape, baseline=(city.base)]{\node [fill=#1, rounded corners, inner ysep=1.5pt, inner xsep=4pt, font=\footnotesize] (city) {#2};
}}}

\begin{table}
	\centering
	\footnotesize
	\begin{tabular}{ccccccc}
		\toprule
		\textbf{Rank} & \textbf{GDP pc. (40\%)} & \textbf{Air Quality (30\%)} & \textbf{Sunniness (30\%)} &
		\textbf{Kemeny} & \textbf{Sq. Kemeny} & \\
		\midrule
		1. & San Francisco & New York City & \citybubble{DubaiColor}{Dubai} & San Francisco & San Francisco & \staysamerank{} \\
		2. & New York City & \citybubble{DublinColor}{Dublin} & Cairo & New York City & New York City & \staysamerank{} \\
		3. & \citybubble{ZurichColor}{Zurich} & \citybubble{TorontoColor}{Toronto} & Johannesburg & \citybubble{ZurichColor}{Zurich} & Sydney & \moveuprank{2} \\
		4. & \citybubble{DublinColor}{Dublin} & Buenos Aires & San Francisco & \citybubble{DublinColor}{Dublin} & \citybubble{TorontoColor}{Toronto} & \moveuprank{2} \\
		5. & Sydney & London & Lahore & Sydney & \citybubble{DubaiColor}{Dubai} & \moveuprank{6} \\
		\midrule
		6. & \citybubble{TorontoColor}{Toronto} & Sydney & Rome & \citybubble{TorontoColor}{Toronto} & Rome & \moveuprank{6} \\
		7. & London & Rio de Janeiro & Sydney & London & Johannesburg & \moveuprank{6} \\
		8. & Paris & San Francisco & Mumbai & Tokyo & Buenos Aires & \moveuprank{6} \\
		9. & Hong Kong & \citybubble{ZurichColor}{Zurich} & New York City & Paris & Rio de Janeiro & \moveuprank{6} \\
		10. & Tokyo & Tokyo & Mexico & Hong Kong & \citybubble{ZurichColor}{Zurich} & \movedownrank{7} \\
		\midrule
		11. & \citybubble{DubaiColor}{Dubai} & Rome & Buenos Aires & \citybubble{DubaiColor}{Dubai} & \citybubble{DublinColor}{Dublin} & \movedownrank{7} \\
		12. & Rome & Moscow & Bangkok & Rome & London & \movedownrank{5} \\
		13. & Seoul & Paris & Rio de Janeiro & Johannesburg & Tokyo & \movedownrank{5} \\
		14. & Moscow & Hong Kong & Istanbul & Buenos Aires & Hong Kong & \movedownrank{4} \\
		15. & Shanghai & Mexico & Seoul & Rio de Janeiro & Mexico & \moveuprank{3} \\
		\midrule
		16. & Johannesburg & Johannesburg & \citybubble{TorontoColor}{Toronto} & Seoul & Seoul & \staysamerank{} \\
		17. & Rio de Janeiro & Seoul & Tokyo & Moscow & Moscow & \staysamerank{} \\
		18. & Istanbul & Bangkok & Shanghai & Mexico & Paris & \movedownrank{9} \\
		19. & Buenos Aires & Istanbul & Lagos & Bangkok & Bangkok & \staysamerank{} \\
		20. & Mexico & Lagos & Hong Kong & Istanbul & Istanbul & \staysamerank{} \\
		\midrule
		21. & Bangkok & Shanghai & Moscow & Shanghai & Shanghai & \staysamerank{} \\
		22. & Mumbai & Cairo & Paris & Cairo & Cairo & \staysamerank{} \\
		23. & Cairo & \citybubble{DubaiColor}{Dubai} & \citybubble{ZurichColor}{Zurich} & Mumbai & Mumbai & \staysamerank{} \\
		24. & Lagos & Mumbai & London & Lagos & Lagos & \staysamerank{} \\
		25. & Lahore & Lahore & \citybubble{DublinColor}{Dublin} & Lahore & Lahore & \staysamerank{} \\
		\bottomrule
	\end{tabular}
	\vspace{5pt}
	\caption{The three input rankings of cities and the results of their aggregation using the Kemeny and Squared Kemeny rule. In the rightmost column we also report how many positions a city has moved in the Squared Kemeny ranking in comparison to the Kemeny ranking.}
	\label{tab:cities}
	\vspace{-20pt}
\end{table}

\subsection{Aggregate City Ranking}
\label{sec:experiments:city}
Our first experiment is a detailed example.
We selected $25$ cities around the world
and ranked them, based on their
Gross Domestic Product (GDP) per capita,
air quality (measured by average PM 2.5 concentration),
and sunniness (based on the average number of sunshine hours in a year).
The details of the sources used can be found in \Cref{app:cities}.
Then, we assigned
weight $40\%$ to GDP per capita,
$30\%$ to air quality, and
$30\%$ to sunniness.
We aggregated the rankings using the Kemeny and Squared Kemeny rules.
While Squared Kemeny selected a unique ranking,
Kemeny selected $4$ tied outputs, in which
a few cities differ by $1$ to $3$ positions.
When increasing the weight on the GDP ranking by a small amount, only
one of these outputs stays optimal, and we went with that output.
The aggregation results are presented in \Cref{tab:cities}.

Note that the top-$5$ of the Kemeny ranking
is identical to the top-$5$ of the GDP per capita ranking.
This includes Dublin and Zurich, which are both in the bottom-$5$ of the sunniness ranking.
This is exactly the behavior we would like to avoid in the proportional aggregation of rankings:
focusing mainly on one input ranking and disregarding another
with still significant weight assigned to it.

In contrast, both Dublin and Zurich do not appear in
the top-$5$ of the ranking selected by the Squared Kemeny rule.
Instead, its top-$5$ includes Toronto
(which is in the middle of the sunniness ranking and 
relatively high in both GDP per capita and air quality rankings)
and Dubai (which, although near the bottom of the air quality ranking,
is the first according to sunniness and $11$-th based on GDP per capita).
In this way, the top-$5$ of the Squared Kemeny rule output ranking
arguably offers more uniform representation of the highest ranked cities
across all input rankings.

\subsection{Drawing Embeddings of Rankings}
\label{sec:embeddings}
Next, we visualize 
how the rankings output by the Kemeny and Squared Kemeny rules
relate to the input rankings, using two methods of embedding rankings into 2-dimensional Euclidean space.

The first method is called \emph{map of preferences},
introduced by \citet{FalKacSorSzuWas-2023-DivAgrPol},
and starts by computing the swap distances between each pair of rankings present in a profile.
Then, we apply a classical multidimensional scaling algorithm~\citep{Tor-1952-MDS}
to put each ranking as a dot on a plane in such a way
that the Euclidean distances between the dots
reflect the swap distances between the rankings as well as possible.
By the size of a dot we signify the weight of a given ranking in the profile.
In order to obtain the coordinates for the outputs of the Kemeny and Squared Kemeny rules as well,
we simply include them while computing the distance matrix.

\Cref{fig:microscope} presents maps of preferences of six profiles
which we have generated using different models.
For each, we sampled 200 rankings, with possible repetitions, over 10 alternatives
and constructed the profile by assigning each ranking a weight
proportional to the number of times it was sampled.
We have also verified that the behavior of Kemeny and Squared Kemeny visible on these examples is consistent across different profiles generated in the same way.

The first two pictures present profiles drawn from the \emph{Euclidean model},
in which we sample 10 \emph{alternative points} and 200 \emph{voter points}
uniformly at random from the unit disc 
(the first picture)
or from its boundary, the unit circle (second picture).
Then, for every voter point we record a ranking of all alternatives
in order of increasing distance from the voter.
For the disc, the Kemeny and Squared Kemeny rules select rankings very close to each other.
However, for the circle, the difference is significant.
While Squared Kemeny chooses a ranking that is in the center, Kemeny outputs a ranking that is similar to one of the input rankings on a circle.
This is because for every ranking on the circle, we also have the reversed (or close to the reversed)  ranking on the opposite side of the circle.
Thus, all possible rankings have a similar average swap distance to the profile
and the smallest occurs on the part of the circle from which, by chance, we sampled more voter points, which is then chosen by Kemeny.
In contrast, since Squared Kemeny minimizes the average squared swap distance,
it tries to equalize the distances to all rankings present in the profile.

\begin{figure}
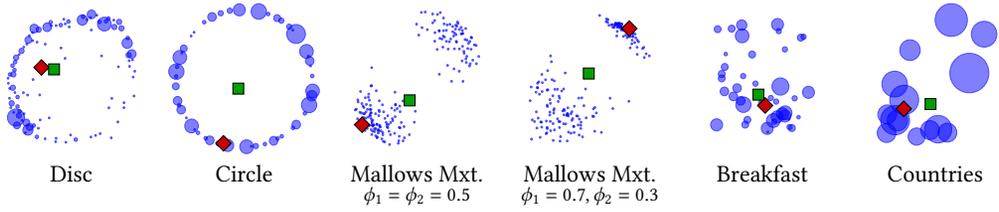

	\definecolor{VoteColor}{RGB}{0, 0, 255}
	\definecolor{KemenyColor}{RGB}{200, 0, 0}
	\definecolor{SqKemenyColor}{RGB}{0, 160, 0}
	\centering

	\caption{Maps of preferences. For the Kemeny rule the position of its output is denoted with a red diamond,
		and for Squared Kemeny with a green square.}
	\label{fig:microscope}
\end{figure}

The next two pictures show the maps of preferences for profiles generated from the mixture of two \emph{Mallows models}~\citep{Mal-1957-Mallows}.
Given a central ranking $\succ$ and a noise parameter $\phi \in [0,1]$, we sample each ranking $\succ'$ with probability proportional to $\phi^{\swap(\succ,\succ')}$ under the Mallows model.
Thus, for smaller $\phi$, the distribution is more concentrated around $\succ$.
We generated $55\%$ of the rankings using one Mallows model with the central ranking $\succ_1$ and noise $\phi_1$ and $45\%$ using another Mallows model with the central ranking $\succ_2$ that is the complete reverse of $\succ_1$ and noise $\phi_2$.
When $\phi_1=\phi_2=0.5$ (the third picture), 
the Kemeny rule outputs $\succ_1$, the central ranking of the Mallows model responsible for $55\%$ of the rankings
while the Squared Kemeny rule selects a ranking in between $\succ_1$ and $\succ_2$.
Interestingly, if $\phi_2$ is significantly smaller than $\phi_1$ (the fourth picture),
the Kemeny rule outputs the central ranking of the smaller but more concentrated model,
while Squared Kemeny is still between $\succ_1$ and $\succ_2$.
This confirms our intuition that Squared Kemeny works like an average of rankings.

Finally, the last two pictures present the profiles drawn from models based on real-world data from Preflib~\citep{MatWal-2013-Preflib}:
the breakfast dataset~\citep{GreRao-1972-Breakfast}
that contains 42 preference orders over 15 breakfast items;
and the 2016 countries ranking dataset~\citep{BoeSch-2023-RealWorld},
where 107 countries are ranked according to 14 different criteria.
For each dataset, we sampled with replacement 200 rankings and then restricted each ranking to 10 randomly chosen alternatives.
For both profiles, Kemeny and Squared Kemeny choose similar rankings, with the former closer to the most concentrated part, and the latter closer to the center of the picture.

\begin{figure}
	\centering
	\begin{tikzpicture}
		[euclid/.style={inner sep=0pt, draw=black!40}]
		\def\ls{\small} \def\w{2.75}
		\node[euclid]            at (0*\w,0) {\includegraphics[width=0.16\textwidth]{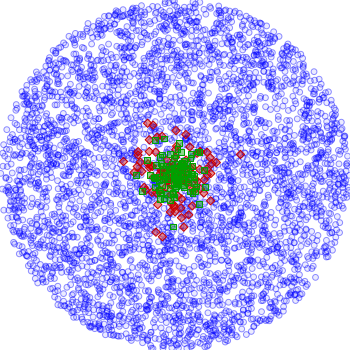}};
		\node[euclid]            at (1*\w,0) {\includegraphics[width=0.16\textwidth]{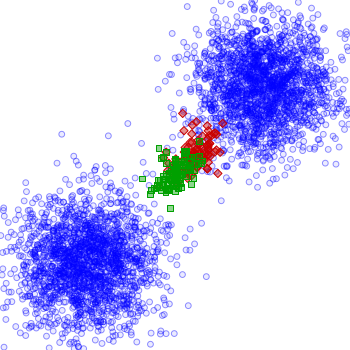}};
		\node[euclid]            at (2*\w,0) {\includegraphics[width=0.16\textwidth]{euclidean/experiment-two-gaussians-30-10.png}};
		\node[euclid]            at (3*\w,0) {\includegraphics[width=0.16\textwidth]{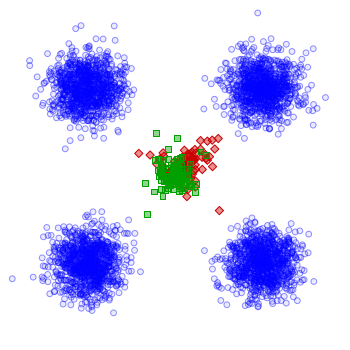}};
		\node[euclid, xscale=-1] at (4*\w,0) {\includegraphics[width=0.16\textwidth]{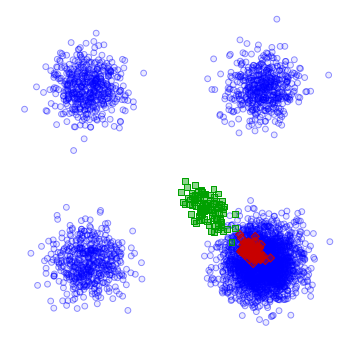}};
	\end{tikzpicture}
	\caption{Euclidean embeddings. For the Kemeny rule, the positions of its outputs are denoted with a red diamond,
		and for Squared Kemeny with a green square.}
	\label{fig:euclidean}
\end{figure}

The second method of visualizing the positions of rankings is specific to profiles drawn from the Euclidean model.
Given a profile specified by voter and alternative locations, and given an output ranking $\rhd$, we try to find a point in the same Euclidean space that would induce the ranking $\rhd$. In general, such a point may not exist, but using an ILP we can find a point that induces a ranking with minimal possible swap distance to $\rhd$. 

In our experiments shown in \Cref{fig:euclidean}, we sample $m = 10$ candidate locations uniformly from the unit square, and $n = 40$ voter locations according to different distributions of interest. We then compute the outputs of the Kemeny and Squared Kemeny rules and embed them as a point in space. For each voter location distribution, we repeat this process 100 times and superimpose the results for the 100 profiles in the same figure, showing voters as blue dots, Kemeny rankings as red diamonds, and Squared Kemeny rankings as green squares.

The first voter distribution samples the voter locations uniformly from the unit disc, and we see that both rules select central rankings, with somewhat more variance in the Kemeny rankings. For the second and third picture, we sample voter locations from two Gaussians centered in the bottom left and top right corners. In the second picture, we sample 20 voters each from the two Gaussians and see that both rules select rankings roughly midway between the center points of the Gaussians. In the third picture, we sample 30 voters from the bottom left and 10 voters from the top right. Kemeny selects rankings located within the bigger cluster, while Squared Kemeny chooses locations that interpolate between the two (while still being closer to the larger cluster). The fourth and fifth picture repeat the same process with four Gaussians with voters uniformly distributed (in the fourth picture) or with 25 voters in the bottom left and 5 voters each in the other three clusters.

\subsection{Worst-Case Average Distance}
\label{sec:experiments:group-distance}
Our final experiment revisits the problem studied in \Cref{sec:proportionality}, where we considered the average distance between the rankings of a group of voters and the output ranking. There, we bounded the distance for the worst-case profile, while here we compute it for randomly sampled profiles.
Consider a profile $R$ and
an output ranking $\rhd$.
For each size $\alpha \in [0,1]$,
we look for the ``unhappiest'' group $S$ (i.e., subprofile of $R$) of size $\alpha$,
in the sense that the average distance between $\rhd$ and the rankings of $S$ is large.
Formally, we define
\(
	\mu_{\alpha}(R, \rhd) = \max_{S \subseteq R: |S| = \alpha} \frac{1}{\alpha} \sum_{{\succ} \in \mathcal{R}} S(\succ) \cdot \swap(\succ, \rhd)  
\)
for this worst average distance, where $S \subseteq R$ denote that $S$ is a subprofile of $R$ with size $|S|= \alpha$.

\begin{wrapfigure}{r}{0.4\textwidth}
	\includegraphics[width=\linewidth]{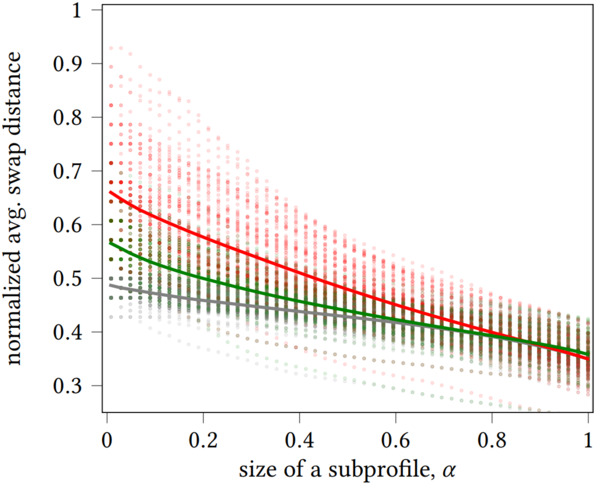}
	\vspace{-23pt}
	\caption{The maximal average distances between the subprofile of size $\alpha$ and the output of the Kemeny (red) and Squared Kemeny (green) rules, plus lower bound (gray).}
	\label{fig:avgdist}
\end{wrapfigure}

For the experiment, we sampled 100 profiles with 8 alternatives and 50 rankings according to various distributions.
The results of our experiment for the Euclidean disc model is presented in \Cref{fig:avgdist}. (See \Cref{app:avgdist} for other distributions.)
For each profile and each size
$\alpha \in \{\nicefrac{1}{50}, \nicefrac{2}{50}, \dots, 1\}$, we put a red dot indicating the value of $\mu_\alpha(R, \rhd)$ for Kemeny, and a green dot for Squared Kemeny. The lines show the average value for each $\alpha$.

\Cref{fig:avgdist} also shows a lower bound. For each $\alpha$, this is computed by finding the ranking $\rhd$ that optimizes $\mu_{\alpha}(R, \rhd)$, and placing a gray dot at that value. Note that different rankings may be optimum for different $\alpha$, and so this lower bound is ``unfair'' to rules like Kemeny or Squared Kemeny, which must choose a single ranking which gets evaluated for all $\alpha$ simultaneously.

We see that when $\alpha$ is close to $1$, Kemeny leads to smaller distances than Squared Kemeny.
This is to be expected, as Kemeny minimizes the overall average distance.
However, for $\alpha$ smaller than $0.8$, it is Squared Kemeny that returns the lower values on average.
Observe also that the difference between the Kemeny and Squared Kemeny for $\alpha$ close to $1$ is negligible,
while the differences for $\alpha$ close to $0$ are significant or even substantial depending on the considered model.

\section{Conclusions and Future Work}

We have studied the Squared Kemeny rule and argued that it behaves more appropriately in contexts where we want to aggregate rankings proportionally, compared to its better-known cousin the Kemeny rule. In particular, we have shown a full characterization of the Squared Kemeny rule based on a proportionality axiom, proved general proportionality guarantees for this rule, and demonstrated in an experimental study that it behaves similar to a mean. Based on these results, we conclude that the Squared Kemeny rule has the potential of providing a consensus ranking in situations where a majoritarian rule such as Kemeny is undesirable.

There are many interesting directions for future work exploring the topic of proportional rank aggregation. In particular, one could study new SPFs with the aim to find more proportional rules. For instance, one could consider rules based on the Spearman footrule distance instead of the swap distance \citep{diaconis1977spearman,viappiani2015characterization}, analogues of Proportional Approval Voting \citep{aziz2017jr}, or the family of ``$p$-Kemeny rules'' that minimize the $p$-th power of the swap distance. 
One could also derive other proportionality axioms that are not defined in terms of swap distance. 
For example, following \citet{skowron2017proportionalrankings} who study proportional rankings based on approval votes, one could phrase proportionality as requiring that every top-initial segment of the output ranking, viewed as a set, should be a proportional committee. 
How to adapt this to ranking input is not clear, though, since it is an open question whether axioms for proportional multi-winner rules (such as Proportionality for Solid Coalitions, PSC, \citealp{aziz2020expanding}) are compatible with \emph{committee monotonicity}, which is necessary to adapt a multi-winner rule to output a ranking.

Finally, we note that the methods that we have introduced may prove useful in other contexts. For example, we considered bounds on the maximum dissatisfaction of a voter, as a function of the voter's weight. Plotting and bounding these functions could provide insights in all kinds of collective decision-making problems. Further, our work can be seen as proportional decision-making on binary issues (``should $a$ be ranked above $b$?'') under constraints (in our case, transitivity). This general topic has just started to be explored by researchers \citep{skowron2022proportionalbinarydecisions,lackner2023proportional,chandak2024proportional,masak2023generalised}.

\subsection*{Acknowledgements}
This work was funded in part by the French government under management of Agence Nationale de la Recherche as part of the ``Investissements d’avenir'' program, reference ANR-19-P3IA-0001 (PRAIRIE 3IA Institute).

\bibliographystyle{ACM-Reference-Format}
\bibliography{references}

\appendix
\newpage

\addtocontents{toc}{\protect\setcounter{tocdepth}{1}}

\section{Proof of Theorem \ref{thm:characterization}}
\label{app:proof-characterization}

In this appendix, we will provide a full proof of \Cref{thm:characterization}. To this end, we first note that the direction from left to right is easy, especially since we have already shown that the Squared Kemeny rule satisfies 2RP. We thus focus on the converse and suppose for this that $f$ is an SPF that satisfies neutrality, continuity, reinforcement, and 2RP. To prove that $f$ is the Squared Kemeny rule, we will use an equivalent formulation of this rule by exchanging the minimum with a maximum in its definition: $\sqK(R)=\arg\max_{\ssucc\in\mathcal{R}} -\sum_{\succ\in\mathcal{R}} R(\succ) \swap(\succ, \ssucc)^2$. Then, our goal is to show that $f$ is the SPF that chooses the rankings that maximize $-\sum_{\succ\in\mathcal{R}} R(\succ) \swap(\succ, \ssucc)^2$. For this, we will use a hyperplane argument as, e.g., showcased by \citet{YouLev-1978-KemenyAxioms}. 

As a first step, we hence change the domain of $f$ from ranking profiles to a numerical space. To this end, let $b:\{1,\dots, |\mathcal{R}|\}\rightarrow \mathcal{R}$ denote an enumeration of all possible input rankings. Moreover, we define $\mathbb{T}=\{v\in \mathbb{Q}^{m!}\colon \sum_{i=1}^{m!} v_i=1 \land v_i\geq 0 \text{ for all } i\in \{1,\dots, m!\}\}$. Using the enumeration $b$, we can represent every profile $R$ as a vector $v\in\mathbb{T}$ by defining $v_i=R(b(i))$ for all $i\in \{1,\dots, m!\}$. For simplicity, we will denote the vector associated with a profile $R$ by $v(R)$. Moreover, for every permutation $\tau:A\rightarrow A$, we define the permutation of a vector $v$ by $\tau(v)_i=v_j$ for all $i,j$ such that $\tau(b(j))=b(i)$. That is, if the ranking $b(j)$ has weight $v_j$ in the profile $R$, then the ranking $\tau(b(j))=b(i)$ has weight $\tau(v)_i=v_j$ in the permuted profile $\tau(R)$. Hence, $\tau(v(R))=v(\tau(R))$.

By the definition of SPFs and profiles, it is straightforward that there is a function $g:\mathbb{T}\rightarrow2^{\mathcal{R}}\setminus \{\emptyset\}$ such that $f(R)=g(v(R))$ for all profiles $R$. Furthermore, $g$ inherits the desirable properties of $f$: 
\begin{itemize}
	\item $g$ satisfies neutrality: it holds  for all $v\in\mathbb{T}$ and all permutations $\tau:A\rightarrow A$ that $g(\tau(v))=\{\tau(\ssucc)\colon \ssucc\in g(v)\}$. 
	\item $g$ satisfies reinforcement: it holds for all $v,v'\in\mathbb{T}$ with $g(v)\cap g(v')\neq\emptyset$ and all $\lambda\in (0,1)\cap \mathbb{Q}$ that $g(\lambda v+(1-\lambda) v')=g(v)\cap g(v')$. 
	\item $g$ satisfies continuity: it holds for all $v,v'\in\mathbb{T}$ that there is a constant $\lambda\in (0,1)\cap \mathbb{Q}$ such that $g(\lambda v+(1-\lambda) v')\subseteq g(v)$. 
\end{itemize}

We next extend $g$ to the domain $\mathbb{Q}^{m!}$ while preserving its desirable properties.

\begin{lemma}\label{lem:domainextension}
	There is a neutral, reinforcing, and continuous function $\hat g:\mathbb{Q}^{m!}\rightarrow 2^{\mathcal{R}}\setminus\{\emptyset\}$ such that $f(R)=\hat g(v(R))$ for all profiles $R\in\mathcal{R}^*$. 
\end{lemma}
\begin{proof}
	As we observed before this lemma, there is a neutral, reinforcing, and continuous  function $g:\mathbb{T}\rightarrow 2^{\mathcal{R}}\setminus \{\emptyset\}$ such that $f(R)=g(v(R))$ for all $R\in\mathcal{R}^*$. We will prove this lemma by extending $g$ to $\mathbb{Q}^{m!}$. To this end, we will first extend the domain of $g$ to $\mathbb{T}^+=\{v\in\mathbb{Q}^{m!}\colon \sum_{i=1}^{m!} v_i>0 \land v_i \geq 0\text{ for all } i\in \{1,\dots, m!\}\}$ and then to $\mathbb{Q}^{m!}$.\medskip
	
	\textbf{Step 1: Extension to $\mathbb{T}^+$}
	
	For extending $g$ to the domain $\mathbb{T}^+$, we let $\bar g(v)=g(\lambda v)$ for all $v\in \mathbb{T}$, where $\lambda=\frac{1}{\sum_{i=1}^{m!} v_i}$ is the unique scalar such that $\sum_{i=1}^{m!} \lambda v_i=1$. First, we note that $\bar g$ is defined for all $v\in \mathbb{T}^+$ as the only difference between $\mathbb{T}$ and $\mathbb{T}^+$ is the assumption that $\sum_{i=1}^{m!} v_i=1$ for all $v\in \mathbb{T}$, whereas $\sum_{i=1}^{m!} v_i>0$ for all $v\in\mathbb{T}^+$. Moreover, we note that $f(R)=g(v(R))=\bar g(v(R))$ for all profiles $R$ as $\bar g(v)=g(v)$ for all profiles $R$. Next, we will show that $\bar g$ is neutral, continuous, and reinforcing. To this end, let $v^1\in\mathbb{T}^+$ denote an arbitrary vector and let $\lambda_1>0$ denote the scalar such that $\bar g(v^1)=g(\lambda_1 v^1)$. 
	
	For neutrality, we note that $\tau(\lambda_1 v^1)=\lambda_1 \tau(v^1)$ for every permutation $\tau:A\rightarrow A$, so it follows that $\bar g(\tau(v^1))= g(\lambda\tau(v^1))=g(\tau(\lambda_1v^1))=\{\tau(\ssucc)\colon \ssucc\in g(\lambda_1 v^1)\}=\{\tau(\ssucc)\colon \ssucc\in \bar g(v^1)\}$. Here, the third equality follows from the neutrality of $g$. This argument shows that $\bar g$ is neutral. 
	
	Next, we turn to reinforcement and hence let $v^2\in\mathbb{T}^+$ denote a second vector with $\bar g(v^1)\cap \bar g(v^2)\neq\emptyset$ and $\lambda_2$ the scalar such that $\bar g(v^2)=g(\lambda_2 v^2)$. 
	We need to show that $\bar g(\kappa v^1+(1-\kappa)v^2)=\bar g(v^1)\cap \bar g(v^2)$ for all $\kappa\in (0,1)\cap \mathbb{Q}$. 
	To this end, we fix such a $\kappa$ and first note that there is by definition a scalar $\lambda_3>0$ such that $\bar g(\kappa v^1+(1-\kappa)v^2)= g(\lambda_3(\kappa v^1+(1-\kappa)v^2))$. 
	Next, we observe that $\bar g(v^1)\cap \bar g(v^2)= g(\lambda_1 v^1)\cap \bar g(\lambda_2 v^2)$, so we can infer from the reinforcement of $g$ that $\bar g(v^1)\cap \bar g(v^2)= g(\kappa'\lambda_1 v^1+(1-\kappa')\lambda_2 v^2)$ for every $\kappa'\in (0,1)\cap \mathbb{Q}$. 
	To show that $\bar g$ is reinforcing, it hence suffices to find a $\kappa'\in (0,1)\cap \mathbb{Q}$ such that $\lambda_3(\kappa v^1+(1-\kappa)v^2)=\kappa'\lambda_1 v^1+(1-\kappa')\lambda_2 v^2$. 
	To this end, we note that $\lambda_3(\kappa v^1+(1-\kappa)v^2)=(\frac{\lambda_3}{\lambda_1} \cdot \kappa)\cdot \lambda_1 v^1 + (\frac{\lambda_3}{\lambda_2} \cdot (1-\kappa))\cdot \lambda_2 v^2$. 
	Now, since $\sum_{i=1}^{m!} \lambda_3(\kappa v^1_i+(1-\kappa)v^2_i)=\sum_{i=1}^{m!} \lambda_1 v^1_i=\sum_{i=1}^{m!} \lambda_1 v^2_i=1$, we can infer that $\frac{\lambda_3}{\lambda_1} \cdot \kappa+\frac{\lambda_3}{\lambda_2} \cdot (1-\kappa)=1$. 
	Moreover, because $\frac{\lambda_3}{\lambda_1}=\frac{\sum_{i=1}^{m!} v_i^1}{\sum_{i=1}^{m!} v_i^1+v_i^2}\in (0,1)\cap \mathbb{Q}$ and $\kappa\in (0,1)\cap \mathbb{Q}$, we can infer that both $\frac{\lambda_3}{\lambda_1} \cdot \kappa\in (0,1)\cap \mathbb{Q}$ and $\frac{\lambda_3}{\lambda_2} \cdot (1-\kappa)\in (0,1)\cap \mathbb{Q}$. 
	Hence, we now define $\kappa'=\frac{\lambda_3}{\lambda_1} \cdot \kappa$ (which implies that $1-\kappa'=\frac{\lambda_3}{\lambda_2} \cdot (1-\kappa)$). It then follows that $\lambda_3(\kappa v^1+(1-\kappa)v^2)=\kappa'\lambda_1 v^1+(1-\kappa')\lambda_2 v^2$, thus proving that $\bar g$ is reinforcing.
	
	Finally, we turn to the continuity of $\bar g$ and again consider a vector $v^2\in\mathbb{T}^+$ with its scalar $\lambda^2$. Our goal is to show that there is $\kappa\in (0,1)\cap\mathbb{Q}$ such that $\bar g(\kappa v^1+ (1-\kappa)v^2)\subseteq \bar g(v^1)$. 
	This is equivalent to finding a $\hat \kappa\in\mathbb{Q}$ such that $\hat\kappa>0$ and $\bar g(\hat\kappa v^1 +v^2)\subseteq \bar g(v^1)$. 
	To see this equivalence, we define $\kappa=\frac{\hat \kappa}{\hat \kappa+1}$, 
	$\hat \lambda_3=\frac{1}{\hat\kappa\sum_{i=1}^{m!} v^1_i + \sum_{i=1}^{m!} v^2_i}$, 
	and $\lambda_3=\frac{1}{\kappa\sum_{i=1}^{m!} v^1_i + (1-\kappa)\sum_{i=1}^{m!} v^2_i}$. 
	It holds that $\hat \lambda_3=\frac{1}{\hat \kappa+1}\lambda_3$, 
	so $\bar g(\hat \kappa v^1+v^2)=g(\hat \lambda_3 (\hat \kappa v^1+ v^2))=g(\lambda_3 (\kappa v^1+ (1-\kappa)v^2))=\bar g(\kappa v^1+ (1-\kappa) v^2)$. 
	Finally, to infer a suitable $\hat \kappa$, we note that $g$ itself is continuous. 
	Hence, there is a $\kappa'\in (0,1)\cap \mathbb{Q}$ such that $g(\kappa' \lambda_1 v^1+(1-\kappa')\lambda_2 v^2)\subseteq g(\lambda v^1)=\bar g(v^1)$. 
	We thus define $\hat \kappa$ by $\hat \kappa=\frac{\kappa'\lambda_1}{(1-\kappa') \lambda_2}$ because then $\bar g(\hat \kappa v^1+ v^2)=g(\kappa' \lambda_1 v^1+(1-\kappa')\lambda_2 v^2)\subseteq \bar g(v^1)$. 
	This implies that $\bar g$ is continuous.\medskip
	
	\textbf{Step 2: Extension to $\mathbb{Q}^{m!}$}
	
	Next, we will extend $\bar g$ to the domain $\mathbb{Q}^{m!}$. To this end, we define $v^*$ as the vector with $v_i^*=\frac{1}{m!}$ for all $i\in \{1,\dots, m!\}$. Since $\tau(v^*)=v^*$ for every permutation $\tau:A\rightarrow A$, it follows from the neutrality of $\bar g$ and $g$ that $\bar g(v^*)= g(v^*)=\mathcal{R}$. Now, to extend $\bar g$ to $\mathbb{Q}^{m!}$, we define $\hat g(v)=\bar g(v+\lambda v^*)$ for all $v\in\mathbb{Q}^{m!}$, where $\lambda\in\mathbb{Q}$ is a positive constant such that $v+\lambda v^*\in\mathbb{T}^+$.
	
	As the first point, we will show that $\hat g$ is well-defined despite the fact that we do not fully specify $\lambda$. To this end, let $v\in\mathbb{Q}^{m!}$ and consider two distinct positive constants $\lambda_1, \lambda_2$ such that $v+\lambda_1 v^*\in\mathbb{T}$ and $v+\lambda_2 v^*\in\mathbb{T}$. We need to prove that $\bar g(v+\lambda_1 v^*)=\bar g(v+\lambda_2 v^*)$. Note for this that $\bar g$ is by definition homogenous, i.e., it holds for every rational constant $\kappa>0$ that $\bar g(v')=\bar g (\kappa v')$ because $\bar g$ will rescale its input vector such that $\alpha \sum_{i=1}^{m!} v'_i=1$ and then apply $g$. Hence, we have that $\bar g(v+\lambda_i v^*)=\bar g(\frac{1}{2} v + \frac{1}{2} \lambda_i v^*)$ for $i\in \{1,2\}$. Now, if $\lambda_1<\lambda_2$, it follows that $\bar g(v+\lambda_1 v^*)=\bar g(v+\lambda_1 v^*)\cap \bar g((\lambda_2-\lambda_1)v^*)=\bar g(\frac{1}{2}v+\frac{1}{2}\lambda_1 v^*)\cap \bar g(\frac{1}{2}(\lambda_2-\lambda_1)v^*)=\bar g(\frac{1}{2}v+\frac{1}{2}\lambda_2 v^*)=\bar g(v+\lambda_2 v^*)$. The first equality here uses that $\bar g((\lambda_2-\lambda_1)v^*)=\mathcal{R}$, the second one that $\bar g$ is homogenous, and the third one that $\bar g$ is reinforcing. Since an analogous argument works if $\lambda_1>\lambda_2$, it now follows that $\hat g$ is well-defined. Moreover, we note that $f(R)=\bar g(v(R)+0v^*)=\hat g(v(R))$ for all profiles $R$.

	It remains to show that $\hat g$ is neutral, reinforcing, and continuous. For this, let $v^1$ denote a vector in $\mathbb{Q}^{m!}$ and let $\lambda_1$ denote a positive scalar in $\mathbb{Q}$ such that $v^1+\lambda_1v^*\in\mathbb{T}^+$. For neutrality, let $\tau$ additionally denote a permutation on $A$. Due to the neutrality of $\bar g$, it holds that $\hat g(\tau(v^1))=\bar g(\tau(v^1)+\lambda_1v^*)=\bar g(\tau(v^1+\lambda_1v^*))=\{\tau(\ssucc)\colon\ssucc\in \bar g(v^1+\lambda_1 v^*)\}=\{\tau(\ssucc)\colon\ssucc\in \hat g(v^1)\}$. This shows that $\bar g$ is neutral. 
	
	For reinforcement, let $v^2$ denote a vector in $\mathbb{Q}^{m!}$ and $\lambda_2$ a scalar such that $v^2+\lambda_2v^*\in\mathbb{T}^+$ and $\hat g(v^1)\cap\hat g(v^2)=\bar g(v^1+\lambda_1 v^*)\cap \bar g(v^2+\lambda_2 v^*)\neq\emptyset$. Since $\bar g$ is reinforcing, it follows for every $\kappa\in (0,1)\cap\mathbb{Q}$ that $\hat g(\kappa v^1+(1-\kappa v^2))=\bar g(\kappa (v^1+\lambda_1 v^*) + (1-\kappa) (v^2+\lambda_2 v^*))=\bar g(v^1+\lambda_1 v^*)\cap \bar g(v^2+\lambda_2 v^*)=\hat g(v^1)\cap \hat g(v^2)$. This proves that $\hat g$ is also reinforcing. 
	
	Finally, for continuity, we let $v^2\in\mathbb{Q}^{m!}$ again denote a second vector and $\lambda_2$ a corresponding scalar. We need to show that there is $\kappa\in (0,1)\cap\mathbb{Q}$ such that $\hat g(\kappa v^1+(1-\kappa)v^2)\subseteq \hat g(v^1)$. To this end, we note that there is $\kappa'\in (0,1)\cap \mathbb{Q}$ such that $\bar g(\kappa' (v^1+\lambda_1 v^*) + (1-\kappa') (v^2+\lambda_2 v^*))\subseteq \bar g(v^1+\lambda_1 v^*)=\hat g(v^1)$. It now follows from the definition of $\hat g$ that $\hat g(\kappa' v^1+(1-\kappa')v^2)=\bar g(\kappa' v^1+(1-\kappa')v^2+(\kappa' \lambda_1+(1-\kappa') \lambda_2) v^*)\subseteq \hat g(v^1)$. This shows that $\hat g$ is also continuous and hence completes the proof of this lemma. 
\end{proof}

We note that our SPF $f$ uniquely entails the function $\hat g$ and that $\hat g$ fully describes $f$. Hence, we aim to describe the function $\hat g$ based on a scoring function $s({\succ}, \ssucc)$. To this end, we first note that $\hat g$ is \emph{homogenous} (i.e., $\hat g(v)=\hat g(\lambda v)$ for every $v\in\mathbb{Q}^{m!}$ and $\lambda\in\mathbb{Q}$ with $\lambda>0$). To see this, we recall that $\bar g$ is by definition homogeneous and let $v\in\mathbb{Q}^{m!}$ denote an arbitrary vector and $\kappa\in\mathbb{Q}$ a positive scalar. Moreover, let $\lambda$ denote a scalar such that $v+\lambda v^*\in\mathbb{T}^+$. Then, the homogeneity of $\hat g$ follows as $\hat g(v)=\bar g(v+\lambda v^*)=\bar g(\kappa(v+\lambda v^*))=\hat g(\kappa v)$.

Next, we further modify the representation of $f$ by considering the sets $R_{\ssucc_i}=\{v\in\mathbb{Q}^{m!}\colon {\ssucc}_i\in \hat g(v)\}$ for all $\ssucc_i\in\mathcal{R}$. 
All sets $R_{\ssucc_i}$ are symmetric to each other (i.e., if $v\in R_{\ssucc_i}$, then $\tau(v)\in R_{\tau(\ssucc_i)}$) and $\mathbb{Q}$-convex (i.e, if $v,v'\in R_{\ssucc_i}$, then $\lambda v+(1-\lambda) v'\in R_{\ssucc_i}$ for all $\lambda\in (0,1)\cap \mathbb{Q}$) because $\hat g$ is neutral and reinforcing. Moreover, since the domain of $\hat g$ is $\mathbb{Q}^{m!}$, it follows that $\bigcup_{\ssucc_i\in\mathcal{R}} R_{\ssucc_i}=\mathbb{Q}^{m!}$. Further, we denote with $\bar R_{\ssucc_i}$ the closure of $R_{\ssucc_i}$ with respect to $\mathbb{R}^{m!}$. In particular, the sets $\bar R_{\ssucc_i}$ are convex and symmetric to each other and $\bigcup_{\ssucc_i\in\mathcal{R}} \bar R_{\ssucc_i}=\mathbb{R}^{m!}$ (see \citet{You-1975-ScoringFunctions}). As the last point, we note that $\hat g(v)=\{\ssucc_i\in\mathcal{R}\colon v\in R_{\ssucc_i}\}\subseteq \{\ssucc_i\in\mathcal{R}\colon v\in \bar R_{\ssucc_i}\}$ for all $v\in\mathbb{Q}^{m!}$. 

We next aim to find a suitable representation for the sets $\bar R_{\ssucc_i}$ and use for this the separating hyperplane theorem for convex sets. In the next lemmas, we write $uv=\sum_{i=1}^{k} u_i v_i$ for the standard scalar product between two vectors $u,v \in\mathbb{R}^k$.

\begin{lemma}\label{lem:hyperplanes}
    For all distinct rankings $\ssucc_i, \ssucc_j\in\mathcal{R}$, there is a non-zero vector $u^{\ssucc_i, \ssucc_j}\in\mathbb{R}^{m!}$ such that $vu^{\ssucc_i,\ssucc_j}\geq 0$ if $v\in \bar R_{\ssucc_i}$ and  $vu^{\ssucc_i,\ssucc_j}\leq 0$ if $v\in \bar R_{\ssucc_j}$.
\end{lemma}
\begin{proof}
    Consider two distinct rankings $\ssucc_i, \ssucc_j\in\mathcal{R}$ and their respective sets $\bar R_{\ssucc_i}$ and $\bar R_{\ssucc_j}$. First, we recall that $\bigcup_{\succ_k\in\mathcal{R}} \bar R_{\ssucc_k}=\mathbb{R}^{m!}$ and that all these sets are symmetric to each other. Since there are only a finitely many such sets, this implies that all $\bar R_{\ssucc_k}$ are fully dimensional. Consequently, $\text{int }\bar R_{\ssucc_i}\neq \emptyset$ and $\text{int }\bar R_{\ssucc_j}\neq \emptyset$. 
    
    We will next show that $\text{int }\bar R_{\ssucc_i}\cap\text{int }\bar R_{\ssucc_j}= \emptyset$. 
    Assume for contradiction that this is not the case. 
    Then, there is a vector $v\in \text{int }\bar R_{\ssucc_i}\cap \text{int }\bar R_{\ssucc_j}\cap \mathbb{Q}^{m!}$. In particular, these conditions entail that $\ssucc_i, \ssucc_j\in \hat g(v)$.
    Next, consider the profile $R$ such that $R(\ssucc_i)=\frac{m(m-1)}{m(m-1)+1}$ and $R(\ssucc_j)=\frac{1)}{m(m-1)+1}$, and let $v'=v(R)$ denote the corresponding vector. 
    By 2RP, we have that $f(R)=\hat g(v')=\{\ssucc_i\}$ because 
    \[\swap(\ssucc_i,\ssucc_j)\cdot\frac{ 1}{m(m-1)+1}\leq \frac{\mchoosetwo}{m(m-1)+1}<\frac{1}{2}.\] 
    The homogeneity of $\hat g$ then shows that $\hat g(\lambda v')=\{\ssucc_i\}$ for every $\lambda\in\mathbb{Q}$ with $\lambda>0$. 
    Moreover, reinforcement shows that $\hat g(v+\lambda v')=\hat g(\frac{1}{2}v+\frac{1}{2}\lambda v')=\hat g(v)\cap \hat g(\lambda v')=\{\ssucc_i\}$ for all $\lambda\in \mathbb{Q}$ with $\lambda>0$. However, $v\in \text{int }\bar R_{\ssucc_i}\cap \text{int }\bar R_{\ssucc_j}\cap \mathbb{Q}^{m!}$ implies that there is $\lambda\in\mathbb{Q}$ with $\lambda>0$ such that $\ssucc_i,\ssucc_j\in \hat g(v+\lambda v')$, contradicting our previous insight. Hence, our initial assumption must have been wrong and the sets $\text{int }\bar R_{\ssucc_i}$ and $\text{int }\bar R_{\ssucc_j}$ are indeed disjoint.

    Now, by the separating hyperplane theorem for convex sets, there is a non-zero vector $u^{\ssucc_i, \ssucc_j}$ such that $v u^{\ssucc_i, \ssucc_j}> 0$ if $v\in \text{int } \bar R_{\ssucc_i}$ and $v u^{\ssucc_i, \ssucc_j}< 0$ if $v\in \text{int } \bar R_{\ssucc_j}$. In particular, note that the constant on the right side of our inequalities must be $0$ as $\bar R_{\ssucc _i}$ and $\bar R_{\ssucc_j}$ are cones. This implies that $v u^{\ssucc_i, \ssucc_j}\geq 0$ for all $v\in \bar R_{\ssucc_i}$ and $v u^{\ssucc_i, \ssucc_j}\leq 0$ for all $v\in \bar R_{\ssucc_j}$. 
\end{proof}

We say that a non-zero vector $u$ \emph{separates} a set $\bar R_{\ssucc_i}$ from another set $\bar R_{\ssucc_j}$ if $vu\geq 0$ for all $v\in \bar R_{\ssucc_i}$ and $vu\leq 0$ for all $v\in \bar R_{\ssucc_j}$.
Our interest in these vectors comes from the next lemma which states that the vectors that separate $\bar R_{\ssucc_i}$ from any other set $\bar R_{\ssucc_j}$ fully describe the set $\bar R_{\ssucc_i}$. 

\begin{lemma}\label{lem:representation}
    Let $\ssucc_i\in\mathcal{R}$ denote an arbitrary ranking and let $u^{\ssucc_i, \ssucc_j}$ denote non-zero vectors that separates $\bar R_{\ssucc_i}$ from $\bar R_{\ssucc_j}$ for every $\ssucc_j\in\mathcal{R}\setminus \{\ssucc_i\}$. It holds that \[\bar R_{\ssucc_i}=\{v\in \mathbb{R}^{m!}\colon\forall \ssucc_j\in \mathcal{R}\setminus \{\ssucc_i\}\colon vu^{\ssucc_i, \ssucc_j}\geq 0\}.\]
\end{lemma}
\begin{proof}
    Fix a ranking $\ssucc_i$ and the vectors $u^{\ssucc_i, \ssucc_j}$ for all $\ssucc_j\in\mathcal{R}\setminus \{\ssucc_i\}$. For a simpler notation, we define $S_{\ssucc_i}=\{v\in \mathbb{R}^{m!}\colon\forall \ssucc_j\in \mathcal{R}\setminus \{\ssucc_i\}\colon vu^{\ssucc_i, \ssucc_j}\geq 0\}$, and we will show that $S_{\ssucc_i}\subseteq \bar R_{\ssucc_i}$ and $\bar R_{\ssucc_i}\subseteq S_{\ssucc_i}$. The second subset relation is straightforward by the definition of the vectors $u^{\ssucc_i, \ssucc_j}$: if $v\in \bar R_{\ssucc_i}$, then $vu^{\ssucc_i, \ssucc_j}\geq 0$ for all $\ssucc_j\neq \ssucc_i$ and therefore $v\in S_{\ssucc_i}$.

    For the converse direction, we first note that $\text{int } S_{\ssucc_i}$ is non-empty as $\bar R_{\ssucc_i}\subseteq S_{\ssucc_i}$ and $\text{int } \bar R_{\ssucc_i}\neq\emptyset$. Now, let $v\in \text{int } S_{\ssucc_i}$, which means that $vu^{\ssucc_i, \ssucc_j}>0$ for all $\ssucc_j\in \mathcal{R}\setminus \{\ssucc_i\}$. By the definition of the vectors $u^{\ssucc_i, \ssucc_j}$, it thus follows that $v\not\in \bar R_{\ssucc_j}$ for all $\ssucc_j\in\mathcal{R}\setminus \{\ssucc_i\}$. Since $\bigcup_{\ssucc_i\in\mathcal{R}} \bar R_{\ssucc_i}=\mathbb{R}^{m!}$, we derive that $v\in\bar R_{\ssucc_i}$, so $\text{int } S_{\ssucc_i}\subseteq \bar R_{\ssucc_i}$. Because $\bar R_{\ssucc_i}$ is a closed set, we finally conclude that $S_{\ssucc_i}\subseteq \bar R_{\ssucc_i}$, which completes the proof of this lemma. 
\end{proof}
 
\Cref{lem:representation} states that the vectors $u^{\ssucc_i, \ssucc_j}$ fully describe the sets $\bar R_{\ssucc_i}$ which, in turn, describe our function $\hat g$ because $\hat g(v)\subseteq \{\ssucc_i\in \mathcal{R}\colon v\in\bar R_{\ssucc_i}\}$ for all $v\in\mathbb{Q}^{m!}$ as $R_{\ssucc_i}\subseteq \bar R_{\ssucc_i}$. We show next that the subset relation between $\hat g(v)$ and $\{\ssucc_i\in \mathcal{R}\colon v\in\bar R_{\ssucc_i}\}$ is an equality for vectors in $\mathbb{Q}^{m!}$. 

\begin{lemma}\label{lem:continuity}
  It holds for all $v\in\mathbb{Q}^{m!}$ that $v\in \bar R_{\ssucc_i}$ if and only if $\ssucc_i\in \hat g(v)$.
\end{lemma}
\begin{proof}
    Consider an arbitrary vector $v\in\mathbb{Q}^{m!}$. 
    First, it immediately follows from the definition of the sets $\bar R_{\ssucc_i}$ that $v\in \bar R_{\ssucc_i}$ if $\ssucc_i\in \hat g(v)$. 
    For the other direction, suppose for contradiction that there is $v\in\mathbb{Q}^{m!}$ and $\ssucc_i\in\mathcal{R}$ such that $\ssucc_i\not\in \hat g(v)$ but $v\in\bar R_{\ssucc_i}$. 
    In this case, we consider the vectors $u^{\ssucc_i,\ssucc_j}$ that separate $\bar R_{\ssucc_i}$ from $\bar R_{\ssucc_j}$ for all $\ssucc_j\in\mathcal{R}$. 
    Due to the definitions of these vectors, we derive that $vu^{\ssucc_i,\ssucc_j}=0$ for all $\ssucc_j\in \hat g(v)$ because $v\in \bar R_{\ssucc_j}$ if $\ssucc_j\in \hat g(v)$. 
    Next, we note that there is a profile $R$ such that $\hat g(v(R))=f(R)=\{\ssucc_i\}$ because of 2RP. 
    By the continuity of $\hat g$, there is for every vector $v'\in\mathbb{Q}^{m!}$ a $\lambda\in(0,1)\cap\mathbb{Q}$ such that $\hat g(\lambda v(R)+(1-\lambda)v')=\{\ssucc_i\}$. 
    This shows that $v(R)\in \text{int } \bar R_{\ssucc_i}$. 
    By \Cref{lem:representation}, it thus follows that $v(R)u^{\ssucc_i,\ssucc_j}>0$ for all $\ssucc_j\in\mathcal{R}\setminus \{\ssucc_i\}$. 
    Now, using again continuity, there is $\lambda\in (0,1)\cap\mathbb{Q}$ such that $\hat g(\lambda v+(1-\lambda)v(R))\subseteq \hat g(v)$. However, if we consider the vectors $u^{\ssucc_i, \ssucc_j}$, we infer for all $\ssucc_j\in \hat g(v)$ that $v(\lambda v+(1-\lambda )v(R)) u^{\ssucc_i,\ssucc_j}=(1-\lambda)v(R) u^{\ssucc_i,\ssucc_j}>0$, which contradicts that $v\in \bar R_{\ssucc_j}$ for $\ssucc_j\in \hat g(v)$. Hence, our assumption that $v\in\bar R_{\ssucc_i}$ and $\ssucc_i\not\in g(v)$ is wrong.
\end{proof}

\Cref{lem:representation,lem:continuity} show that it suffices to determine the vectors $u^{\ssucc_i,\ssucc_j}$ to characterize $\hat g$ and thus $f$. We start the analysis of these vectors for rankings that only differ in a single swap. Recall for the next lemma that $b$ is the function that enumerates all possible rankings. We moreover note that the subsequent lemma is very similar to Step 1 in the proof of \Cref{thm:welfarist}.

\begin{lemma}\label{lem:swap_hyperplanes}
    Let $\ssucc_i,\ssucc_j\in\mathcal{R}$ denote two rankings such that $\swap(\ssucc_i,\ssucc_j)=1$. The vector $u$ defined by $u_k=-\swap(\ssucc_i, b(k))^2+\swap(\ssucc_j, b(k))^2$ for all $k\in \{1,\dots, m!\}$ separates $\bar R_{\ssucc_i}$ from $\bar R_{\ssucc_j}$.
\end{lemma}
\begin{proof}
    Consider two arbitrary rankings $\ssucc_i,\ssucc_j\in\mathcal{R}$ such that $\swap(\ssucc_i,\ssucc_j)=1$ and let $a,b$ denote the alternatives with $a\mathrel{\ssucc_i} b$ and $b\mathrel{\ssucc_j} a$. Moreover, let $u^{\ssucc_i,\ssucc_j}$ denote the non-zero vector that separates $\bar R_{\ssucc_i}$ from $\bar R_{\ssucc_j}$ given by \Cref{lem:hyperplanes}. Based on $u^{\ssucc_i,\ssucc_j}$, we will show that the vector $u$ defined in this lemma also separates $\bar R_{\ssucc_i}$ from $\bar R_{\ssucc_j}$.
    
    As a first step, we consider the profiles $R$ and $R'$ with $R(\ssucc_i)=R(\ssucc_j)=\frac{1}{2}$, 
    and $R'(\ssucc_1)=\frac{2}{3}$ and $R'(\ssucc_2)=\frac{1}{3}$. 
    2RP implies for the profile $R$ that $f(R)=\hat g(v(R))=\{\ssucc_i, \ssucc_j\}$, 
    so we conclude that $v(R)\in\bar R_{\ssucc_i}\cap \bar R_{\ssucc_j}$. 
    This means that $v(R) u^{\ssucc_i, \ssucc_j}=0$ and hence $u^{\ssucc_i, \ssucc_j}_i=-u^{\ssucc_i, \ssucc_j}_j$ 
    (we assume here that $b(i)=\ssucc_i$, and $b(j)=\ssucc_j$ for simplicity). 
    By contrast, 2RP requires for $R'$ that $f(R')=\hat g(v(R'))=\{\ssucc_i\}$. 
    Thus, \Cref{lem:continuity} entails that $v(R)$ is only contained in the set $\bar R_{\ssucc_i}$. 
    Moreover, based on continuity, it is easy to show that $v(R)$ is even in $\text{int }\bar R_{\ssucc_i}$. 
    Consequently, \Cref{lem:representation} implies that $v(R')u^{\ssucc_i,\ssucc_j}>0$. 
    Since $v(R')u^{\ssucc_i,\ssucc_j}=\frac{2}{3} u^{\ssucc_i, \ssucc_j}_i+\frac{1}{3} u^{\ssucc_i,\ssucc_j}_j$ and $u^{\ssucc_i, \ssucc_j}_i=-u^{\ssucc_i,\succ_j}_j$, 
    we can now conclude that $u^{\ssucc_i, \ssucc_j}_i>0$ and $u^{\ssucc_i, \ssucc_j}_j<0$. Because every rescaling of $u^{\ssucc_i, \ssucc_j}$ also separates $\bar R_{\ssucc_i}$ from $\bar R_{\ssucc_j}$, 
    we assume from now on that 
    
    \begin{align*}
    	u^{\ssucc_i, \ssucc_j}_i&=1=-\swap(\ssucc_i,\ssucc_i)^2+\swap(\ssucc_i,\ssucc_j)\qquad\qquad \text{and} \\
    	u^{\ssucc_i, \ssucc_j}_j&=-1=-\swap(\ssucc_j,\ssucc_i)^2+\swap(\ssucc_j,\ssucc_j).
   \end{align*} 
    
    Next, we consider an arbitrary index $k$ such that $b\succ a$ for the ranking ${\succ}=b(k)$. Moreover, we define $d=\swap(\succ,\ssucc_i)$ and consider the profile $R$ with $R({\succ})=\frac{1}{2d}$ and $R(\ssucc_i)=\frac{2d-1}{2d}$. 
    For this profile, it holds that $d\cdot (1-R(\succ))=d-\frac{1}{2}$ and $d\cdot (1-R(\ssucc_i))=\frac{1}{2}$. Hence, 2RP requires that 
    \begin{align*}
    	f(R) = \{ \ssucc\in\mathcal{R} : (\swap(\succ,\ssucc)=d \qquad\text{ and } \quad &\swap(\ssucc_i, \ssucc)=0) \\
    	\text{or } (\swap(\succ, \ssucc)=d-1 \:\text{ and } \quad &\swap(\ssucc_i, \ssucc)=1)\}.
    \end{align*}
    Thus, ${\ssucc_i} \in f(R)=\hat g(v(R))$ because $\swap(\ssucc_i, \succ)=d$ and $\swap(\ssucc_i, \ssucc_i)=0$, and ${\ssucc_j} \in f(R)$ because $\swap(\ssucc_j, \succ)=d - 1$ and $\swap(\ssucc_i, \ssucc_j)=1$. This means that $v(R)\in \bar R_{\ssucc_i}\cap \bar R_{\ssucc_j}$ and the definition of the vector $u^{\ssucc_i,\ssucc_j}$ therefore implies that $v(R) u^{\ssucc_i, \ssucc_j}=0$. We thus infer that
    \begin{equation}
	    \label{eq:entry-i-vs-k}
	    u^{\ssucc_i, \ssucc_j}_k=-(2d-1) u^{\ssucc_i,\ssucc_j}_i=-2(d-1)=-d^2+(d-1)^2=\swap(b(k), \ssucc_i)^2-\swap(b(k), \ssucc_j)^2. 
    \end{equation}

    Finally, we can infer the entries $u^{\ssucc_i, \ssucc_j}_\ell$ for every input ranking $ {\succ}=b(\ell)$ with $a \succ b$ based on a symmetric argument by exchanging the roles of $\ssucc_i$ and $\ssucc_j$. In more detail, we first compute $d=\swap(\succ, \ssucc_j)$, then consider the profile $R$ with $R(\succ)=\frac{1}{2d}$ and $R(\ssucc_j)=\frac{2d-1}{2d}$, and lastly determine the output for this profile based on 2RP. This approach yields that 
    \[u^{\ssucc_i, \ssucc_j}_\ell=-\swap(b(\ell), \ssucc_i)^2+\swap(b(\ell), \ssucc_j)^2.\] 

	This completes the proof as we have shown that the vector $u$ with $u_k=-\swap(b(k), \ssucc_i)^2+\swap(b(k), \ssucc_j))^2$ indeed separates $\bar R_{\ssucc_i}$ from $\bar R_{\ssucc_j}$. 
\end{proof}

For the sake of readability, we assume from now on that \[u^{\ssucc_i,\ssucc_j}_k=-\swap(b(k), \ssucc_i)^2+\swap(b(k), \ssucc_j)^2\] 
for all output rankings $\ssucc_i,\ssucc_j$ with $\swap(\ssucc_i,\ssucc_j)=1$ and all $k\in \{1,\dots, m!\}$. Put differently, this means that the vectors that separate adjacent rankings are described by the score function of the Squared Kemeny rule. 

We next aim to extend this insight to pairs of rankings $\ssucc_i,\ssucc_j$ with larger swap distance. To this end, we proceed as follows: first, we investigate some general consequences of 2RP for the vector $u^{\ssucc_i,\ssucc_j}$. Next, we start to investigate single-crossing swap sequences. Formally, a \emph{swap sequence} $\ssucc_1,\dots,\ssucc_k$ (of length $k$) is a sequence of rankings such that $\swap(\ssucc_i, \ssucc_{i+1})=1$ for all $i\in \{1,\dots,k-1\}$. Moreover, this sequence is \emph{single-crossing} if every pair of alternatives is swapped at most once. The motivation for these single-crossing swap sequences is \Cref{lem:swap_sequence}: we can find a single-crossing swap sequence from $\ssucc_i$ to $\ssucc_j$ of length $\swap(\ssucc_i,\ssucc_j)+1$ for any pair of $\ssucc_i$, $\ssucc_j$. Finally, we show in \Cref{lem:Kemenysquared,lem:linearIndependence,lem:SCNI} that $f$ omits for every single-crossing swap sequence a profile such that $f(R)$ contains precisely the rankings in our sequence. Based on this insight, we then show that $u^{\ssucc_i,\ssucc_j}$ can be represented as $\sum_{k=0}^{t-1} u^{\ssucc_{i_k},\ssucc_{i_{k+1}}}$ for a single-crossing swap sequence $\ssucc_{i_1},\dots,\ssucc_{i_t}$ from $\ssucc_i$ to $\ssucc_j$. 

We start this analysis by proving some consequences of 2RP for arbitrary pairs of rankings.

\begin{lemma}\label{lem:2RP_general}
    Consider two arbitrary rankings $\ssucc_i,\ssucc_j\in\mathcal{R}$ such that $\swap(\ssucc_i,\ssucc_j)\geq 1$. For all $b(k)\in\mathcal{R}$, it holds that 
    \begin{enumerate}
        \item $u^{\ssucc_i,\ssucc_j}_k=0$ if $\swap(b(k),\ssucc_i)=\swap(b(k),\ssucc_j)$,
        \item $u^{\ssucc_i,\ssucc_j}_k\geq 0$ if $\swap(b(k),\ssucc_i)<\swap(b(k),\ssucc_j)$,
        \item $u^{\ssucc_i,\ssucc_j}_k>0$ if $b(k)=\ssucc_i$. 
    \end{enumerate}
\end{lemma}
\begin{proof}
    Consider two arbitrary rankings $\ssucc_i,\ssucc_j\in\mathcal{R}$ with $\swap(\ssucc_i,\ssucc_j)\geq 1$ and let ${\succ}=b(k)$ denote an arbitrary input ranking. Moreover, we define $\bar \succ$ as the ranking that is completely inverse to $\succ$. We prove each of our three claims separately.\medskip

    \textbf{Claim 1)}: Assume that $\swap(\succ,\ssucc_i)=\swap(\succ,\ssucc_j)$. This implies also that \begin{align*}
    	\swap(\bar\succ,\ssucc_i)
    	&=\mchoosetwo-\swap(\succ,\ssucc_i)\\
    	&=\mchoosetwo-\swap(\succ,\ssucc_j)\\
    	&=\swap(\bar\succ,\ssucc_j).
    \end{align*} 
    
    Moreover, since $\swap(\succ,\ssucc_i)=\swap(\succ,\ssucc_j)$, we can conclude that ${\succ}\neq\ssucc_i$ and ${\succ}\neq\ssucc_j$. 
    Now, consider the profiles $R$ and $R'$ defined by 
    
\begin{center}
    \begin{tabular}{ll}
    	$R(\bar \succ)=\frac{ \swap(\succ, \ssucc_i)}{ \mchoosetwo} \qquad \qquad$ &  $R(\succ)=\frac{ \mchoosetwo- \swap(\succ,\ssucc_i)}{ \mchoosetwo}$\\
    	$R'(\bar\succ)=\frac{ \swap(\succ, \ssucc_i)}{ \mchoosetwo}-\epsilon \qquad \qquad$ & $R'(\succ)=\frac{ \mchoosetwo -  \swap(\succ, \ssucc_i) }{ \mchoosetwo} + \epsilon;$
    \end{tabular}
\end{center}
    
    $\epsilon>0$ denotes here a rational constant such that $\round(\mchoosetwo(1-R(\bar\succ)))=\{\mchoosetwo-\swap(\succ, \ssucc_i))\}$ and $\round(\mchoosetwo(1-R(\succ)))=\{\swap(\succ, \ssucc_i))\}$. 2RP thus implies that 
    \begin{align*}
    	f(R)=f(R')=\{\ssucc\in\mathcal{R}\colon \swap(\succ, \ssucc)=\swap(\succ, \ssucc_i)\text{  and  } \swap(\bar\succ, \ssucc)=\swap(\bar \succ, \ssucc_i)\}.
    \end{align*}
    
    Consequently, $\ssucc_i$ and $\ssucc_j$ are both chosen for $R$ and $R'$. Hence, $v(R), v(R')\in\bar R_{\ssucc_i}\cap\bar R_{\ssucc_j}$ and $v(R)u^{\ssucc_i, \ssucc_j}=v(R')u^{\ssucc_i, \ssucc_j}=0$ by the definitions of the sets $\bar R_{\ssucc}$ and the vector $u^{\ssucc_i, \ssucc_j}$. Finally, let $\lambda=\frac{R'(\bar\succ)}{R(\bar\succ)}$. It is now easy to verify that 
   \begin{align*}
   	0=(v(R')-\lambda v(R))u^{\ssucc_i, \ssucc_j}=(R'(\succ)-\lambda R(\succ)) u^{\ssucc_i, \ssucc_j}_k.
   \end{align*}
   
   Since $R'(\succ)>R(\succ)$ and $0<\lambda<1$, this implies that $u^{\ssucc_i, \ssucc_j}_k=0$, which completes the proof of this claim.
    \medskip

    \textbf{Claim 2)}: Next, assume that $\swap(\succ,\ssucc_i)<\swap(\succ,\ssucc_j)$. In this case, we consider the following two profiles $R$ and $R'$: 
    \begin{center}
    	\begin{tabular}{ll}
    		$R(\bar \succ)=\frac{ \swap(\succ, \ssucc_i)}{ \mchoosetwo} \qquad \qquad$ &  $R(\succ)=\frac{ \mchoosetwo- \swap(\succ,\ssucc_i)}{ \mchoosetwo}$\\
    		$R'(\bar\succ)=\frac{ \swap(\succ, \ssucc_j)}{ \mchoosetwo}\qquad \qquad$ & $R'(\succ)=\frac{ \mchoosetwo -  \swap(\succ, \ssucc_j) }{ \mchoosetwo}.$
    	\end{tabular}
    \end{center}
    
    Using 2RP, it can be verified that $\ssucc_i\in f(R)$ and $\ssucc_j\not\in f(R)$, and $\ssucc_i\in f(R')$ and $\ssucc_j\not\in f(R')$. Consequently, $v(R)\in \bar R_{\ssucc_i}$ and $v(R')\in\bar R_{\ssucc_j}$, which implies that $v(R)u^{\ssucc_i,\ssucc_j}\geq 0$ and $v(R') u^{\ssucc_i, \ssucc_j}\leq 0$. Finally, we define $\lambda=\frac{R(\bar \succ)}{R'(\bar\succ)}$ and note that $0\leq\lambda<1$. Since $-\lambda v(R')u^{\ssucc_i,\ssucc_j}\geq 0$, we derive that 
    \begin{align*}
    	0\leq (v(R)-\lambda v(R')) u^{\ssucc_i, \ssucc_j} = (R(\succ)-\lambda R'(\succ))u^{\ssucc_i, \ssucc_j}_k. 
    \end{align*}
    
   Since $R(\succ)>R'(\succ)$ (as $\swap(\succ, \ssucc_i)<\swap(\succ, \ssucc_j)$) and $0\leq\lambda<1$, this inequality shows that $u^{\ssucc_i, \ssucc_j}_k\geq 0$.\medskip

    \textbf{Claim 3)}: For our last claim, we assume that ${\succ}=\ssucc_i$, which means that $\swap(\succ,\ssucc_i)=0$ and $\swap(\succ,\ssucc_j)>0$. By Claim 2), we know that $u^{\ssucc_i,\ssucc_j}_k\geq 0$ and we only need to show that this inequality is strict. For doing so, consider the profile $R$ with $R(\succ)=1-\epsilon$ and $R(\bar\succ)=\epsilon$, where $\epsilon>0$ is so small that $\round(\mchoosetwo(1-R(\succ)))=\{0\}$. As a consequence, 2RP implies that $f(R)=\{\ssucc_i\}$. By \Cref{lem:continuity}, this means that $v(R)\in\bar R_{\ssucc_i}$ and $v(R)\not\in\bar R_{\ssucc'}$ for all $\ssucc'\in \mathcal{R}\setminus \{\ssucc_i\}$. \Cref{lem:representation}, in turn, implies that $v(R) u^{\ssucc_i, \ssucc_j}\geq 0$ and that there is for every $\ssucc_k\in\mathcal{R}\setminus \{\ssucc_i\}$ another ranking $\phi(\ssucc_k)$ such that $v(R)u^{\ssucc_k, \phi(\ssucc_k)}<0$. 
    
    We will next show that $u^{\ssucc_i, \ssucc_j}>0$ and assume thus for contradiction that $u^{\ssucc_i, \ssucc_j}=0$. Moreover, let $v'$ denote a vector such that $v'u^{\ssucc_i, \ssucc_j}<0$; such a vector exists as $u^{\ssucc_i, \ssucc_j}$ is non-zero. We then define $v^*=v(R)+\delta v'$, where $\delta>0$ is so small that $v(R) u^{\ssucc_1,\ssucc_2}<0$ implies that $v^* u^{\ssucc_1, \ssucc_2}<0$ for all $\ssucc_1,\ssucc_2\in\mathcal{R}$. In particular, this means that $v^*u^{\ssucc_k, \phi(\ssucc_k)}<0$ for all $\ssucc_k\in\mathcal{R}\setminus\{\ssucc_i\}$, so $v^*\not\in\bar R_{\ssucc_k}$ for these rankings. Moreover, $v^*u^{\ssucc_i,\ssucc_j}=\delta v'u^{\ssucc_i,\ssucc_j}<0$, so $v'\not\in\bar R_{\ssucc_i}$. However, this contradicts that $\bigcup_{\ssucc_k\in\mathcal{R}}\bar R_{\ssucc_k}=\mathbb{R}^{m!}$, so it must hold that $v(R)u^{\ssucc_i,\ssucc_j}>0$.

    As the last point, we observe that $v(R)u^{\ssucc_i,\ssucc_j}>0$ is only possible if $u^{\ssucc_i,\ssucc_j}_k>0$ or $u^{\ssucc_i,\ssucc_j}_{k'}>0$ ($k'$ denotes here the index such that $\bar\succ=b(k')$). Now, assume for contradiction that $u^{\ssucc_i,\ssucc_j}_k\leq 0$, so $u^{\ssucc_i,\ssucc_j}_{k'}>0$. By Claim 2), we know that $u^{\ssucc_i,\ssucc_j}_k\geq 0$ and thus $u^{\ssucc_i,\ssucc_j}_k=0$. Finally, consider the profile $R'$ in the proof of Claim 2) with $\ssucc_j\in f(R')$. In particular, this means that $v(R')\in\bar R_{\ssucc_j}$. However, if $u^{\ssucc_i,\ssucc_j}_k=0$ and $u^{\ssucc_i,\ssucc_j}_{k'}>0$, then $v(R')u^{\ssucc_i,\ssucc_j}>0$ which contradicts that $v(R')\in \bar R_{\ssucc_j}$. Hence, the assumption that $u^{\ssucc_i,\ssucc_j}_k\leq 0$ is wrong and our third claim follows.
\end{proof}

As explained before, we now turn our focus to swap sequences. For the sake of completeness, we will next show that for every pair of rankings $\ssucc_i,\ssucc_j$, there is a swap sequence from $\ssucc_i$ to $\ssucc_j$ of length $\swap(\ssucc_i,\ssucc_j)+1$. 

\begin{lemma}\label{lem:swap_sequence}
    Let $\ssucc_i,\ssucc_j\in\mathcal{R}$ denote two rankings such that $\swap(\ssucc_i,\ssucc_j)=k$ for some $k\in\mathbb{N}$. There is a swap sequence $\hat \ssucc_1,\dots,\hat \ssucc_{k+1}$ of length $k+1$ such that $\hat \ssucc_1=\ssucc_i$ and $\hat \ssucc_{k+1}=\ssucc_j$.
\end{lemma}
\begin{proof}
    We prove the claim by induction over the swap distance $\swap(\ssucc_i,\ssucc_j)=k$ for our two considered rankings $\ssucc_i,\ssucc_j$. For the induction basis, assume that $k=1$. This means that there is a single pair of alternatives $a,b$ such that $a\mathrel{\ssucc_i} b$ and $b\mathrel{\ssucc_j} a$. This is only possible if we can transform $\ssucc_i$ to $\ssucc_j$ by simply swapping $a$ and $b$ and thus, $\ssucc_i,\ssucc_j$ forms our swap sequence of length $k+1$. 
    
    Now, suppose that we can construct a swap sequence of length $k$ between any two rankings with $\ssucc_i,\ssucc_j$ with $\swap(\ssucc_i,\ssucc_j)=k-1$. We will prove that the same holds for rankings $\ssucc_i,\ssucc_j$ with $\swap(\ssucc_i,\ssucc_j)=k$. To this end, let $D=\{(a,b)\in {\ssucc_i}\setminus {\ssucc_j}\}$ and note that $|D|=k$ by definition. For proving the lemma, let $x_k$ denote the $k$-th best alternative in $\ssucc_i$ and $y_k$ the $k$-th best alternative in $\ssucc_j$ for every $k\in \{1,\dots,m\}$. Since $\ssucc_i\neq\ssucc_j$, there is an integer $k$ such that $x_k\neq y_k$ and we let $k^*$ denote the largest such integer. Put differently, this means that $x_k=y_k$ for all $k>k^*$. Consequently, there is an integer $\ell<k^*$ such that $x_\ell=y_{k^*}$. We claim that $(x_\ell, x_{\ell+1})\in D$. If this was not the case, then $x_\ell\mathrel{\ssucc_j} x_{\ell+1}$, so there is an index $k'>k^*$ such that $x_{\ell+1}=y_{k'}$. However, this contradicts that $x_{k'}=y_{k'}$, so $(x_\ell, x_{\ell+1})\in D$ is true. 

    Next, let $\ssucc_i'$ denote the ranking derived from $\ssucc_i$ by swapping $x_{\ell}$ and $x_{\ell+1}$. It holds that $\swap(\ssucc_i',\ssucc_j)=k-1$, so there is a swap sequence $\hat \ssucc_1,\dots, \hat \ssucc_k$ of length $k$ from $\ssucc_i'$ to $\ssucc_j$ by the induction hypothesis. Finally, $\ssucc_i, \hat \ssucc_1,\dots, \hat \ssucc_k$ is then a swap sequence of length $k+1$ connecting $\ssucc_i$ and $\ssucc_j$.
\end{proof}

We note that the swap sequences constructed in \Cref{lem:swap_sequence} are minimal: for any pair of rankings $\ssucc_i,\ssucc_j$ with $\swap(\ssucc_i,\ssucc_j)=k$, there cannot be a swap sequence of length less than $k+1$ that connects these two rankings. In particular, this means that the constructed sequences are single-crossing.

We next aim to show that $f$ admits for every single-crossing swap sequence $\ssucc_{i_0},\dots, \ssucc_{i_t}$ a profile $R$ such that $f(R)=\{\ssucc_{i_0},\dots, \ssucc_{i_t}\}$. To this end, we analyze the linear independence of the vectors $u^{\ssucc_{i_{k}}, \ssucc_{i_{k+1}}}$ for $k\in \{0,\dots, t-1\}$. Since all these vectors separate rankings that differ only in a single swap, \Cref{lem:swap_hyperplanes} applies and shows that they can be described by the scores assigned by the Squared Kemeny rule. Hence, we will first analyze the Squared Kemeny rule in more detail. In particular, we will show that the Squared Kemeny rule admits for every single-crossing swap sequence $\ssucc_0,\dots,\ssucc_t$ a profile $R$ such that $\sqK(R)=\{\ssucc_0,\dots,\ssucc_t\}$.

\begin{lemma}\label{lem:Kemenysquared}
	Let $\ssucc_0, \dots, \ssucc_t$ denote a single-crossing swap sequence. There is a profile $R$ such that $\sqK(R)=\{\ssucc_0,\dots,\ssucc_t\}$.
\end{lemma}
\begin{proof}
	We will first show the lemma for a single-crossing swap sequence $\ssucc_0,\dots,\ssucc_t$ of length $t+1=\mchoosetwo+1$. To this end, we introduce some auxiliary notation: for every ranking $\succ$, we define $R^{\succ}$ as the profile with $R^{\succ}(\succ)=1$ and $R^{-\succ}$ as the profile with $R^{-\succ}(\succ')=\frac{1}{m!-1}$ for all $\succ'\in\mathcal{R}\setminus \{\succ\}$. Moreover, we recall that $C_{\sqK}(R, \ssucc)=\sum_{{\succ}\in\mathcal{R}} R(\succ)\swap(\succ, \ssucc)^2$ denotes the Squared Kemeny cost of a ranking $\ssucc$ in the profile $R$. It immediately follows that $C_{\sqK}(R^{\succ}, \ssucc)=\swap(\succ, \ssucc)^2$ for all $\succ,\ssucc\in\mathcal{R}$. Next, for determining $C_{\sqK}(R^{-\succ}, \ssucc)$, we first consider the profile $R$ with $R(\succ)=\frac{1}{m!}$ for all ${\succ}\in\mathcal{R}$. Due to the symmetry of this profile, it holds that $C_{\sqK}(R,\ssucc)=C_{\sqK}(R, \ssucc')$ for all $\ssucc \in\mathcal{R}$ and we thus define $c=m!\cdot C_{\sqK}(R,\ssucc)=\sum_{{\succ}\in\mathcal{R}} \swap(\succ, \ssucc)^2$. We then compute that $(m!-1) C_{\sqK}(R^{-\succ}, \ssucc)=\sum_{{\succ'}\in\mathcal{R}} \swap(\succ', \ssucc)^2 - \swap(\succ, \ssucc)^2= c - \swap(\succ, \ssucc)^2$ for all $\succ, \ssucc\in\mathcal{R}$. 
	
	Based on these insights, we now define the profile $R^*$ as convex combination of profiles $R^{\succ}$ and $R^{-\succ}$ for ${\succ}\in \{\ssucc_0, \dots, \ssucc_t\}$. In more detail, if $t$ is odd, then \[R^*=\frac{1}{Z} \big((m!-1)\frac{t-1}{2} (R^{-\ssucc_0}+R^{-\ssucc_t})+\sum_{k=1}^{t-1} R^{\ssucc_k}\big),\]
	and if $t$ is even, then 
	\[R^*=\frac{1}{Z} \big((m!-1)\frac{t}{2} (R^{-\ssucc_0}+R^{-\ssucc_t})+R^{\ssucc_{t/2}}+\sum_{k=1}^{t-1} R^{\ssucc_k}\big).\]
	
	For both cases, $Z$ denotes a normalization constant such that $\sum_{\succ\in\mathcal{R}} R^*(\succ)=1$. Based on our previous insights, we can compute for all rankings $\ssucc$ that  
	\begin{align*}
		Z\cdot C_{\sqK}(R^*, \ssucc) =
		(t-1)c - \frac{t-1}{2}(\swap(\ssucc_0, \ssucc)^2+\swap(\ssucc_t, \ssucc)^2) + \sum_{k=1}^{t-1} \swap(\ssucc_k, \ssucc)^2
	\end{align*}
	if $t$ is odd, and 
	\begin{align*}
		Z\cdot C_{\sqK}(R^*, \ssucc) =
		tc - \frac{t}{2} (\swap(\ssucc_0, \ssucc)^2+\swap(\ssucc_t, \ssucc)^2) + \swap(\ssucc_{t/2}, \ssucc)^2+ \sum_{k=1}^{t-1} \swap(\ssucc_k, \ssucc)^2
	\end{align*}
	if $t$ is even. 
	
	We will next show that $\sqK(R^*)=\{\ssucc_0, \dots, \ssucc_t\}$. As a first point, we will show that $f(R^*)\subseteq \{\ssucc_0, \dots, \ssucc_t \}$. For this, let $\ssucc$ denote a ranking that is not on our swap sequence. In particular, this means that $d=\swap(\ssucc_0, \ssucc)>0$ and $\swap(\ssucc_t,\ssucc)>0$. Moreover, since our swap sequence $\ssucc_0,\dots, \ssucc_t$ has maximal length, we know that $\ssucc_0$ and $\ssucc_t$ are completely inverse. This implies that $\swap(\ssucc_t, \ssucc)=\mchoosetwo-d$, so $d<\mchoosetwo$. Now, let $\ssucc_d$ denote the ranking in our swap sequence such that $\swap(\ssucc_0,\ssucc_d)=d$ and $\swap(\ssucc_t, \ssucc_d)=\mchoosetwo-d$; we will show that $Z\cdot C_{\sqK}(R^*, \ssucc_d)<Z\cdot C_{\sqK}(R^*, \ssucc)$. We hence note that $\swap(\ssucc_0, \ssucc_d)=\swap(\ssucc_0, \ssucc)$ and $\swap(\ssucc_t, \ssucc_d)=\swap(\ssucc_t, \ssucc)$, so the cost caused by $\ssucc_0$ and $\ssucc_t$ does not matter for comparing these rankings. Moreover, it is obvious that $\swap(\ssucc_d,\ssucc_d)=0<\swap(\ssucc_d,\ssucc)$. Further, consider an arbitrary ranking $\ssucc_k$ in our swap sequence with $k\in \{1,\dots, d-1\}$. It holds that $\swap(\ssucc_k, \ssucc)+\swap(\ssucc_k, \ssucc_0)\geq \swap(\ssucc_0, \ssucc)$, so we can infer that $\swap(\ssucc_k, \ssucc)\geq d-k=\swap(\ssucc_k, \ssucc_d)$. Since a symmetric argument holds for all $\ssucc_k$ with $k\in \{d+1, \dots, t-1\}$, we can now conclude that $\sum_{k=1}^{t-1} \swap(\ssucc_k, \ssucc)^2>\sum_{k=1}^{t-1} \swap(\ssucc_k, \ssucc_d)^2$, which implies that $\ssucc\not\in \sqK(R)$. 

    Next, we need to show that $C_{\sqK}(R, \ssucc)=C_{\sqK}(R, \ssucc')$ for all $\ssucc,\ssucc'\in\{\ssucc_0,\dots,\ssucc_t\}$. For this, we define $\ell=\frac{t}{2}-1$ if $t$ is even and $\ell=\frac{t-3}{2}$ if $t$ is odd. Moreover, we consider the auxiliary profiles $R^i$ (for $i\in \{0, \dots, \ell\}$) defined by 
    
    \[R^i=\frac{1}{Z^i} \big((m!-1)(R^{-\ssucc_{i}} + R^{-\ssucc_{t-i}})+R^{\ssucc_{i+1}} + R^{\ssucc_{t-(i+1)}}\big),
    \]
    
    where $Z^i$ is again a normalization constant. Now, we recall that $(m!-1)R^{-\ssucc}+R^{\ssucc}=m!R$ for all $\ssucc$, where $R$ is the profile in which every ranking has weight $\frac{1}{m!}$. Hence, we infer that
    \begin{align*}
    \sum_{i=0}^{\ell} (\ell+1-i)\cdot Z^i\cdot R^{i}
    &=\sum_{i=0}^{\ell} (\ell+1-i) \cdot \big((m!-1)(R^{-\ssucc_{i}} + R^{-\ssucc_{t-i}})+R^{\ssucc_{i+1}} + R^{\ssucc_{t-(i+1)}}\big)\\
    &=\sum_{i=1}^{\ell} (\ell+1-i)(m!-1)(R^{-i})(R^{-\ssucc_i}+R^{-\ssucc_{t-i}})+(\ell+2-i)(R^{\ssucc_i}+R^{\ssucc_{t-i}})\\
    &\quad+(\ell+1)(m!-1)(R^{-\ssucc_0} + R^{-\ssucc_t})\\
    &=\sum_{i=1}^{\ell} R^{\ssucc_i}+R^{\ssucc_{t-i}} + 2(\ell+1-i)m! R +(\ell+1)(m!-1)(R^{-\ssucc_0} + R^{-\ssucc_t})\\
    &=Z R^* + \ell' m! R.
    \end{align*}
    
    Here, we define $\ell'=2\sum_{i=1}^\ell (\ell+1-i)$ for simplicity.  Since $C_{\sqK}(R, \ssucc)=C_{\sqK}(R,\ssucc')$ for all rankings $\ssucc, \ssucc'$, we infer from this equation that $C_{\sqK}(R^*,\ssucc)=C_{\sqK}(R^*,\ssucc')$ if and only if $C_{\sqK}(\sum_{i=0}^{\ell} (\ell+1-i)\cdot Z^i\cdot R^{i},\ssucc)=C_{\sqK}(\sum_{i=0}^{\ell} (\ell+1-i)\cdot Z^i\cdot R^{i},\ssucc')$ for all rankings $\ssucc, \ssucc'$. In turn, the latter equality holds if $C_{\sqK}(R^i,\ssucc)=C_{\sqK}(R^i,\ssucc')$ for all $i\in \{0,\dots, \ell\}$. 
    We thus consider now an arbitrary ranking $\ssucc_d\in \{\ssucc_0,\dots, \ssucc_t\}$ and such a profile $R^i$. First, we note that $\swap(\ssucc_j,\ssucc_d)^2=(d-j)^2$ for every ranking $\ssucc_j$ in our swap sequence. Hence, we compute that    
    \begin{align*}
    	&(m-1)!C_{\sqK}(R^{-\ssucc_i}, \ssucc_d)=c-\swap(\ssucc_i, \ssucc_d)^2=c-(d-i)^2\\
    	&(m-1)!C_{\sqK}(R^{-\ssucc_{t-i}}, \ssucc_d)=c-\swap(\ssucc_{t-i}, \ssucc_d)^2=c-(d-(t-i))^2\\
    	&C_{\sqK}(R^{\ssucc_{i+1}}, \ssucc_d)=\swap(\ssucc_{i+1}, \ssucc_d)=(d-(i+1))^2\\
    	&C_{\sqK}(R^{\ssucc_{t-(i+1)}}, \ssucc_d)=\swap(\ssucc_{t-(i+1)}, \ssucc_d)=(d-(t-(i+1)))^2
    \end{align*}
    
    Our central observation is now that 
        \begin{align*}
            &-(d-i)^2-(d-(t-i))^2+(d-(i+1))^2+(d-(t-(i+1)))^2\\
            &=-(d-i)^2-(d-(t-i))^2+(d-i-1)^2+(d-(t-i)+1)^2\\
            &=-(d-i)^2-(d-(t-i))^2+(d-i)^2-2(d-i)+1+(d-(t-i))^2+2(d-(t-i))+1\\
            &=2+4i-2t.
        \end{align*}
        
     Hence, $C_{\sqK}(R^i, \ssucc^d)=\frac{1}{Z^i}(2c+2+4i-2t)$ for all $\ssucc_d\in\{\ssucc_0, \dots, \ssucc_t\}$ and all $i\in \{0,\dots, \ell\}$, so we can conclude that $\sqK(R^*)=\{\ssucc_0, \dots, \ssucc_t\}$. 
    
    As the last point, we need to extend our argument to swap sequence of length $t'+1<t+1=\mchoosetwo+1$. Hence, consider a single-crossing swap sequence $\ssucc_0,\dots,\ssucc_{t'}$ for $t'<t$. First, if $t'=0$, then there is clearly a profile $R$ such that $\sqK(R)=\{\ssucc_0\}$ due to 2RP. We hence suppose that $t'\geq 1$. We can extend this sequence to a single-crossing swap sequence $\ssucc_0',\dots, \ssucc'_t$ of length $t+1$ using \Cref{lem:swap_sequence}. This means that $\ssucc_i'=\ssucc_i$ for all $i\in \{0,\dots,t'\}$ and the remaining rankings form a path to $\ssucc'_t$ (which is completely inverse to $\ssucc_0$).
Now, let $\bar \ssucc$ denote the completely inverse ranking of every $\ssucc$ (i.e., $\ssucc'_t=\bar \ssucc_0$). It is easy to see that 
    \[\bar \ssucc_{t-i}',\bar \ssucc_{t-i+1}', \dots,\bar \ssucc_{t-1}', \bar \ssucc_{t}' = \ssucc'_0, \ssucc'_1, \dots, \ssucc'_{i}\]
    is a single-crossing swap sequence for every $i$. Moreover, our construction yields for every $i$ a profile $\hat R^i$ such that the Squared Kemeny rule chooses exactly the given swap sequence. Now, consider the profile $\hat R^0$ (for which $\sqK(\hat R^0)=\{\ssucc_0',\dots,\ssucc_t'\}$) and the profile $\hat R^{t-t'}$ (for which $\sqK(\hat R^{t-t'})=\{\bar \ssucc_{t-t'}', \dots \bar \ssucc_t', \ssucc_1', \dots, \ssucc_{t'}'\}$). It can be verified that $\sqK(\hat R^0)\cap \sqK(\hat R^{t-t'})=\{\ssucc_0',\dots, \ssucc_{t'}'\}$, so $\sqK(\frac{1}{2} \hat R^0+\frac{1}{2} \hat R^{t-t'})=\{\ssucc_0',\dots, \ssucc_{t'}'\}$ as the Squared Kemeny rule is reinforcing.
\end{proof}

Based on \Cref{lem:Kemenysquared}, we next show that the vectors $u^{\ssucc_0,\ssucc_1},\dots,u^{\ssucc_{t-1},\ssucc_t}$ are linearly independent for every swap sequence $\ssucc_0,\dots,\ssucc_t$.

\begin{lemma}\label{lem:linearIndependence}
    Consider a single-crossing swap sequence $\ssucc_0,\dots,\ssucc_t$. The vectors $u^{\ssucc_0,\ssucc_1},\dots,u^{\ssucc_{t-1},\ssucc_t}$ are linearly independent.
\end{lemma}
\begin{proof}
    Let $\ssucc_0,\dots,\ssucc_t$ denote a single-crossing swap sequence and let $u^{\ssucc_0,\ssucc_1},\dots,u^{\ssucc_{t-1},\ssucc_t}$ denote the vectors that separate $\bar R_{\ssucc_i}$ from $\bar R_{\ssucc_{i+1}}$. Since any two consecutive rankings in our sequence only differ in a swap, we know that 
    $u^{\ssucc_i,\ssucc_{i+1}}_k=-\swap(\ssucc_i, b(k))^2+\swap(\ssucc_{i+1}, b(k))^2$
    for all $i\in \{0,\dots, t-1\}$ and $k\in \{1,\dots, m!\}$ (see \Cref{lem:swap_hyperplanes}). In particular, this means that $v(R) u^{\ssucc_i,\ssucc_{i+1}}=0$ for every profile $R$ in which $\ssucc_i$ and $\ssucc_{i+1}$ have the same Squared Kemeny score. Now, by \Cref{lem:Kemenysquared}, there are profiles $R^j$ such that $\sqK(R^j)=\{\ssucc_0,\dots,\ssucc_j\}$ for all $j\in \{1,\dots,t\}$. By the definition of the Squared Kemeny rule, this means that $v(R^j) u^{\ssucc_i,\ssucc_{i+1}}=0$ for all $i<j$ and $v(R^j) u^{\ssucc_j, \ssucc_{j+1}}>0$.

    We will now use these profiles to inductively show that the considered vectors are linearly independent. For the induction basis, let $j=1$ and note that the set $\{u^{\ssucc_0,\ssucc_1}\}$ is trivially linearly independent. Now, assume that the set $\{u^{\ssucc_0,\ssucc_1}, \dots, u^{\ssucc_{j-1},\ssucc_j}\}$ is linearly independent for some $j\leq t-1$. For the set $\{u^{\ssucc_0,\ssucc_1}, \dots, u^{\ssucc_{j},\ssucc_{j+1}}\}$, the linear independence follows by considering the vector $v(R^j)$ because $v(R^j) u^{\ssucc_i,\ssucc_{i+1}}=0$ for all $i<j$ and $v(R^j) u^{\ssucc_j,\ssucc_{j+1}}>0$.
    This is only possible if $u^{\ssucc_j,\ssucc_{j+1}}$ is linearly independent of the remaining vectors in our set. Moreover, since the set $\{u^{\ssucc_0,\ssucc_1}, \dots, u^{\ssucc_{j-1},\ssucc_j}\}$ is linearly independent by the induction hypothesis, the full set $\{u^{\ssucc_0,\ssucc_1}, \dots, u^{\ssucc_{j},\ssucc_{j+1}}\}$ is linearly independent and the lemma follows. 
\end{proof}

As the next step, we show that \Cref{lem:Kemenysquared} also holds for our SPF $f$: for every single-crossing swap sequence $\ssucc_0,\dots,\ssucc_t$, there is a profile $R$ such that $f(R)=\{\ssucc_0,\dots,\ssucc_t\}$. As it will turn out, the same profiles as for the Squared Kemeny rule show this claim. 

\begin{lemma}\label{lem:SCNI}
    Let $\ssucc_0, \dots, \ssucc_t$ denote a single-crossing swap sequence. There is a profile $R$ such that $f(R)=\{\ssucc_0,\dots,\ssucc_t\}$.
\end{lemma}
\begin{proof}
    We will prove the lemma only for single-crossing swap sequences $\ssucc_0,\dots,\ssucc_t$ with $t=\mchoosetwo$; since $f$ satisfies 2RP and reinforcement, we can shorten the sequence as demonstrated in \Cref{lem:Kemenysquared}. Hence, consider such a sequence, and let $R^*$ denote the corresponding profile defined in \Cref{lem:Kemenysquared}. 
    
    We will first show that $f(R^*)\subseteq \{\ssucc_0,\dots, \ssucc_t\}$. To this end, consider an arbitrary ranking $\ssucc$ not in our sequence. Moreover, we define $d=\swap(\ssucc,\ssucc_0)$ and note that $0<d<\mchoosetwo$ as $d\not\in \{\ssucc_0,\ssucc_t\}$. Now, let $\ssucc_d$ denote the $d$-th ranking in our sequence. We will show that $v u^{\ssucc_d,\ssucc}>0$ for the vector $v=v(R^*)$. This shows that $v\not\in\bar R_{\ssucc}$ which, in turn, implies that $\ssucc\not\in \hat g(v)=f(R^*)$ due to \Cref{lem:continuity}. 
    
    To this end, we first note that 
    $\swap(\ssucc_0,\ssucc_d)=d=\swap(\ssucc_0, \ssucc)$ and  $\swap(\ssucc_t,\ssucc_d)=t-d=\swap(\ssucc_t,\ssucc)$.
    Hence, it holds by Claim 1) of \Cref{lem:2RP_general} that  $u^{\ssucc_d,\ssucc}_k=u^{\ssucc_d,\ssucc}_{k'}=0$ for the indices $k,k'$ with $b(k)=\ssucc_0$ and $b(k')=\ssucc_t$. Next, let $v'$ denote the vector with $v'_i=\frac{1}{m!-1}$ for all $i\in \{1,\dots, m!\}$ and note that $v'\in\bar R_{\ssucc'}$ for all $\ssucc'\in\mathcal{R}$ due to the symmetry of this vector. Hence, 
    $v'u^{\ssucc_d,\ssucc}=0$. We can now compute that $v(R^{-\ssucc_0}) u^{\ssucc_d,\ssucc}=v'u^{\ssucc_d, \ssucc}-\frac{1}{m!-1}u^{\ssucc_d, \ssucc}_k=0$ and $v(R^{-\ssucc_t}) u^{\ssucc_d,\ssucc}=v'u^{\ssucc_d, \ssucc}-\frac{1}{m!-1}u^{\ssucc_d, \ssucc}_{k'}=0$.
    
    Moreover, an analogous analysis as in the proof of \Cref{lem:Kemenysquared} shows that $\swap(\ssucc_i,\ssucc_d)\leq \swap(\ssucc_i,\ssucc)$
    for every ranking $\ssucc_i$ in our swap sequence. By Claims 1) and 2) in \Cref{lem:2RP_general}, this means that $u^{\ssucc_d,\ssucc}_k\geq 0$ for the index $k$ with $b(k)=\ssucc_i$. Finally, $R^*(\ssucc_d)>0$ and we know by Claim 3) in \Cref{lem:2RP_general} that $u^{\ssucc_d,\ssucc}_k>0$ for the corresponding index. Since $v(R^{\ssucc_i}) u^{\ssucc_d,\ssucc}=u^{\ssucc_d,\ssucc_k}_k$ for all $\ssucc_i$ in our swap sequence and the corresponding index $k=b(\ssucc_i)$, we can now infer that $v u^{\ssucc_d,\ssucc}>0$ as $R^*$ is a convex combination of $R^{-\ssucc_0}$, $R^{-\ssucc_t}$, and all $R^{\ssucc_i}$ for $i\in \{1,\dots, t-1\}$. This proves that $f(R^{\ssucc_0,\dots, \ssucc_t})\subseteq \{\ssucc_0,\dots, \ssucc_t\}$.

    Next, we need to show that $f(R^*)= \{\ssucc_0,\dots, \ssucc_t\}$. Assume for a contradiction that this is not the case, i.e., there is a ranking $\ssucc_i$ in our sequence that is not chosen. By \Cref{lem:continuity}, this means that $v\not\in \bar R_{\ssucc_i}$. Moreover, by combining \Cref{lem:representation,lem:continuity}, we know that, for every $\ssucc_j\not\in f(R^*)$, there is another ranking $\phi(\ssucc_j)$ such that $v u^{\ssucc_j,\phi(\ssucc_j)}<0$. Next, we note that $\sqK(R^*)=\{\ssucc_0,\dots, \ssucc_t\}$ because of \Cref{lem:Kemenysquared}. Combined with \Cref{lem:swap_hyperplanes}, this implies that $vu^{\ssucc_j,\ssucc_{j+1}}=0$ for all $j\in \{0,\dots, t-1\}$. Finally, all these vectors $u^{\ssucc_j,\ssucc_{j+1}}$ are linearly independent of each other (see \Cref{lem:linearIndependence}). Hence, the matrix $M$ that contains these vectors as rows has full (row) rank, which equivalently means that its image has full dimension. As a consequence, there is a vector $v'\in\mathbb{R}^{m!}$ such that 
    \begin{align*}
    	v' u^{\ssucc_{j},\ssucc_{j+1}} &<0 \qquad \text{for all $j\in \{0,\dots, i-1\}$, and} \\
    	v' u^{\ssucc_{j},\ssucc_{j+1}} &>0 \qquad \text{for all $j\in \{i,\dots, t-1\}$.}
    \end{align*}
    
    Finally, we consider the vector $v+\epsilon v'$, where $\epsilon>0$ is so small that $(v+\epsilon v')u^{\ssucc_j,\phi(\ssucc_j)}<0$ still holds for all $\ssucc_j\not\in f(R^*)$. This shows that $v+\epsilon v'\not\in \bar R_{\ssucc_j}$ for any $\ssucc_j\not\in f(R^*)$. Next, it holds that 
    \begin{align*}
    (v+\epsilon v')u^{\ssucc_j,\ssucc_{j+1}}&=\epsilon v'u^{\ssucc_j,\ssucc_{j+1}}<0\qquad \text{for all $j\in \{0,\dots, i-1\}$, and}\\
    (v+\epsilon v')u^{\ssucc_j,\ssucc_{j+1}}&=\epsilon v'u^{\ssucc_{j},\ssucc_{j+1}}>0\qquad \text{for all $j\in \{i,\dots, t-1\}$.}
    \end{align*}
    Hence, we also have that $v+\epsilon v'\not\in \bar R_{\ssucc_j}$ for all $\ssucc_j\in \{\ssucc_0,\dots, \ssucc_t\}\setminus \{\ssucc_i\}$. However, since $\ssucc_i\not\in f(R^*)$ by assumption, this means that $v+\epsilon v'\not\in \bar R_{\ssucc_j}$ for every $\ssucc_j\in \mathcal{R}$. This contradicts that $\bigcup_{\ssucc_j\in\mathcal{R}}\bar R_{\ssucc_j}=\mathbb{R}^{m!}$, so our initial assumption that $f(R^*)\subsetneq \{\ssucc_0,\dots, \ssucc_t\}$ must have been wrong and $f(R^*)=\{\ssucc_1,\dots,\ssucc_t\}$.
\end{proof}

We are now ready to fully generalize \Cref{lem:swap_hyperplanes} to all vectors $u^{\ssucc_i,\ssucc_j}$. We note that the subsequent lemma is the equivalent of Step 3 in the proof of \Cref{thm:welfarist}.

\begin{lemma}\label{lem:linearDependence}
    Consider a single-crossing swap sequence $\ssucc_0,\dots, \ssucc_t$ for some $t\geq 1$. There is $\lambda>0$ such that  $u^{\ssucc_0,\ssucc_t}_k=\lambda (-\swap(\ssucc_0, b(k))^2+\swap(\ssucc_k, b(k))^2)$ for all $k\in \{1,\dots, m!\}$.
\end{lemma}
\begin{proof}
    First, we note that the lemma follows immediately from \Cref{lem:swap_hyperplanes} if $t=1$, so we focus subsequently on the case that $t\geq 2$. We denote by $\ssucc_0,\dots, \ssucc_t$ a given single-crossing swap sequence and prove the lemma in multiple steps. In particular, we first show the vector $u^{\ssucc_0,\ssucc_t}$ is linearly dependent on the vectors $u^{\ssucc_0,\ssucc_1},\dots, u^{\ssucc_{t-1},\ssucc_t}$, which means that there are scalars $\lambda_{i}$, not all of which are $0$, such that $u^{\ssucc_0,\ssucc_t}=\sum_{i=0}^{t-1} \lambda_i u^{\ssucc_i, \ssucc_{i+1}}$. The lemma now follows by showing that all scalars $\lambda_i$ are non-negative and equal since not all of them are $0$. We thus prove in the second step all $\lambda_i$ are non-negative. In the third step, we then prove the lemma for the case that $t=2$ and finally generalize the lemma to arbitrary $t$ in the last step.\medskip

    \textbf{Step 1:} As first step, we show that $u^{\ssucc_0,\ssucc_t}$ is linearly dependent on $u^{\ssucc_0,\ssucc_1}, \dots, u^{\ssucc_{t-1}, \ssucc_t}$. Assume for contradiction that this is not the case, which means that the set $\{u^{\ssucc_0,\ssucc_1}, \dots, u^{\ssucc_{t-1}, \ssucc_t},u^{\ssucc_0,\ssucc_t}\}$ is linearly independent. We consider now the matrix $M$ that contains all these vectors as rows. By basic linear algebra, this matrix has full (row) rank, so its image has full dimension. This implies that there is a vector $v'$ such that $v'u^{\ssucc_0,\ssucc_t}<0$ and $vu^{\ssucc_i,\ssucc_{i+1}}>0$ for all $i\in \{0,\dots, t-1\}$. By the definition of these vectors, this means that $v'\not\in\bar R_{\ssucc_i}$ for any $\ssucc_i\in \{\ssucc_0,\dots, \ssucc_t\}$. 
    
    Moreover, by \Cref{lem:SCNI}, there is a profile $R$ such that $f(R)=\{\ssucc_0,\dots,\ssucc_t\}$. Next, by \Cref{lem:continuity}, it follows for $v=v(R)$ that $v\in\bar R_{\ssucc_i}$ if and only if $\ssucc_i\in \{\ssucc_0,\dots,\ssucc_t\}$. By \Cref{lem:representation}, this means that there is a mapping $\phi$ from $\mathcal{R}\setminus f(R)$ to $\mathcal{R}$ such that $vu^{\ssucc,\phi(\ssucc)}<0$ for all $\ssucc\in\mathcal{R}\setminus f(R)$. Moreover, it holds that $v\in \bar R_{\ssucc_i}$ for all $\ssucc_i\in f(R)$. By the definition of the vectors $u^{\ssucc_i,\ssucc_j}$, it hence follows that $vu^{\ssucc_i,\ssucc_j}=0$ for all $\ssucc_i,\ssucc_j\in \{\ssucc_0,\dots,\ssucc_t\}$.
    
    Finally, we can find as sufficiently small $\epsilon>0$ such that $(v+\epsilon v')u^{\ssucc,\phi(\ssucc)}<0$ still holds for every $\ssucc\in\mathcal{R}\setminus f(R)$. It is also straightforward to verify that $(v+\epsilon v')u^{\ssucc_i,\ssucc_j}=\epsilon v' u^{\ssucc_i,\ssucc_j}$ for all $\ssucc_i,\ssucc_j\in \{\ssucc_0,\dots, \ssucc_t\}$. By the definition of $v'$, this means that $(v+\epsilon v')\not\in\bar R_{\ssucc_i}$ for any $\ssucc_i\in\mathcal{R}$, which contradicts that $\bigcup_{\ssucc_i\in \mathcal{R}} \bar R_{\ssucc_i}=\mathbb{R}^{m!}$. Hence, the initial assumption is wrong and $u^{\ssucc_0,\ssucc_t}$ is linearly depending on $u^{\ssucc_0,\ssucc_1},\dots, u^{\ssucc_{t-1},\ssucc_t}$.\medskip 

    \textbf{Step 2:}
We next show that all $\lambda_i$ are non-negative. Assume for contradiction that this is not true, i.e., there is an index $i$ with $\lambda_i<0$. Now, recall that the vectors $u^{\ssucc_i,\ssucc_{i+1}}$ for $i\in \{0,\dots, t-1\}$ are linearly independent, so there is a vector $v'$ such that \[
    v' u^{\ssucc_j,\ssucc_{j+1}} =
    \begin{cases}
    	\epsilon & \text{for all } j \in \{0,\dots, t-1\} \setminus \{ i \}, \\
    	1 & \text{for } j = i.
    \end{cases}
    \]
   In particular, we can choose $\epsilon>0$ so small that 
    \[\textstyle -\lambda_i v' u^{\ssucc_i,\ssucc_{i+1}}>\sum_{j=0, j\neq i}^{t-1} \lambda_j v'u^{\ssucc_j,\ssucc_{j+1}}.\]
    This means that $vu^{\ssucc_0,\ssucc_t}<0$. However, it follows now that $v'\not\in \bar R_{\ssucc_i}$ for all $\ssucc_i\in \{\ssucc_0,\dots, \ssucc_t\}$. Analogous to the last step, we can combine $v'$ again with the vector $v\in\mathbb{Q}^{m!}$ with $\hat g(v)=\{\ssucc_0, \dots,\ssucc_t\}$ to derive a vector $v+\epsilon' v'$ such that $v+\epsilon v'\not\in\bar R_{\ssucc_i}$ for all $\ssucc_i\in\mathcal{R}$. This gives the same contradiction as in the last step, so $\lambda_i\geq 0$ for all $i\in \{0,\dots, t-1\}$.\medskip

    \textbf{Step 3:} In our third step, we prove the lemma for the case that $t=2$ and hence let $\ssucc_0,\ssucc_1,\ssucc_2$ denote the considered sequence. By the first two steps, we know that $u^{\ssucc_0,\ssucc_2}=\lambda_0u^{\ssucc_0,\ssucc_1}+\lambda_1u^{\ssucc_1,\succ_2}$ for some values $\lambda_0,\lambda_1$ such that both are non-negative and at least one is strictly positive. We hence only need to show that $\lambda_0=\lambda_1$. To this end, we note that $\ssucc_0$ differs from $\ssucc_2$ either in two disjoint swaps, or we shift an alternative $x$ by two positions. 
    
    We continue with a case distinction with respect to these two options and first assume that $\ssucc_0$ differs from $\ssucc_2$ in two disjoint swaps. In this case, let $\succ_1$ and $\succ_2$ denote the rankings where we have swapped precisely one of these two pairs. In more detail, we assume subsequently that 
    \begin{align*}
        \ssucc_0&=\dots,a,b,\dots,c,d,\dots\\
        \ssucc_2&=\dots,b,a,\dots,d,c,\dots\\
        {\succ_1}&=\dots,b,a,\dots,c,d\dots\\
        {\succ_2}&=\dots,a,b,\dots,d,c,\dots
    \end{align*} 
    for some alternatives $a,b,c,d$. It is easy to check that $\swap(\ssucc_0,{\succ_1})=\swap(\ssucc_0,{\succ_2})=1$, $\swap(\ssucc_2,{\succ_1})=\swap(\ssucc_2,{\succ_2})=1$, and either $\swap(\ssucc_1,{\succ_1})=0$ and $\swap(\ssucc_1,{\succ_2})=2$, or $\swap(\ssucc_1,{\succ_1})=2$ and $\swap(\ssucc_1,{\succ_2})=0$. Next, it follows from 2RP that $f(R)=\{\ssucc_0,\ssucc_2\}$ for the profile $R$ with  $R(\succ_1)=\frac{1}{2}$ and $R(\succ_2)=\frac{1}{2}$. Moreover, using \Cref{lem:swap_hyperplanes}, we derive that
    \begin{align*}
    	v(R) u^{\ssucc_0,\ssucc_1} =& \frac{1}{2}(-\swap(\ssucc_0,\succ_1)^2 + \swap(\ssucc_1,\succ_1)^2-\swap(\ssucc_0,\succ_2)^2 + \swap(\ssucc_1,\succ_2)^2)\\
    	=& \frac{1}{2} (-2 + 4) \\
    	=&\, 1
    \end{align*}
    and
    \begin{align*}
    	v(R) u^{\ssucc_1,\ssucc_2} &= \frac{1}{2}(-\swap(\ssucc_1,\succ_1)^2 + \swap(\ssucc_2,\succ_1)^2 -\swap(\ssucc_1,\succ_2)^2 + \swap(\ssucc_2,\succ_2)^2) \\
    	&= \frac{1}{2}(-4 + 2) \\
    	&= -1.
    \end{align*}
    Finally, since $\ssucc_0,\ssucc_2\in f(R)=\hat g(v)$, we have that $v(R)u^{\ssucc_0, \ssucc_2}=\lambda_0 vu^{\ssucc_0,\ssucc_1}+\lambda_1 vu^{\ssucc_1,\ssucc_2}=\lambda_0-\lambda_1=0$. This is only true if $\lambda_0=\lambda_1$ which proves our third step in this case. 

    Next, suppose that we derive $\ssucc_2$ from $\ssucc_0$ by shifting an alternative $a$ by two positions. We assume that $\ssucc_2$ is derived from $\ssucc_0$ by shifting an alternative $a$ down as the case of shifting an alternative up is symmetric. Hence, let \begin{align*}
    	\ssucc_0 &= \dots, a, b, c, \dots \\
    	\ssucc_1 &= \dots, b, a, c, \dots \\
    	\ssucc_2 &= \dots, b, c, a, \dots.
    \end{align*} 
   Now, let ${\succ}=\dots, c,a,b,\dots$ and consider the profile $R$ such that $R(\ssucc_0)=R(\ssucc_2)=R(\succ)=\frac{1}{3}$. First, we note that $v(R)\not\in \bar R_{\ssucc}$ for any $\ssucc$ that differs from $\ssucc_0$ in any other pair as $(a,b)$, $(b,c)$, and $(a,c)$. Indeed, if such a pair exists in $\ssucc$, there is also a pair $(x',y')$ such that $x'\mathrel{\ssucc} y'$, $y'\mathrel{\ssucc_0} x'$, and $x'$ and $y'$ are adjacent in $\ssucc$. Using \Cref{lem:swap_hyperplanes}, we can then infer that $v(R) u^{\ssucc,\ssucc'}<0$ for the ranking $\ssucc'$ in which we swapped $x'$ and $y'$, which implies that $v(R)\not\in \bar R_{\ssucc}$. This is true because the rankings $\ssucc_0$, $\ssucc_2$, and $\succ$ agree on the order of all alternatives except $a,b,c$ and thus $\swap(\ssucc',\succ')<\swap(\ssucc,\succ')$ for all ${\succ'}\in\{\ssucc_0,\ssucc_2,\succ\}$. 
    
    Furthermore, it holds for the ranking $\ssucc_{1}$ that $v(R) u^{\ssucc_0,\ssucc_{1}}>0$ as we can simply compare the Squared Kemeny costs of both rankings ($\ssucc_0$ has a cost of $\frac{8}{3}$ in $R$, $\ssucc_{i_1}$ of $\frac{11}{3}$). Analogously, we can infer that $v(R) u^{\ssucc_2,\ssucc_{2}'}>0$ for the ranking $\ssucc_2'=\dots, c,b,a,\dots$, and that $v(R) u^{\succ,\succ'}>0$ for the ranking ${\succ'}=\dots, a,c,b,\dots$. Due to \Cref{lem:representation,lem:continuity} and neutrality, we can now infer that $f(R)=g(v(R))=\{\ssucc_0,\ssucc_2,\succ\}$. This means that $v(R)u^{\ssucc_0,\ssucc_2}=0$. Furthermore, we can compute that 
    \begin{align*}
        v(R) u^{\ssucc_0,\ssucc_1}=&\frac{1}{3}(-\swap(\ssucc_0,\ssucc_0)^2-\swap(\ssucc_0,\succ)^2-\swap(\ssucc_0,\ssucc_2)^2)\\
        &+\frac{1}{3}(\swap(\ssucc_1,\ssucc_0)^2+\swap(\ssucc_1,\succ)^2+\swap(\ssucc_1,\ssucc_2)^2)\\
        =&\frac{1}{3}(-0-4-4+1+1+9)=1,
    \end{align*} 
    and an analogous calculation shows that $v(R) u^{\ssucc_1,\ssucc_2}=-1$. Since \[v(R)u^{\ssucc_0,\ssucc_2}=\lambda_0 v(R)u^{\ssucc_0,\ssucc_1}+\lambda_1 v(R)u^{\ssucc_1,\ssucc_2},\] we can now infer that $\lambda_0=\lambda_1$. This completes the case that $t=2$.\medskip

    \textbf{Step 4:} Finally, we consider the case that $t>2$. First, we assume for contradiction that there is $i\in \{0,\dots, t-2\}$ such that $\lambda_i<\lambda_{i+1}$. Clearly, this means that there is $k$ such that $k(\lambda_{i+1}-\lambda_i)>\sum_{j\in \{0,\dots, t-1\}\setminus \{i+1\}} \lambda_i$. Since the vectors $u^{\ssucc_0, \ssucc_{1}},\dots, u^{\ssucc_{t-1},\ssucc_t}$ are linearly independent, we can find a vector $v'$ such that $v' u^{\ssucc_i,\ssucc_{i+1}}=k+1$, $v'u^{\ssucc_{i+1}, \ssucc_{i+2}}=-k$, and $v'u^{\ssucc_j, \ssucc_{j+1}}=1$ for all other vectors. By the linear dependence of $u^{\ssucc_0,\ssucc_t}$, we derive that \begin{align*}
    	v'u^{\ssucc_0,\ssucc_t} &= \textstyle\sum_{j\in \{0,\dots,t-1\}} \lambda_{j} v'u^{\ssucc_j,\ssucc_{j+1}} \\
    	&= -k\lambda_{i+1} + k\lambda_i + \textstyle\sum_{j\in \{0,\dots, t-1\}\setminus \{i+1\}} \lambda_{j} \\
    	&< 0.
    \end{align*} This means that $v\not\in \bar R_{\ssucc_0}$. Furthermore, it holds for all $\ssucc_j$ with $j\neq i+2$ that $v'u^{\ssucc_{j-1},\ssucc_j}>0$, which implies that $v'\not\in \bar R_{\ssucc_j}$ either. Finally, for $\ssucc_{i+2}$, we use the fact that 
    $u^{\ssucc_i,\ssucc_{i+2}}=\alpha(u^{\ssucc_i,\ssucc_{i+1}}+u^{\ssucc_{i+1},\ssucc_{i+2}})$ for some $\alpha>0$ (see Step 3) to derive that 
\[ v'u^{\ssucc_i,\ssucc_{i+2}} =\alpha(v'u^{\ssucc_i,\ssucc_{i+1}} +v'u^{\ssucc_{i+1},\ssucc_{i+2}}) =\alpha(k+1-k)>0,\]
so $v'\not\in \bar R_{\ssucc_{i+2}}$ either. Just as in Steps 2 and 3, we can now infer a contradiction by combining $v'$ with the vector $v\in\mathbb{Q}^{m!}$ which guarantees that $\hat g(v)=\{\ssucc_0,\dots, \ssucc_t\}$. In particular, there is a sufficiently small $\epsilon>0$ such that $v+\epsilon v'\not\in \bar R_{\ssucc_j}$ for all $\ssucc_j\in\mathcal{R}$. This contradicts that $\bigcup_{\ssucc_j\in \mathcal{R}} \bar R_{\ssucc_j}=\mathbb{R}^{m!}$, so we infer that hold that $\lambda_0 \geq \lambda_1\geq \dots\geq \lambda_{t-1}$. Finally, we note that the case $\lambda_i>\lambda_{i+1}$ follows symmetrically by simply changing the ``direction'' of our construction: we now choose $v'$ such that \begin{align*}
    v'u^{\ssucc_i,\ssucc_{i+1}} &= k, \\
    v'u^{\ssucc_{i+1}, \ssucc_{i+2}} &= -(k+1), \\
    v'u^{\ssucc_j,\ssucc_{j+1}} &= -1 \quad \text{for all other } j.
\end{align*} An analogous analysis as in the last case leads to a contradiction and the lemma follows. 
\end{proof}

Finally, we are ready to show our main result. 

\characterization*
\begin{proof}
    Consider an SPF $f$ that satisfies all our requirements. By \Cref{lem:domainextension}, there is a function $\hat g:\mathbb{Q}^{m!}\rightarrow 2^{\mathcal{R}}\setminus\{\emptyset\}$ such that $f(R)=g(v(R))$ for all profiles $R\in\mathcal{R}^*$. Now, let $R_{\ssucc_i}=\{v\in\mathbb{Q}^{m!}\colon \ssucc_i\in \hat g(v)\}$ for every $\ssucc_i\in\mathcal{R}$ and let $\bar R_{\ssucc_i}$ denote the closure of $R_{\ssucc_i}$ with respect to $\mathbb{R}^{m!}$. By \Cref{lem:hyperplanes}, we know that there are non-zero vectors $u^{\ssucc_i,\ssucc_j}$ for all $\ssucc_i,\ssucc_j\in\mathcal{R}$ such that $vu^{\ssucc_i,\ssucc_j}\geq 0$ if $v\in\bar R_{\ssucc_i}$ and $vu^{\ssucc_i,\ssucc_j}\leq 0$ if $v\in\bar R_{\ssucc_j}$. Moreover, by \Cref{lem:swap_hyperplanes} and \Cref{lem:linearDependence}, we also know these vectors can be represented as follows: there is $\lambda>0$ such that \[u^{\ssucc_i,\ssucc_j}_k=\lambda (-\swap(\ssucc_i,b(k))^2+\swap(\ssucc_j,b(k))^2)\] for all $k\in \{1,\dots, m!\}.$ In turn, \Cref{lem:representation} shows that $\bar R_{\ssucc_i}=\{v\in\mathbb{R}^{m!}\colon \forall \ssucc_j\in\mathcal{R}\setminus \{\ssucc_i\}\colon vu^{\ssucc_i,\ssucc_j}\geq 0\}$. 
    
    Now, let $s(v,\ssucc_i)=\sum_{k=1}^{m!} -v_k \swap(\ssucc_i,b(k))^2$ for every vector $v\in\mathbb{R}^{m!}$ and every $\ssucc_i\in\mathcal{R}$.
     Since our normal vectors are invariant under scaling, it is easy to see that \[\bar R_{\ssucc_i}=\{v\in\mathbb{R}^{m!}\colon \forall \ssucc_j\in\mathcal{R}\setminus \{\ssucc_i\}\colon  s(v,\ssucc_i)\geq s(v,\ssucc_j)\}.\] 
     Finally, \Cref{lem:continuity} shows for all $v\in \mathbb{Q}^{m!}$ that \[\hat g(v)=\{\ssucc_i\in\mathcal{R}\colon v\in\bar R_{\ssucc_i}\}=\{\ssucc_i\in\mathcal{R}\colon \forall \ssucc_j\in\mathcal{R}\setminus \{\ssucc_i\}\colon  s(v,\ssucc_i)\geq s(v,\ssucc_j)\}.\] 
     Hence, $f(R)=g(v(R))=\arg\max_{\ssucc\in\mathcal{R}} s(v(R), \ssucc)=\arg\min_{\ssucc\in\mathcal{R}} C_{\sqK}(R,\ssucc)$ for all profiles $R\in\mathcal{R}^{*}$. This proves that $f$ is the Squared Kemeny rule. 
\end{proof}

\section{Additional Material for Section \ref{sec:proportionality}}

\subsection{Linear Programs for Computing Worst-Case Profiles}
\label{app:proportionality-linear-programs}

Here, we describe on a high level some linear programs (LPs) that we used to compute the curves shown in plots in \ref{sec:proportionality}. Throughout this section, we consider a fixed number $m$ of alternatives, where $m \in \{4,5,6\}$ tend to lead to feasible programs (and sometimes also $m = 7$).

\subsubsection*{Maximum distance between Squared Kemeny and an input ranking}

In \Cref{fig:intro-alpha-curve}, we show for each $\alpha \in [0,1]$ the maximum swap distance between the Squared Kemeny ranking and an input ranking of weight $\alpha$. The way we compute this is to fix (wlog) some focal input ranking ${\succ^*} \in \mathcal{R}$, and iterate through all possible output rankings ${\rhd^*} \in \mathcal{R}$ and use an LP to ask for the maximum $\alpha$ such that there exists a profile $R$ where $R(\succ^*) = \alpha$ and ${\rhd^*}\in \sqK(R)$, which can be done using the following program (where $\weight_{\succ}$ is a variable encoding the weight $R(\succ)$ of the ranking $\succ$):
\begin{alignat*}{2}
	\text{maximize}\quad & \alpha \\
	\text{subject to}\quad 
	& \weight_{\succ^*} = \alpha \\
	& \textstyle \sum_{{\succ} \in \mathcal{R}} \weight_{\succ} = 1  \\
	& \textstyle \sum_{{\succ} \in \mathcal{R}} \weight_{\succ} \cdot \swap(\succ, \rhd)^2
	\ge \sum_{{\succ} \in \mathcal{R}} \weight_{\succ} \cdot \swap(\succ, \rhd^*)^2
	&& \quad  \text{for all } {\rhd} \in \mathcal{R} \\
	& \weight_{\succ} \in [0,1] && \quad \text{for all } {\succ} \in \mathcal{R} \\
	& \alpha \ge 0
\end{alignat*}

\subsubsection*{Maximum average distance between Squared Kemeny and a group of voters}

In \Cref{fig:noncohesive-group-bound}, we show for each $\alpha$ what is the worst average distance of a group $S$ of voters of size $\alpha$ to the Squared Kemeny ranking. To compute this, we fix (wlog) some output ranking ${\rhd^*} \in \mathcal{R}$. We then iterate through values $q \in \{0, \frac{1}{200}, \frac{2}{200}, \dots, \frac{200}{200}\}$ and for each $q$, we maximize $\alpha$ such that there exists a profile $R$ and a subprofile $S$ of $R$ with
\begin{itemize}
	\item $\rhd^*$ is a Squared Kemeny ranking for $R$,
	\item $S$ has size $\alpha$,
	\item the average distance between $S$ and $\rhd^*$ is at least $q\cdot \binom{m}{2}$.
\end{itemize}

This can be phrased as the following linear program:
\begin{alignat*}{2}
	\text{maximize}\quad & \alpha \\
	\text{subject to}\quad 
	& \text{group}_{\succ} \le \weight_{\succ} && \quad \text{for all } {\succ} \in \mathcal{R} \\
	& \textstyle \sum_{{\succ} \in \mathcal{R}} \text{group}_{\succ} \ge \alpha \\
	& \textstyle \sum_{{\succ} \in \mathcal{R}} \text{group}_{\succ} \cdot \swap(\succ, \rhd^*) \ge (q \cdot \binom{m}{2}) \cdot \alpha \\
	& \textstyle \sum_{{\succ} \in \mathcal{R}} \weight_{\succ} = 1  \\
	& \textstyle \sum_{{\succ} \in \mathcal{R}} \weight_{\succ} \cdot \swap(\succ, \rhd)^2
	\ge \sum_{{\succ} \in \mathcal{R}} \weight_{\succ} \cdot \swap(\succ, \rhd^*)^2
	&& \quad  \text{for all } {\rhd} \in \mathcal{R} \\
	& \weight_{\succ} \in [0,1] && \quad \text{for all } {\succ} \in \mathcal{R} \\
	& \text{group}_{\succ} \in [0,1] && \quad \text{for all } {\succ} \in \mathcal{R} \\
	& \alpha \ge 0
\end{alignat*}

\subsubsection*{Lower bound on achievable guarantee on the average distance to a group of voters}
In \Cref{fig:noncohesive-group-bound}, we also show a lower bound. To obtain this, we iterate through values $q \in \{0, \frac{1}{200}, \frac{2}{200}, \dots, \frac{200}{200}\}$ and for each $q$, we maximize $\alpha$ such that there is a profile $R$ satisfying:
\begin{quote}
	for every output ranking $\rhd$, there exists a subprofile $S^{\rhd}$ of $R$ with size at least $\alpha$ such that the average distance between $S^{\rhd}$ and $\rhd$ is at least $q \cdot \binom{m}{2}$.
\end{quote}

This can be phrased as a (very large) linear program with $O((m!)^2)$ variables, similar to the previous program, with variables $\text{group}^{\rhd}_{\succ}$ for each $\rhd$ together with copies of the associated constraints (the top three types of constraints in the previous program).

\subsection{Profiles where Squared Kemeny Returns the Reverse Ranking}
\label{app:hump-profiles}

We mentioned in \Cref{sec:proportionality} that the Squared Kemeny rule may return the reverse ranking of an input ranking with weight $\alpha$ for $\alpha$ up to around $17\%$. It follows from \Cref{thm:prop-guarantee} that this cannot happen for $\alpha > 25\%$. (While for the Kemeny rule, it can happen up to $\alpha = 50\%$.) Here we exhibit, for $m = 4,5$ alternatives, profiles where Squared Kemeny returns the reverse ranking for maximum $\alpha$. We show these profiles to underline our claim in the main body that these profiles are difficult to interpret and have no obvious structure.

For $m = 4$, the maximum is $\alpha = 17.5\% = 7/40 = 21/120$ with the following profile with 120 voters (divide the voter counts by 120 to obtain fractional weights), where Squared Kemeny uniquely selects $d \succ c \succ b \succ a$, the reverse of the first ranking:
\begin{align*}
	21 \quad & a \succ b \succ c \succ d \\
	18 \quad & a \succ d \succ c \succ b \\
	14 \quad & b \succ d \succ c \succ a \\
	18 \quad & c \succ b \succ a \succ d \\
	14 \quad & c \succ d \succ b \succ a \\
	21 \quad & d \succ b \succ a \succ c \\
	14 \quad & d \succ c \succ a \succ b
\end{align*}

For $m = 5$, the maximum is $\alpha = 3696/21088 = 231/1318 \approx 17.527\%$ with the following profile with 21088 voters (divide the voter counts by 21088 to obtain fractional weights), where Squared Kemeny uniquely selects $e \succ d \succ c \succ b \succ a$, the reverse of the first ranking:
\begin{align*}
	3696 \quad & a \succ b \succ c \succ d \succ e \\
	2984 \quad & a \succ e \succ d \succ c \succ b \\
	435 \quad & b \succ a \succ e \succ d \succ c \\
	200 \quad & b \succ e \succ d \succ a \succ c \\
	2188 \quad & b \succ e \succ d \succ c \succ a \\
	435 \quad & c \succ b \succ a \succ e \succ d \\
	200 \quad & c \succ e \succ b \succ a \succ d \\
	2253 \quad & c \succ e \succ d \succ b \succ a \\
	2984 \quad & d \succ c \succ b \succ a \succ e \\
	1272 \quad & d \succ e \succ c \succ a \succ b \\
	2188 \quad & e \succ c \succ b \succ a \succ d \\
	2253 \quad & e \succ d \succ b \succ a \succ c
\end{align*}

These profiles are obtained via (integer) linear programs like the first one described in \Cref{app:proportionality-linear-programs}. For $m = 5$, our profile uses the smallest possible number of integral voters (and also the smallest number of different rankings (12) and smallest number of distinct weights (8)). For $m = 4$, there exists a different profile with 40 instead of 120 voters, but the profile we show is optimal with respect to the number of different rankings used (7) and is optimal with respect to the number of different weights used (4).

For $m = 6$, the optimal value of $\alpha$ is $\alpha \approx 17.0567\%$, and for $m = 7$, it is $\alpha \approx 16.6893\%$.

\section{Additional Material for Section \ref{sec:computation}}

\subsection{Proof of \Cref{thm:npcomplete}}
\label{app:np-hard-proof}

Before we prove the NP-completeness of computing the Squared Kemeny rule, we should clarify how the input profiles are to be encoded: they should be given as a list of rankings that occur in the profile, together with their weights.

\npcomplete*
\begin{proof}
	Recall a standard inequality between the 1- and 2-norms:
	\[
	\|x\|_2 \le \|x\|_1 \le \sqrt{n} \cdot \|x\|_2 \quad \text{for all } x\in \mathbb{R}^n,
	\] where $\|x\|_2 = \sqrt{x_1^2 + \dots + x_n^2}$ and $\|x\|_1 = \sum_{i=1}^n |x_i|$.
	We will use this inequality with $n = 4$:
	\begin{equation}
		\|x\|_1 \le 2 \cdot \|x\|_2 
		\label{eq:norms-r4}
	\end{equation}
	
	Write $C_{\K}(R, \rhd) = \sum_{{\succ} \in \mathcal{R}} R(\succ) \cdot \swap(\succ, \rhd)$ for the Kemeny score of the ranking $\rhd$ in profile $R$.
	It is well-known that the following problem of computing the Kemeny rule on profiles with 4 voters is NP-complete \citep{dwork2001rank,biedl2009complexity,bachmeier2019k}.
	
	\medskip
	\begin{tabular}{ll}
		Input: & Profile $R$ in which exactly 4 rankings occur with equal weight, target score $k$ \\
		Question: & Does there exist $\rhd$ with $C_{\K}(R, \rhd) \le k$?
	\end{tabular}
	\medskip
	
	By reducing from that problem, we will show that the following problem is NP-complete:
	
	\medskip
	\begin{tabular}{ll}
		Input: & Profile $R$ in which exactly 4 rankings occur with equal weight, target score $s$ \\
		Question: & Does there exist $\rhd$ with $C_{\sqK}(R, \rhd) \le s$?
	\end{tabular}
	\medskip
	
	Consider an instance $(R,k)$ of the Kemeny problem, with $R$ defined on alternative set $A$ with $|A| = m$ and with rankings $\succ_1, \succ_2, \succ_3, \succ_4$ occurring in $R$ with weights $1/4$ each. 
	We construct an instance of the Squared Kemeny problem. 
	Our alternative set is going to be $A' = \bigcup_{a\in A} \{a^{(1)}, a^{(2)}, a^{(3)}, a^{(4)}\}$, consisting of 4 copies of each alternative $a \in A$. Thus $|A| = 4m$.
	The profile $R'$ on $A'$ consists of the following 4 rankings, each with weight 1/4:\footnote{The same profile construction is used in the reduction in Theorem 6 of \citet{biedl2009complexity} (showing that egalitarian Kemeny is NP-complete).}
	\begin{align*}
	{\succ_a'} &= {\succ_1^{(1)}} \cdot {\succ_2^{(2)}} \cdot {\succ_3^{(3)}} \cdot {\succ_4^{(4)}} \\
	{\succ_b'} &= {\succ_2^{(1)}} \cdot {\succ_3^{(2)}} \cdot {\succ_4^{(3)}} \cdot {\succ_1^{(4)}} \\
	{\succ_c'} &= {\succ_3^{(1)}} \cdot {\succ_4^{(2)}} \cdot {\succ_1^{(3)}} \cdot {\succ_2^{(4)}} \\
	{\succ_d'} &= {\succ_4^{(1)}} \cdot {\succ_1^{(2)}} \cdot {\succ_2^{(3)}} \cdot {\succ_3^{(4)}}
	\end{align*}
	This notation is to be understood as follows: $\succ_i^{(j)}$ refers to the ranking induced by $\succ_i$ applied to the alternatives $\{ a^{(j)} : a \in A \}$, i.e., the $j$-th copy of $A$.
	Rankings separated by a dot are concatenated. 
	Thus, each ranking in $R'$ has in its top $m$ positions the alternatives from the first copy of $A$, in the next $m$ positions the alternatives from the second copy of $A$, and so on.
	However, the rankings in $R'$ differ in how they order each copy; for example $\succ_a$ ranks the first copy in the same way that $\succ_1$ ranks $A$.
	Our target score is going to be $s = \tfrac14 k^2$.

	We show that there exist a ranking $\rhd$ with $C_{\K}(R, \rhd) \le k$ if and only if there exists a ranking $\rhd'$ with $C_{\sqK}(R', \rhd') \le s$.
	
	$\implies:$
Let $\rhd$ be such that $C_{\K}(R, \rhd) \le k$.
Let ${\rhd'} = {\rhd^{(1)}} \cdot {\rhd^{(2)}} \cdot {\rhd^{(3)}} \cdot {\rhd^{(4)}}$, i.e., the concatenation of 4 copies of $\rhd$.
Note that we have 
	\[
		\swap(\succ_a', \rhd') = \swap(\succ_1, \rhd) + \swap(\succ_2, \rhd) + \swap(\succ_3, \rhd) + \swap(\succ_4, \rhd) = C_{\K}(R, \rhd) \le k.
	\]
	Similarly, $\swap(\succ_b', \rhd') = \swap(\succ_c', \rhd') = \swap(\succ_d', \rhd') \le k$.
	Hence 
	\[
		C_{\sqK}(R', \rhd') \le \tfrac14 (k^2 + k^2 + k^2 + k^2) = \tfrac14 k^2 = s.
	\]
	
	$\impliedby:$
	Let $\rhd'$ be such that $C_{\sqK}(R', \rhd') \le s = \tfrac14 k^2$.
	
	Consider an optimum Kemeny ranking $\rhd$ in $R$, and let $t = C_{\K}(R, \rhd)$ be its Kemeny score. We want to show that $t \le k$.
	
	It is easy to see that ${\rhd^*} = {\rhd^{(1)}} \cdot {\rhd^{(2)}} \cdot {\rhd^{(3)}} \cdot {\rhd^{(4)}}$ is an optimum Kemeny ranking in $R'$, 
	and it has Kemeny score $C_{\K}(R', \rhd') = t$ because $\swap(\succ_x, \rhd^*) = t$ for $x = a,b,c,d$. Now, we have
	\begin{align*}
		t = C_{\K}(R', \rhd^*) \overset{\text{opt}}{\le} C_{\K}(R', \rhd') \overset{\eqref{eq:norms-r4}}{\le} 2 \cdot \sqrt{C_{\sqK}(R', \rhd')} \le 2 \cdot \sqrt{s} = 2 \cdot \sqrt{\tfrac14 k^2} = k,
	\end{align*}
	and thus $t \le k$, as required.
\end{proof}

\subsection{ILP Formulation}
\label{app:ilp-formulation}

Here, we present an Integer Linear Programming formulation for computing the Squared Kemeny rule.
Let $R$ be a profile with $n$ different rankings. Then the ILP formulation will contain $\binom{m}{2}$ binary variables, $2n$ continuous variables, and $O(m^3 + n)$ constraints. To transform swap distances into squared swap distances, we use the trick described by \citet{caragiannis2019unreasonable} for computing the maximum Nash welfare solution for fair allocation, and generalized by \citet{bredereck2020mixed} to other ILPs with convex or concave objective functions.

In the formulation, for every pair $a, b \in A$ of distinct alternatives, we have a binary variable $x_{a,b}$ encoding whether $a \rhd b$ in the output ranking $\rhd$. The first type of constraint encodes completeness of the binary relation $\rhd$, while the second encodes transitivity of $\rhd$. For each ranking $\succ$ appearing in the profile $R$, the formulation includes a (continuous) variable $\text{dist}_{\succ}$ which is constrained to equal $\swap(\succ, \rhd)$.\footnote{This variable can be eliminated from the formulation by replacing its value in the constraints placed on $\text{sqdist}_{\succ}$.}
There is also a continuous variable $\text{sqdist}_{\succ}$ which is constrained to be at least $\swap(\succ, \rhd)^2$. It will equal that value in the optimum solution.
\begin{alignat*}{2}
	\text{minimize}\quad & \textstyle \sum_{{\succ} \in \mathcal{R}} R(\succ) \cdot \text{sqdist}_{\succ} \\
	\text{subject to}\quad 
	& x_{a,b} + x_{b,a} = 1 && \quad \text{for all } a, b \in A, a \neq b \\
	& x_{a,b} + x_{b,c} + x_{c,a} \leq 2 && \quad \text{for all } a, b, c \in A, \text{ all distinct} \\
	& \text{dist}_{\succ} = \textstyle  \sum_{a, b \in A \: : \: a \succ b} \: x_{b,a} && \quad \text{for all } {\succ} \in \mathcal{R} \\
	& \text{sqdist}_{\succ} \geq k^2 + ((k + 1)^2 - k^2) \cdot (\text{dist}_{\succ} - k) && \quad \text{for all } {\succ} \in \mathcal{R} \text{ and all } k \in [\textstyle\binom{m}{2}]  \\
	& x_{a,b} \in \{0,1\} && \quad \text{for all } a, b \in A, a \neq b
\end{alignat*}

\subsection{4-Approximation to Squared Kemeny}
\label{app:4-approximation}

\begin{theorem}
	There is a polynomial-time 4-approximation algorithm for the Squared Kemeny rule.
\end{theorem}
\begin{proof}
	Let $R$ be a profile and let ${\rhd^*} \in \sqK(R)$ be some Squared Kemeny output ranking. Consider the expected Squared Kemeny cost of a random ranking $\rhd$ drawn according to the weights in $R$ (i.e., viewing $R$ as a probability distribution). We have
	\begin{align*}
		\mathbb{E}_{{\rhd} \sim R} [C_{\sqK}(R, \rhd)]
		&= \mathbb{E}_{{\rhd} \sim R} [\mathbb{E}_{{\succ} \sim R}[\swap(\succ, \rhd)^2]]
		\tag{definition} \\
		&= \mathbb{E}_{{\succ} \sim R}[\mathbb{E}_{{\rhd} \sim R} [\swap(\succ, \rhd)^2]]
		\tag{linearity of expectation} \\
		&\le \mathbb{E}_{{\succ} \sim R}[\mathbb{E}_{{\rhd} \sim R} [(\swap(\succ, \rhd^*) + \swap(\rhd^*, \rhd))^2]]
		\tag{triangle inequality} \\
		&\le \mathbb{E}_{{\succ} \sim R}[\mathbb{E}_{{\rhd} \sim R} [2\swap(\succ, \rhd^*)^2 + 2\swap(\rhd^*, \rhd)^2]] 
		\tag{$(a+b)^2 \le 2a^2 + 2b^2$}
		\\
		&= \mathbb{E}_{{\succ} \sim R}[\mathbb{E}_{{\rhd} \sim R} [2\swap(\succ, \rhd^*)^2]] + \mathbb{E}_{{\succ} \sim R}[\mathbb{E}_{{\rhd} \sim R} [2\swap(\rhd^*, \rhd)^2]] 
		\tag{linearity}
		\\
		&= \mathbb{E}_{{\succ} \sim R}[2\swap(\succ, \rhd^*)^2] + \mathbb{E}_{{\rhd} \sim R} [2\swap(\rhd^*, \rhd)^2] \tag{$\mathbb{E}[\text{const.}] = \text{const.}$} \\
		&= 4 \cdot C_{\sqK}(R, \rhd^*). \tag{definition}
	\end{align*}
	Thus, it follows that there exists ${\rhd} \in \supp(R)$ with
	\[
		C_{\sqK}(R, \rhd) \le 4 \cdot C_{\sqK}(R, \rhd^*).
	\]
	Thus, the algorithm that goes through all rankings $\rhd$ in $\supp(R)$ and outputs one with minimum $C_{\sqK}(R, \rhd)$ is a 4-approximation of the Squared Kemeny rule.
\end{proof}

The same proof strategy has been used to obtain a 2-approximation to the Kemeny rule \citep{ailon2008aggregating}.

\subsection{Proof of \Cref{thm:approximation}}
\label{app:approximation-proof}

\approximation*
\begin{proof}
	Fix an arbitrary profile $R$ and ranking $\rhd \in \mathcal{R}$.
	Also, let $\rhd_{\sqK} \in \sqK(R)$ be a ranking selected by the Squared Kemeny rule.
	Furthermore, let $\alpha \ge 1$
	be such that 
	\(
	\textstyle \sum_{\succ \in \mathcal{R}} R(\succ) \cdot \swap(\succ, \rhd) =
	\alpha \textstyle \sum_{\succ \in \mathcal{R}} R(\succ) \cdot \swap(\succ, \rhd_{\sqK}).
	\)
	Observe that we can always find such an $\alpha$:
	if the (weighted) average swap distance to $\rhd$ 
	is smaller than to $\rhd_{\sqK}$, we can set $\alpha = 1$,
	and otherwise we can set it as the ratio between the distances.
	
	It is sufficient to prove that
	\begin{equation*}
		C_{\sqK}(R, \rhd) \le 2 \alpha^2 \cdot C_{\sqK}(R, \rhd_{\sqK}). 
	\end{equation*}
	Indeed, if $\rhd$ is a $(\sqrt{1 + \nicefrac{\epsilon}{2}})$-approximation of Kemeny
	(and as e.g., $(1 + \epsilon^2) < \sqrt{1 + \nicefrac{\epsilon}{2}}$
	for small enough values of $\epsilon$,
	we can find such a ranking in polynomial time \citep{KMSc07a}),
	then surely $\alpha \le (\sqrt{1 + \nicefrac{\epsilon}{2}})$, thus
	$\rhd$ is a $(2 + \epsilon)$-approximation of Squared Kemeny as well.
	Observe that this inequality is equivalent to
	\(
	C_{\sqK}(R, \rhd) - \alpha^2 \cdot C_{\sqK}(R, \rhd_{\sqK}) \le \alpha^2 \cdot C_{\sqK}(R, \rhd_{\sqK}). 
	\)
	Which, using the $x^2 - y^2 = (x + y)(x - y)$ formula, we can write as
	\begin{multline}
		\label[ineq]{ineq:approximation}
		\textstyle \sum_{\succ \in \mathcal{R}} R(\succ) 
		(\swap(\succ, \rhd) + \alpha \swap(\succ,\rhd_\sqK))
		(\swap(\succ, \rhd) - \alpha \swap(\succ,\rhd_\sqK))
		\le \\
		\textstyle \sum_{\succ \in \mathcal{R}} R(\succ)
		\cdot (\alpha \swap(\succ,\rhd_\sqK))^2.
	\end{multline}
	We will prove \cref{ineq:approximation} using the following lemma.
	\begin{lemma}
		\label{lemma:approximation}
		Given real numbers
		$x_1, x_2, \dots,x_n \ge 0$,
		$y_1,y_2,\dots,y_n \ge 0$
		and $z_1,z_2,\dots,z_n$ such that
		\(
		x_1 + x_2 + \dots + x_n = 1,
		\)
		\(
		x_1 z_1 + x_2 z_2 + \dots + x_n z_n \le 0
		\)
		and $z_i \le y_j$ for every $i, j \in [n]$, it holds that
		\[
		\sum_{i=1}^n x_i y_i z_i \le \frac{1}{4} \left(\sum_{i=1}^n x_i y_i\right)^2.
		\]
	\end{lemma}
	\begin{proof}
		Without loss of generality, let us assume that
		$y_1 \ge y_2 \ge \dots \ge y_n \ge 0$.
		Then, the condition that $z_i \le y_j$,
		for every $i, j \in [n]$,
		means simply that $z_i \le y_n$, for every $i \in [n]$.
		
		We will first show that for fixed values of $x_1, \dots, x_n, y_1, \dots, y_n$
		the sum $\sum_{i=1}^n x_i y_i  z_i$ is maximized if
		$z_i = y_n$, for every $i \in [n-1]$, and $z_n = (1 - x_n) y_n$.
		To this end, take $z_1, \dots, z_n$ defined like that
		and arbitrary $z'_1, \dots, z'_n$ such that
		$x_1 z'_1 + \dots + x_n z'_n = 1$ and $z'_i \le y_n$, for every $i \in [n]$.
		We will show that $\sum_{i=1}^n x_i y_i z_i \ge \sum_{i=1}^n x_i y_i z'_i$.
		Since for every $i \in [n-1]$ it holds that
		$y_i \ge y_n$ and $z'_i \le y_n = z_i$ we get that
		\[
		\sum_{i=1}^n x_i y_i z_i - \sum_{i=1}^n x_i y_i z'_i =
		\sum_{i=1}^n x_i y_i (z_i - z'_i) \ge
		y_n \sum_{i=1}^n x_i (z_i - z'_i).
		\]
		Furthermore,
		we have that $\sum_{i=1}^n x_i z_i = 0$ and $\sum_{i=1}^n x_i 'z_i \le 0$,
		hence $\sum_{i=1}^n x_i (z_i - z'_i) \ge 0$.
		Thus, indeed
		$\sum_{i=1}^n x_i y_i z_i \ge \sum_{i=1}^n x_i y_i z'_i$.
		
		In the remainder of the proof of this lemma,
		let us show that with such values of $z_1, \dots z_n$
		maximizing the sum $\sum_{i=1}^n x_i y_i z_i$,
		the upper bound from the thesis still holds.
		Observe that for such values of $z_1, \dots, z_n$
		we have
		\[
		\sum_{i=1}^n x_i y_i z_i = 
		\sum_{i=1}^{n-1} x_i y_i y_n - (1 - x_n) y_n^2 =
		y_n \left( \sum_{i=1}^{n-1} (x_i y_i) - (1 - x_n) y_n \right)=
		y_n \left( \sum_{i=1}^n (x_i y_i) - y_n \right).
		\]
		Now, assume that the sum $\sum_{i=1}^n (x_i y_i)$ is a constant equal to $S$.
		What, given the value of $S$, would be the value of $y_n$ that
		would maximize the $\sum_{i=1}^n x_i y_i z_i$?
		The one that would maximize the term $y_n (S - y_n)$, which is a quadratic function
		with roots in $0$ and $S$.
		Hence, since it is concave,
		the maximum is obtained halfway between the roots at $y_n = \nicefrac{S}{2}$.
		Thus, we obtain that
		\(
		\sum_{i=1}^n x_i y_i z_i \le
		\nicefrac{1}{2} \cdot \sum_{i=1}^n (x_i y_i) \cdot \nicefrac{1}{2} \cdot \sum_{i=1}^n (x_i y_i),
		\)
		form which the thesis of the lemma follows.
	\end{proof}
	
	Now, let us use \Cref{lemma:approximation} in order to prove \cref{ineq:approximation},
	by showing that $x_\succ = R(\succ)$,
	$y_\succ = (\swap(\succ, \rhd) + \alpha \swap(\succ,\rhd_\sqK))$, and
	$z_\succ = (\swap(\succ, \rhd) - \alpha \swap(\succ,\rhd_\sqK))$
	would satisfy the conditions of the lemma.
	To this end, observe that indeed
	$\sum_{\succ \in \mathcal{R}} x_\succ = \sum_{\succ \in \mathcal{R}} R(\succ) = 1$.
	Moreover,
	by the definition of $\alpha$,
	\[
	\sum_{\succ \in \mathcal{R}} x_\succ z_\succ =
	\sum_{\succ \in \mathcal{R}} R(\succ) \cdot \swap(\succ, \rhd) -
	\alpha \sum_{\succ \in \mathcal{R}} R(\succ) \cdot \swap(\succ, \rhd_{\sqK}) \le 0.
	\]
	Thus, it suffices to show that for every $\succ, \succ' \in \mathcal{R}$ it holds that
	$z_\succ \le y_{\succ'}$, i.e.,
	\(
	\swap(\succ, \rhd) - \alpha \swap(\succ,\rhd_\sqK) \le
	\swap(\succ', \rhd) + \alpha \swap(\succ',\rhd_\sqK).
	\)
	For this, observe that using triangle inequality for $\swap$ distance
	two times, we get that
	\[
	\swap(\succ, \rhd) \le
	\swap(\succ', \rhd) + \swap(\succ, \succ') \le 
	\swap(\succ', \rhd) + \swap(\succ,\rhd_\sqK) + \swap(\succ',\rhd_\sqK).
	\]
	Since, as we assumed, $\alpha \ge 1$, we get the desired inequality.
	Therefore, we can indeed use \Cref{lemma:approximation} to show that
	\begin{multline}
		\label[ineq]{ineq:approximation2}
		\textstyle \sum_{\succ \in \mathcal{R}} R(\succ) 
		(\swap(\succ, \rhd) + \alpha \swap(\succ,\rhd_\sqK))
		(\swap(\succ, \rhd) - \alpha \swap(\succ,\rhd_\sqK))
		\le \\
		\textstyle \frac{1}{4} \left( \textstyle \sum_{\succ \in \mathcal{R}} R(\succ) (\swap(\succ, \rhd) + \alpha \swap(\succ,\rhd_\sqK)) \right)^2.
	\end{multline}
	Observe that from the definition of $\alpha$, we can bound the right-hand side of \cref{ineq:approximation2} by
	\[
	\frac{1}{4} \left(
	\sum_{\succ \in \mathcal{R}}
	R(\succ) (\swap(\succ, \rhd) + \alpha \swap(\succ,\rhd_\sqK))
	\right)^2 \le
	\alpha^2 \left( \sum_{\succ \in \mathcal{R}} R(\succ) \swap(\succ,\rhd_\sqK) \right)^2.
	\]
	Next, observe that $\sum_{\succ \in \mathcal{R}} R(\succ) \swap(\succ,\rhd_\sqK)$
	is a weighted average, hence from Jensen's inequality, we get that
	\[
	\frac{1}{4} \left(
	\sum_{\succ \in \mathcal{R}}
	R(\succ) (\swap(\succ, \rhd) + \alpha \swap(\succ,\rhd_\sqK))
	\right)^2 \le
	\alpha^2 \sum_{\succ \in \mathcal{R}} R(\succ) \swap(\succ,\rhd_\sqK)^2.
	\]
	Combining this with \cref{ineq:approximation2}, yields \cref{ineq:approximation}.
\end{proof}

\subsection{Distance Between Kemeny and Squared Kemeny}
\label{app:dist-kem-sqkem}

There are profiles where the outputs of the Kemeny and the Squared Kemeny rules are almost reverse to each other, namely have distance $\binom{m}{2} - 1$. Write $A = \{a_1, \dots, a_m\}$, $m \ge 4$, let $\succ_1$ be the ranking $a_1 \succ \cdots \succ a_m$, and let $\succ_2$ be the ranking $a_2 \succ a_1 \succ a_3 \succ \cdots \succ a_m$. Note that $\swap(\succ_1, \succ_2) = 1$.

Let $\epsilon > 0$. Consider the profile $R$ with
\[
R(\succ_1) = 2 + \epsilon, \:
R(\succ_2) = 0,
\text{ and }
R(\succ) = 1 
\text{ for all }
{\succ} \in \mathcal{R} \setminus \{\succ_1, \succ_2\}.
\]
Note that these weights sum up to more than $1$, but we can normalize the weights without changing the argument.

Now, let us discuss the outputs of the Kemeny and Squared Kemeny rules for such profile $R$.
To this end, consider an arbitrary ranking $\rhd \in \mathcal{R}$
and denote the sum of distances from $\rhd$ to every other ranking in $\mathcal{R}$ by
\(
	D_1 = \sum_{\succ \in \mathcal{R}} \swap(\succ, \rhd).
\)
Observe that the Kemeny rule cost of ranking $\rhd$ is equal to
\[
	C_{\text{Kemeny}}(R, \rhd) = D_1 + (1 + \epsilon)\cdot \swap(\succ_1,\rhd) - \swap(\succ_2,\rhd).
\]
Since $\succ_1$ and $\succ_2$ differ only on the ordering of the pair of alternatives, $a_1, a_2$, we get that
\begin{equation}
	\label{eq:dist-kem-sqkem-2alternatives}
	\swap(\succ_1,\rhd) - \swap(\succ_2,\rhd) =
	\begin{cases}
		1, & \mbox{if } a_2 \rhd a_1, \mbox{ and}\\
		-1, & \mbox{otherwise.}
	\end{cases}
\end{equation}
Thus, we obtain
\[
	C_{\text{Kemeny}}(R, \rhd)  =
	\begin{cases}
		D_1 + \epsilon \cdot \swap(\succ_1, \rhd) +1, & \mbox{if } a_2 \rhd a_1, \mbox{ and}\\
		D_1 + \epsilon \cdot \swap(\succ_1, \rhd) -1, & \mbox{otherwise,}
	\end{cases}
\]
which is minimized for $\rhd = {\succ_1}$.
Thus, $\text{Kemeny}(R) = \{\succ_1\}$ for every $\epsilon > 0$.

Now, let us consider the output of the Squared Kemeny rule.
To this end, let us denote
\(
	D_2 = \sum_{\succ \in \mathcal{R}} \swap(\succ, \rhd)^2
\)
and observe that
\[
	C_{\sqK}(R, \rhd) = D_2 + (1 + \epsilon)\cdot \swap(\succ_1,\rhd)^2 - \swap(\succ_2,\rhd)^2.
\]
Let us denote $d=\min(\swap(\succ_1,\rhd), \swap(\succ_2,\rhd))$ and observe that
$\swap(\succ_1,\rhd)+\swap(\succ_2,\rhd) = 2d +1$.
Hence, by \Cref{eq:dist-kem-sqkem-2alternatives}, we get
\begin{equation*}
\swap(\succ_1,\rhd)^2 - \swap(\succ_2,\rhd)^2 =
	(2d+1)(\swap(\succ_1,\rhd) - \swap(\succ_2,\rhd)) =
	\begin{cases}
		2d+1, & \mbox{if } a_2 \rhd a_1, \mbox{ and}\\
		-2d-1, & \mbox{otherwise.}
	\end{cases}
\end{equation*}
Therefore, we obtain that
\[
	C_{\sqK}(R, \rhd)   =
	\begin{cases}
		D_2 + \epsilon \cdot \swap(\succ_1, \rhd)^2 + 2d+1, & \mbox{if } a_2 \rhd a_1, \mbox{ and}\\
		D_2 + \epsilon \cdot \swap(\succ_1, \rhd)^2 - 2d-1 & \mbox{otherwise.}
	\end{cases}
\]
For $\epsilon < \binom{m}{2}^{-2}$, the term 
$\epsilon \cdot \swap(\succ_1, \rhd)^2$ will be strictly smaller than 1.
Hence, the value $C_{\sqK}(R, \rhd)$
will be minimized for a ranking $\rhd$ such that $a_1 \rhd a_2$ and the value of $d$ is maximized.
This will be the case for the ranking
$\rhd^*$ such that $a_m \rhd^* a_{m-1} \rhd^* \dots \rhd^* a_3 \rhd^* a_1 \rhd^* a_2$.
Thus, $\sqK(R) = \{\rhd^*\}$.
Since $\swap(\succ_1, \rhd^*)=\binom{m}{2} -1$,
we see that for the profile $R$ with $0 <\epsilon < \binom{m}{2}^{-2}$,
the Kemeny and Squared Kemeny rules indeed output almost reversed rankings.

\section{Additional Material for Section \ref{sec:experiments}}

\subsection{City Experiment Data}
\label{app:cities}

\Cref{tab:data:cities} presents the data we have used for the city ranking experiment in \Cref{sec:experiments:city}.
The GDP per capita data is taken from Wikipedia.\footnote{\url{https://en.wikipedia.org/wiki/List_of_cities_by_GDP}, accessed: 6 February 2024}
The air quality ranking is based on the average PM 2.5 concentration for the year 2018.
This was the year for which the data was the most complete,
however in a few cases we had to use the data from different year (as noted in the table).
The majority of the PM 2.5 concentration data comes from the World Health Organization database~\citep{WHO-2024-AirQuality}
(the only exception is the data for Cairo and Lagos that comes from an online article~\citep{Ogu-2023-AirQuality}).
Finally, the sunniness ranking is based on the average number of hours of sunshine per year.
The data for this was gathered from the Wikipedia articles about each city on 6 February 2024.

\begin{table}[tbh]
	\small
	\begin{tabular}{llll}
		\toprule
		\textbf{City} & \textbf{GDP per capita (US\$)} & \textbf{Avg. PM 2.5 conc. ($\mathbf{\mu g/m^3}$)} & \textbf{Avg. Sunshine h. per year} \\
		\midrule
		Bangkok & $12,670$ & $23.14^{2019}$ & $2,212$ \\
		Buenos Aires & $14,024$ & $10.26^{2015}$ & $2,384$ \\
		Cairo & $8,685$ & $47.40^{2022, \dagger}$ & $3,451$ \\
		Dubai & $47,557$ & $53.93$ & $3,570$ \\
		Dublin & $104,394$ & $7.89$ & $1,452$ \\
		Hong Kong & $52,431$ & $20.09$ & $1,829$ \\
		Istanbul & $14,989$ & $28.78$ & $2,181$ \\
		Johannesburg & $16,033$ & $22.69$ & $3,124$ \\
		Lagos & $3,607$ & $36.10^{2022, \dagger}$ & $1,844$ \\
		Lahore & $2,878$ & $123.88^{2019}$ & $3,034$ \\
		London & $66,108$ & $10.49$ & $1,675$ \\
		Mexico & $13,798$ & $22.00$ & $2,526$ \\
		Moscow & $29,012$ & $14.00^{2016}$ & $1,731$ \\
		Mumbai & $10,651$ & $75.45$ & $2,612$ \\
		New York City & $114,293$ & $7.65$ & $2,535$ \\
		Paris & $63,119$ & $14.01$ & $1,717$ \\
		Rio de Janeiro & $15,742$ & $11.45^{2015}$ & $2,182$ \\
		Rome & $40,535$ & $13.98$ & $2,724$ \\
		San Francisco & $157,704$ & $11.65$ & $3,062$ \\
		Seoul & $36,677$ & $22.93$ & $2,143$ \\
		Shanghai & $26,672$ & $37.66$ & $1,851$ \\
		Sydney & $73,034$ & $11.27^{2019}$ & $2,639$ \\
		Tokyo & $51,124$ & $12.91$ & $1,927$ \\
		Toronto & $69,110$ & $8.00$ & $2,066$ \\
		Zurich & $108,104$ & $12.13$ & $1,694$ \\
		\bottomrule
	\end{tabular}
	\caption{The data used to the city ranking analysis in \Cref{sec:experiments:city}.
	The year in a superscript of the values PM 2.5 concentration column
	signifies that the data used was from a year that is different from $2018$
	(since for $2018$ no data was available).
	Also, $\dagger$ signifies a different source of data.}
	\label{tab:data:cities}
\end{table}

\subsection{Worst-Case Average Distance}
\label{app:avgdist}

\Cref{fig:avgdist:app} presents the plots described in \Cref{sec:experiments:group-distance}
for the profiles sampled from models described in \Cref{sec:embeddings}.

\begin{figure}[th]
	\centering
	\begin{subfigure}[t]{0.3\linewidth}
		\includegraphics[width=\linewidth]{img/avgdist/small/avgdist_disc_8_50_100.jpg}
		\caption{Disc}
	\end{subfigure}
	\quad
	\begin{subfigure}[t]{0.3\linewidth}
		\includegraphics[width=\linewidth]{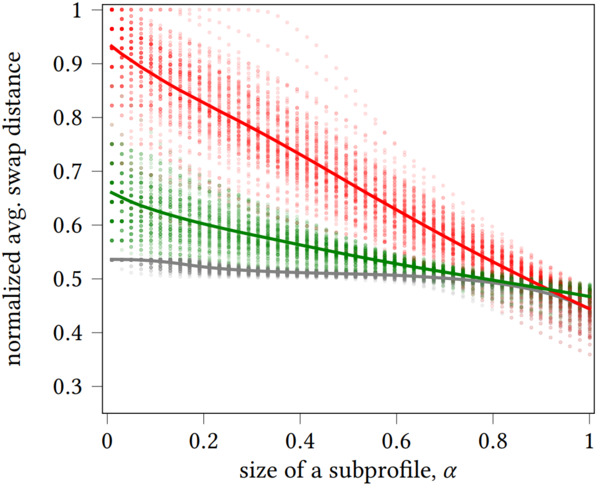}
		\caption{Circle}
	\end{subfigure}
	\quad
	\begin{subfigure}[t]{0.3\linewidth}
	\includegraphics[width=\linewidth]{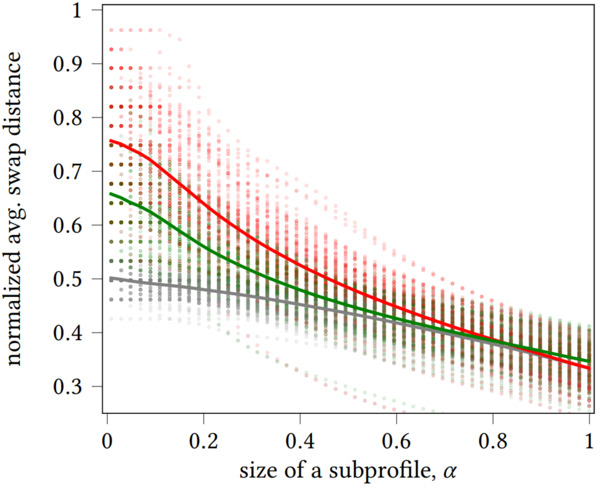}
	\caption{Countries}
	\end{subfigure}
	\\[10pt]
	\begin{subfigure}[t]{0.3\linewidth}
	\includegraphics[width=\linewidth]{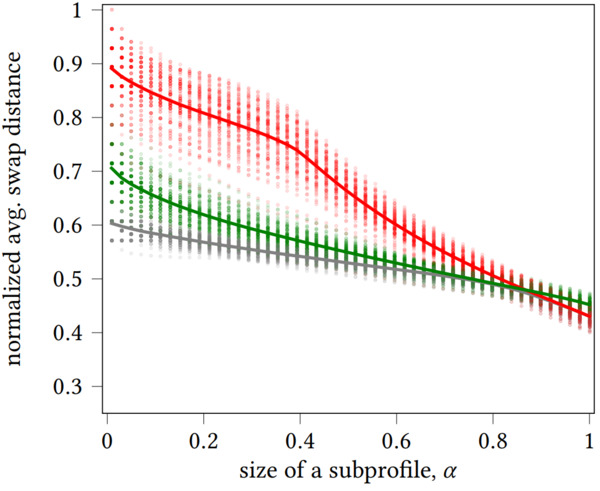}
	\caption{Mallows Mxt. {\scriptsize ($\phi_1=\phi_2=0.5$)}}
	\end{subfigure}
	\quad
	\begin{subfigure}[t]{0.3\linewidth}
	\includegraphics[width=\linewidth]{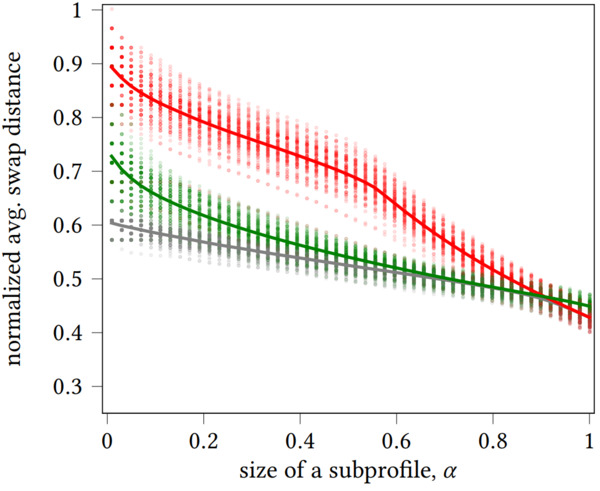}
	\caption{Mallows Mxt. {\scriptsize ($\phi_1=0.7, \phi_2=0.3$)}}
	\end{subfigure}
	\quad
	\begin{subfigure}[t]{0.3\linewidth}
	\includegraphics[width=\linewidth]{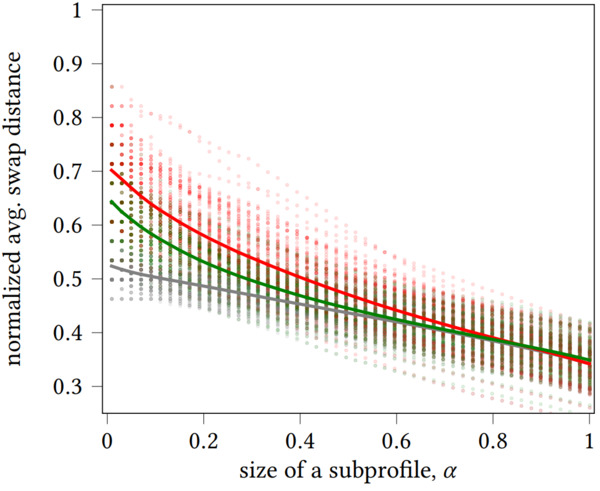}
	\caption{Breakfast}
	\end{subfigure}
	\caption{The maximal average distances between the subprofile of size $\alpha$ and the output of the Kemeny (red) and Squared Kemeny (green) rules, plus lower bound (gray).}
	\label{fig:avgdist:app}
\end{figure}

\end{document}